\documentclass[10pt]{article}

\usepackage[english]{babel}
\usepackage[letterpaper,top=2cm,bottom=2cm,left=3cm,right=3cm,marginparwidth=1.75cm]{geometry}
\usepackage[utf8]{inputenc}
\usepackage[T1]{fontenc}
\usepackage{setspace}
\usepackage{amsmath,amsthm,amssymb,amsfonts,amscd,keyval}
\usepackage{mathtools,mathrsfs}
\mathtoolsset{mathic=true}
\usepackage{bbm}
\usepackage{adjustbox}
\usepackage{tensor}
\usepackage{braket}
\usepackage{slashed}
\usepackage{tikz-cd}
\usepackage{circuitikz}
\usepackage{multirow}

\usepackage{graphicx}
\usepackage{subcaption}
\usepackage{float}
\usepackage{abstract}
\usepackage{titlesec}
\usepackage{titletoc}
\usepackage[backend=biber,  style=ieee,
  citestyle=numeric-comp, maxcitenames=3, maxbibnames=99 ]{biblatex}
\usepackage[colorlinks=true, allcolors=blue]{hyperref}
\usepackage{microtype} 
\usepackage{url}
\usepackage{tikz}
\usetikzlibrary{arrows.meta,calc}
\usepackage{authblk}

\usepackage{palatino}

\newcommand{\comments}[1]{}

\DeclareMathOperator{\im}{im}

\DeclareMathOperator{\tr}{Tr}

\newcommand{\CZ}{\mathrm{C}Z}

\newcommand{\CCZ}{\mathrm{C}\mathrm{C}Z}
\newcommand{\mCZ}{\mathrm{C}^{r-1}Z}
\DeclareMathOperator{\supp}{supp}
\DeclareMathOperator{\C}{\mathcal{C}}

\DeclareMathOperator{\F}{\mathcal{F}}
\DeclareMathOperator{\HH}{\mathsf{H}}

\DeclareMathOperator{\sd}{sd}
\DeclareMathOperator{\id}{id}

\DeclareMathOperator{\ev}{ev}
\DeclareMathOperator{\RS}{RS}

\numberwithin{equation}{section} 
\theoremstyle{definition}
\newtheorem{definition}{Definition}[section]
\newtheorem*{example}{Example}

\theoremstyle{plain}
\newtheorem{remark}{Remark}

\newtheorem{theorem}[definition]{Theorem}
\newtheorem{proposition}[definition]{Proposition}
\newtheorem{lemma}[definition]{Lemma}
\newtheorem{corollary}[definition]{Corollary}

\title{Transversal non-Clifford gates on almost-good \\ quantum LDPC and quantum locally testable codes}
\author{Yiming Li}
\author{Zimu Li}
\author{Zi-Wen Liu}
\affil{Yau Mathematical Sciences Center, Tsinghua University}
\date{}
\begin{document}

\maketitle

\vspace{-1cm}
\begin{abstract}
We exhibit nontrivial transversal logical multi-controlled-$Z$ gates on $[\![N,\Theta(N),\tilde\Theta(N)]\!]$ quantum low-density parity-check (qLDPC) codes with soundness $\tilde\Theta(1)$, combining nearly optimal code parameters with fault-tolerant non-Clifford gates on qLDPC and quantum locally testable codes for the first time.
Remarkably, our proofs proceed through highly general algebraic arguments. Building on insights from [Li et al.,~arXiv:2603.25831], we develop a general covering space framework for constructing and computing a rich family of cohomological invariant forms on sheaf codes that induce transversal logical multi-controlled-$Z$. To certify their nontriviality, we further demonstrate the existence of two-way product-expanding punctured Reed--Solomon codes, which is striking in light of the many negative examples for the product expansion behavior of ordinary Reed--Solomon codes. {This approach directly overcomes the previous obstruction to realizing nontrivial logical operations while simultaneously preserving the code parameters.} The claimed almost-good code results follow immediately as examples. 
\end{abstract}

\vspace{-4mm}
\tableofcontents


\section{Introduction and main results}\label{sec:intro}

Constructing quantum low-density parity-check (qLDPC) codes with desired properties is a deep and central problem in quantum computation that has attracted sustained effort over many years.
This pursuit carries both practical importance as a pathway to efficient fault-tolerant quantum computation~\cite{Gottesman2013,Fawzi2018PolyTime,Yamasaki2024QusiPloylog,Tamiya2024PolylogTime,Nguyen2025FT}, as well as profound theoretical relevance through substantive connections with mathematics~\cite{Freedom2002,BH2013homological,FreedmanHastings2021,Li2025Poincare,LSWLL2026Theory}, physics~\cite{Rakovszky2023PhysicsLDPCI,DeRoeck2025LDPCGappedPhases,Yin2025StablePhases}, and complexity theory~\cite{AharonovAradVidick2013_qPCPSurvey,Eldar2016NLETS,Anshu_2023}.  
A strong notion known as quantum locally testable codes (qLTCs)~\cite{AharonovEldar2015QLTC} further involves a ``soundness'' property and has important motivations for e.g.~improving the efficiency of quantum fault tolerance~\cite{Nguyen2025FT} and the quantum probabilistically checkable proofs (qPCP) conjecture~\cite{AharonovAradVidick2013_qPCPSurvey,AharonovEldar2015QLTC,Eldar2016NLETS,Dinur2024sheaf}, a central problem in quantum complexity. 
In particular, qLTC constructions with good distance and soundness were shown to imply the NLTS theorem~\cite{Eldar2016NLETS}, a major landmark towards qPCP, prior to its independent resolution via qLDPC codes~\cite{Anshu_2023}. 
Long-term work has recently culminated in breakthroughs which established ``good'' qLDPC codes with asymptotically optimal code distance and rate~\cite{PK2022Good,QuantumTanner2022,DHLV2022,losslessQLDPC} and ``almost-good'' qLTCs whose distance and soundness are polylogarithmically away from optimal~\cite{Dinur2024sheaf}.

Beyond basic code parameters, the logical operation properties are an essential consideration, since quantum computation requires efficient dynamical processing of logical quantum information. Specifically, we would like the codes to directly support fault-tolerant (e.g., transversal) physical implementations of logical gates, especially non-Clifford gates, since they dominate the cost of universal fault-tolerant quantum computation.
Therefore, the ideal goal is to combine fault-tolerant non-Clifford logical gates with desirable code parameters, although this has proven persistently difficult and its very possibility has remained in doubt, with progress so far being limited despite extensive efforts from diverse directions.
First, there are fundamental no-go theorems that rule out certain general combinations of qLDPC code and non-Clifford gate  structures~\cite{Burton2022,Fu2025nogo}. In terms of existing progress, non-Clifford gates have only been established in either
i) qLDPC codes with parameters at least polynomially below optimal~\cite{Bombin_2007,Bomb_2014,Bombin_2013,Kubica2015,RainbowCode,10.21468/SciPostPhys.14.4.065,Chen_2023,Zhu2023,Lin2024transversal,Breuckmann2024Cups,Tiew2026copycup,Zhu2025A,Zhu2025B,Barg_2026QRM,Coble_2026Coxeter,qtannercolor}, or
ii) quantum (algebraic geometry) codes with good parameters yet are non-LDPC~\cite{Nguyen2025CCZ,Wills2024magic,He2025addressable,Golowich_Lin2024,Golowich_Guruswami2025A,Golowich_Guruswami2025B,VirgileCCZ}.
Namely, while any two of the three desired properties---qLDPC, near-optimal parameters, and fault-tolerant non-Clifford gates---have been achieved, it remains unknown how, or indeed whether, all three can be realized simultaneously.
Within the existing framework of almost-good qLDPC codes and qLTCs on high-dimensional expanders and sheaves~\cite{Kaufman2014,PK2022Good,DHLV2022,QuantumTanner2022,VirgileLiftCSSHandlebody,Virgilemoderatelengthliftedquantumtanner,VirgileLiftCSS,First2024,Dinur2024sheaf,Panteleev2024}, the manifest algebraic complexity of the code space poses formidable challenges for analyzing the structure of logical gates. A recent work~\cite{LSWLL2026Theory} proposes a systematic methodology for formulating cohomology invariants, including logical representatives and cup products, and presents explicit computation of nontrivial non-Clifford operations on such sheaf codes, but the related results there involve extra assumptions from number theory.

\paragraph{Main results.}
In this work, building on the cohomological invariant forms on sheaf codes proposed in Ref.~\cite{Li2025Poincare}, we realize nontrivial transversal logical multi-controlled-$Z$ gates on almost-good qLDPC codes and qLTCs~\cite{Dinur2024sheaf}. To this end, we demonstrate the existence of \emph{two-way product-expanding punctured Reed--Solomon local codes} 
{that support nontrivial cup products in cohomology~\cite{LSWLL2026Theory}.} 
This resolves the open problem of achieving fault-tolerant non-Clifford gates simultaneously with nearly optimal parameters on qLDPC codes and qLTCs.

As is standard, $\mCZ$ denotes a (multi-)controlled-$Z$ gate on $r$ qubits with $r-1$ controls.

\begin{theorem}\label{thm:main}
For any integer $r\geq 2$, there exist  
\begin{itemize}
    \item $[\![N,\Theta(N),\Theta(N/(\operatorname{log}N)^{r-1})]\!]$ quantum LDPC codes\footnote{For $r=2$, namely $\mathrm{C}Z$ gates, the code parameters can be further made optimal, i.e., $[\![N,\Theta(N),\Theta(N)]\!]$; see~\cite{Li2025Poincare}.} 
    \item $[\![N,\Theta(N),\Theta(N/(\operatorname{log}N)^{2r-1})]\!]$ quantum locally testable codes with soundness $\Theta(1/(\operatorname{log}N)^{2r-1})$
\end{itemize}
that support nontrivial transversal logical $\mCZ$ gates.
\end{theorem}

The key insight underpinning our proof is that the base spaces of almost-good qLDPC codes and qLTCs are covering spaces of the base spaces of certain homological product (HGP) codes, with the local coefficients of the former induced from those of the latter~\cite{LSWLL2026Theory}. The nontrivial cohomological invariants on these primitive HGP codes lift to almost-good codes and remain nontrivial, enabling us to establish nontrivial gate actions. A crucial leap from previous work \cite{LSWLL2026Theory} is the use of judiciously designed two-way product-expanding local codes (see Sec.~\ref{sec:GRS}). This result is based on the notion of extendability proposed in Ref.~\cite{KP2025Extendable} and allows us to directly address the main obstacle arising from tensor products of sheaves. Our results establish Conjecture~1.2 in Ref.~\cite{Li2025Poincare}.

It is well-known that for two Reed--Solomon (RS) codes $\C_1, \C_2$ evaluating over  $\mathbb{F}_q$, $(\C_1,\C_2)$ and $(\C_1^\perp,\C_2^\perp)$ cannot be product-expanding at the same time \cite{PK2023RobustlyTestable,QuantumTanner2022}. Moreover, the collection $(\C_1,\ldots,\C_t)$ of primitive RS codes with rate $< 1/t$ is proved to have $\rho(\C_1,\ldots,\C_t) \leq 1/n$ \cite{KP2025Extendable}. Working over sufficiently large finite fields and using distinct evaluation sets, we show that punctured RS and generalized Reed--Solomon (GRS) codes bypass the above obstructions, yielding the following theorem.

\begin{theorem}\label{thm:GRS}
Fix an integer $t\ge 2$ and constants $\nu_i \in (0,1)$ for each $i \in [t]$. 
Let $n$ be a sufficiently large integer and let $m_i=\lfloor \nu_i n\rfloor$ for each $i\in[t]$.
Let $\mathbb F_q$ be a finite field with $q > 2^{t n^t}$, and sample independently and uniformly at random $t$ subsets $E_1,\dots, E_t\subseteq \mathbb F_q$,  each of size $|E_i|=n$. Then with probability at least $1 - 2^{t n^t - \log q}$, the associated punctured RS codes $\C_i =\RS_q(E_i,m_i)$, each of code dimension $m_i$, have dual codes that satisfy
	\begin{align}
		\rho(\C_1^\perp, \ldots, \C_t^\perp) > c(\nu_i,t),
	\end{align}
	where $c$  depends only on $\nu_i$ and $t$. Moreover, if $\nu_i \leq \frac{1}{2}$, then
	\begin{align}\label{eq:GRS_rho}
		\rho(\C_1^\times,\ldots,\C_t^\times) > c(\nu_i,t),
	\end{align}
	where for each $i$, $C_i^\times$ can be chosen freely from $C_i$ or its dual $C_i^\perp$. 
\end{theorem}

Explicit bounds on $\epsilon$ in terms of general $\nu_i$ and $t=2,3$ are given in \eqref{eq:epsilon_2D} and \eqref{eq:eplision_3D_condition}, respectively. The lower bound $q > 2^{t n^t}$ is due to a rough counting on certain subsets in $t$-dimensional grid $[n]^t$. For $\nu_i \leq \frac{1}{2}$, \eqref{eq:GRS_rho} is even stronger and indicates the two-way product expansion:
\begin{align}
	\rho(\C_1,\ldots,\C_t), \quad \rho(\C_1^\perp,\ldots,\C_t^\perp) > c(\nu_i,t).
\end{align}
If $\nu_i$ is further assumed to $< \frac{1}{t}$, then the $t$-fold Schur product (see below Definition \ref{def:KR_product}) of each $\C_i$ is a proper subspace in $\mathbb{F}_q^n$. We show in Section \ref{sec:proof} that this multiplication property is a flexible generalization of the multi-orthogonality conditions which is commonly used in previous constructions of $Z$-rotation and $\mCZ$ gates~\cite{Bravyi_2012,Paetznick_2013,Calderbank2020T,Lin2024transversal,Nguyen2025CCZ}. It is important for achieving nontrivial coding rates in our setting.

{
It is noted in Ref.~\cite{Golowich_Guruswami2025B} that a collection $t-1$ tensor products of random punctured RS codes and one dual code possess a weak product expansion $\rho = \Omega(1/n^\epsilon)$. We prove by explicit construction to accomplish the extendability that leads to a constant product-expanding factor \cite{KP2025Extendable} without using tensor products and locally testable codes. By \eqref{eq:GRS_rho}, the product expansion even holds for any combination of punctured RS codes or their dual GRS codes. This result is of independent interest and may find further applications.
}

\emph{Remark.} 
In v1 of this paper, the cap product used there did not satisfy the required Leibniz rule, hence the cupcap gate argument is invalid. Here, we correct this issue and retain the results by removing the erroneous cap product and turning to using well-behaved cup products. The invariance is now guaranteed but proving a nontrivial logical operation becomes more difficult. To achieve the goal, we develop two-way product-expanding punctured Reed--Solomon codes, which constitute a new contribution of this version. Used as local codes, their product expansion preserves the sheaf code parameters. We then rigorously verify that the associated tensor product sheaf admits nonzero top degree (co)homology, yielding nontrivial logical operations.

\section{Preliminaries}

\subsection{Cell complexes and Alexandrov topology}
\begin{definition}[Regular cell complex]
	A finite \emph{regular cell complex} (or finite \emph{regular CW complex}) is a topological space $X$ constructed inductively as follows:
	
	\begin{enumerate}
		\item Start with a discrete set $X^0$, whose points are regarded as $0$-cells.
		
		\item Inductively form the \emph{$n$-skeleton} $X^n$ from $X^{n-1}$ by attaching $n$-cells $e_\alpha^n$
		via maps
		\begin{equation}
			\varphi_\alpha : S^{n-1} \longrightarrow X^{n-1},
		\end{equation}
		where $S^{n-1}=\partial D^n$ is the boundary of an $n$-dimensional closed ball.
		This means that $X^n$ is the quotient of
		\begin{equation}
			X^{n-1}\;\sqcup\;\bigsqcup_\alpha D_\alpha^n
		\end{equation}
		under the identifications $x\sim \varphi_\alpha(x)$ for $x\in \partial D_\alpha^n$.
		The cell $e_\alpha^n$ is the image of the interior $D_\alpha^n-\partial D_\alpha^n$ under the quotient map, so that
		\begin{equation}
			X^n = X^{n-1} \cup \bigcup_\alpha e_\alpha^n,
		\end{equation}
		where each $e_\alpha^n$ is an open $n$-ball.
		
		\item For each cell $e_\alpha^n$, the attaching map $\varphi_\alpha$ restricts to a homeomorphism from $\partial D_\alpha^n$ onto its image in $X^{n-1}$.
		
		\item Terminate this inductive process at a finite stage, setting $X=X^n$ for some $n<\infty$.
	\end{enumerate}
\end{definition}

In particular, the closure of each cell in a regular cell complex is homeomorphic to a closed ball. Therefore, there will be no $1$-dimensional cell joined to the same vertex twice. In this paper, we consider only finite regular cell complexes and will henceforth omit  ``finite'' and ``regular'' when referring to them. 

For our coding-theoretic applications, we will not emphasize the usual point-set topological aspects of cell complexes. Instead, we view cell complexes in a more combinatorial way.

\begin{definition}[Cell poset]
	A \emph{cell poset} $P_X$ is a poset constructed from a cell complex $X$ by
    \begin{equation}
    P_X\coloneqq \{e^n_\alpha:e^n_\alpha \ \text{is an $n$-cell in} \ X\},
    \end{equation}
    and for $e^n_\alpha$ and $e^m_\beta$, we define $e^n_\alpha\leq e^m_\beta$ if and only if $e^n_\alpha$ is contained in the closure of $e^m_\beta$ in $X$. 
\end{definition}

From now on, we will slightly abuse notation and simply denote $P_X$ by $X$. We denote by $X(k)$ the set of $k$-dimensional cells in $X$, and call an element $\sigma \in X(k)$ a $k$-cell. For $\sigma,\tau \in X$ with $\sigma \leq \tau$ and $\dim(\sigma)<\dim(\tau)$, we write $\sigma<\tau$. In particular, when $\dim(\sigma)=\dim(\tau)-1$, we write $\sigma\lessdot\tau$, and say that $\sigma$ and $\tau$ are joined. For each cell $\sigma\in X$, we write $X_{\ge \sigma}$ and $X_{\le \sigma}$ for the sets of cells greater than and less than $\sigma$ in the cell poset, respectively.

In this paper, we restrict attention to cell complexes whose cellular incidence numbers are either $\pm 1$ or $0$, or equivalently, such that a $k$-cell is incident to any given $(k-1)$-cell at most once. This further implies that for any fixed $k$-cell $\sigma$ and $(k+2)$-cell $\pi$ with $\sigma<\pi$, there exist an even number of $(k+1)$-cells $\tau$ such that $\sigma \lessdot \tau \lessdot \pi$. Although all results can be established without this assumption, for simplicity, one can regard this property as part of the definition.

Although a finite cell poset is a purely combinatorial object, it admits a nontrivial topology, which will play a key role in this paper.

\begin{definition}[Alexandrov topology]
    Given a poset $X$, the \emph{Alexandrov topology} of  $X$ is the topology generated by the basis 
    \begin{equation}
       \mathcal{B} \coloneqq  \{X_{\ge \sigma} : \sigma \in X \}, 
    \end{equation}
    {i.e., a subset $U\subseteq X$ is open if and only if it is a union of sets in $\mathcal B$.}
\end{definition}

It is easy to verify the following proposition:
\begin{proposition}
    Let $X,Y$ be posets equipped with Alexandrov topology. A map $f:X\to Y$ is continuous if and only if it is an order-preserving map, i.e. whenever $\sigma, \sigma'\in X$ and $\sigma\leq \sigma'$, then $f(\sigma)\le f(\sigma')$. 
\end{proposition}

To build qLDPC codes, the following requirement is common:
\begin{definition}[Sparse cell complex] \label{def:sparse}
	We say that a family of cell complexes $\{X_n\}_{n=1}^\infty$ is \emph{sparse} if for each $k$-cell in $X_n$, it is joined to at most a uniform constant number of $(k+1)$-cells and $(k-1)$-cells.
\end{definition}

\subsection{Sheaves and their operations}\label{sec:sheaves}

Throughout the paper, $\mathbb F_q$ is a finite field of characteristic 2, and all the vector spaces are finite dimensional over $\mathbb F_q$. This is for the sake of finding high dimensional product-expanding local codes in Section \ref{sec:GRS} (see also \cite{kalachev2023ProductExpansion,Dinur2006,KP2025Extendable}). Since $q$ is a power of $2$, the final results can be pulled back to $\mathbb{F}_2$ and realized in the system of qubits \cite{PhysRevA.95.012329,He2025addressable}.

\begin{definition}[Sheaf]
	Let $X$ be a poset with Alexandrov topology. A \emph{sheaf} $\mathcal F$ on $X$ is a functor from $X$ to vector spaces over $\mathbb{F}_q$, i.e., it assigns to each element $\sigma$ a vector space $\mathcal{F}_{\sigma}$, and to each $\sigma\leq \sigma'$  a linear map $\mathcal{F}_{\sigma, \sigma'}:\mathcal{F}_{\sigma}\rightarrow \mathcal{F}_{\sigma'}$ such that for all $\sigma\leq\sigma'\leq\sigma''$, we have $\mathcal{F}_{\sigma,\sigma''}=\mathcal{F}_{\sigma',\sigma''}\circ \mathcal{F}_{\sigma,\sigma'}$. 
\end{definition}

\begin{definition}[Pullback sheaf]
    Let $X, Y$ be two posets with Alexandrov topology, and $f:X\to Y$ is an order-preserving map. Suppose $\mathcal G$ is a sheaf on $Y$, then the pullback sheaf $f^*\mathcal G$ on $X$ is a sheaf defined by, for each $\sigma\in X$,
    \begin{align}
        (f^*\mathcal G)_\sigma\coloneqq\mathcal G_{f(\sigma)},
    \end{align}
    and when $\sigma\le \sigma'\in X$, the restriction map is
    \begin{align}
        (f^*\mathcal G)_{\sigma,\sigma'}\coloneqq \mathcal G_{f(\sigma),f(\sigma')}.
    \end{align}
\end{definition}

\begin{definition}[Tensor product of sheaves]
    Let $X$ be a poset equipped with two sheaves $\mathcal F$ and $\mathcal{G}$, the tensor product of the two sheaves $\mathcal F\otimes\mathcal{G}$ is defined by, for each $\sigma\in X$,
    \begin{equation}
(\mathcal F\otimes\mathcal{G})_\sigma\coloneqq\mathcal F_\sigma\otimes\mathcal{G}_\sigma,
\end{equation}
and when $\sigma\leq \sigma'\in X$, the restriction map is
\begin{equation}
(\mathcal F\otimes\mathcal{G})_{\sigma,\sigma'}\coloneqq\mathcal F_{\sigma,\sigma'}\otimes\mathcal{G}_{\sigma,\sigma'}.
\end{equation}
\end{definition}

\begin{definition}[External tensor product of sheaves]
Let $X$ and $Y$ be posets with sheaves $\mathcal F$ and $\mathcal G$ respectively. Let $p_X:X\times Y\to X$ and $p_Y:X\times Y\to Y$ be the projection maps. The external tensor product of $\mathcal F$ and $\mathcal G$ is the sheaf on $X\times Y$ defined by
\begin{align}
    \mathcal F\boxtimes \mathcal G
:=
p_X^*\mathcal F \otimes  p_Y^*\mathcal G.
\end{align}
\end{definition}

\subsection{Chain complexes and CSS codes}

\begin{definition}[Chain complex and cochain complex]
    A \emph{chain complex} $C_\bullet$ is a sequence of vector spaces over $\mathbb{F}_q$
    \begin{equation}
    C_\bullet=\cdots{\longleftarrow} C_{i-1}\overset{\partial_{i}}{\longleftarrow} C_i\overset{\partial_{i+1}}{\longleftarrow} C_{i+1}\longleftarrow\cdots
    \end{equation}
    such that $\partial_i\circ\partial_{i+1}=0$. The corresponding \emph{cochain complex} $C^\bullet$ is defined by applying the functor $\operatorname{Hom}(-,\mathbb{F}_q)$
    \begin{equation}
    C^\bullet=\cdots{\longrightarrow} C^{i-1}\overset{\delta^{i-1}}{\longrightarrow} C^i\overset{\delta^i}{\longrightarrow} C^{i+1}\longrightarrow\cdots
    \end{equation}
    where $C^i=\operatorname{Hom}(C_i,\mathbb{F}_q)$ is the dual vector space of $C_i$ and $\delta^i=\operatorname{Hom}(\partial_{i+1},\mathbb{F}_q)$.
\end{definition}

An element in $C^i$ is a function $C_i\to\mathbb F_q$. We define the pairing
\begin{align}
    \langle-,-\rangle: C^i\times C_i \rightarrow \mathbb F_q
\end{align}
to be the evaluation pairing. Then by definition, for $\alpha\in C^i$ and $x\in C_{i+1}$, $\langle\alpha,\partial x\rangle=\langle\delta\alpha,x\rangle$. Since all vector spaces in this paper are finite-dimensional, once a basis of each $C_i$ is fixed, we identify $C^i$ with $C_i$. Under this identification, $\delta^i$ is represented by the transpose matrix $\partial_{i+1}^T$.

\begin{definition}[Homology and cohomology]
    The $i$-th \emph{homology group} of a chain complex $C_\bullet$ and the $i$-th \emph{cohomology group} of a cochain complex $C^\bullet$ are defined by
    \begin{equation}
        H_i(C_\bullet)\coloneqq \ker \partial_i / \im \partial_{i+1},
        \qquad
        H^i(C^\bullet)\coloneqq \ker \delta^i / \im \delta^{i-1}.
    \end{equation}
    Elements of $\ker \partial_i$, $\im \partial_{i+1}$, $\ker \delta^i$, and $\im \delta^{i-1}$ are called
    \emph{$i$-cycles}, \emph{$i$-boundaries}, \emph{$i$-cocycles}, and \emph{$i$-coboundaries}, respectively.
\end{definition}

\begin{definition}[Chain map] \label{def:chain_map}
	Suppose $C_\bullet$ and $D_\bullet$ are chain complexes with boundary map $\partial_C$ and $\partial_D$, respectively. We say that $f_\# = (f_i: C_i \rightarrow D_i)$ is a \emph{chain map} if $f_\# \circ \partial_C = \partial_D \circ f_\#$ at each degree of the complexes, i.e., the following diagram is commutative
	\begin{equation}
	\begin{tikzcd}[column sep=1.5cm, row sep=0.7cm]
		\cdots \arrow[r, "\partial_{C,i+1}"] & C_i \arrow[r, "\partial_{C,i}"] \arrow[d,"f_i"] & C_{i-1} \arrow[d, "f_{i-1}"] \arrow[r, "\partial_{C,i-1}"] & \cdots \\
		\cdots \arrow[r, "\partial_{D,i+1}"] & D_i \arrow[r, "\partial_{D,i}"] & D_{i-1} \arrow[r, "\partial_{D,i-1}"] & \cdots
 	\end{tikzcd}
	\end{equation}
    
    Applying the functor $\operatorname{Hom}(-,\mathbb{F}_q)$ to a chain map $f_\#$ induces a \emph{cochain map} $f^\#:D^\bullet\to C^\bullet$. These in turn induce maps on homology and cohomology, denoted by $f_*:H_\bullet(C_\bullet)\to H_\bullet(D_\bullet)$ and $f^*:H^\bullet(D^\bullet)\to H^\bullet(C^\bullet)$.
\end{definition}

\begin{definition}\label{def:CSS}
	A \emph{Calderbank--Shor--Steane (CSS) code} is a stabilizer code given by three consecutive terms of a cochain complex $C^\bullet$ over a finite field $\mathbb{F}_q$: 
	\begin{align}
		C^{i-1}\overset{\delta^{i-1}}{\longrightarrow} C^i\overset{\delta^i}{\longrightarrow} C^{i+1}
	\end{align}
	with $\delta^i = H_Z$ and $\delta^{i-1} = H_X^T$ defining the stabilizer generators. The $X$-type and $Z$-type logical operators are represented by elements from the cohomology $H^i(C^\bullet)$ and the homology $H_i(C_\bullet)$, respectively. Although all linear spaces in this paper are constructed over $\mathbb F_q$, when discussing CSS code parameters we always regard them as $\mathbb F_2$-vector spaces by restriction of scalars, following the convention of Ref.~\cite{Dinur2024sheaf}. Thus, the number of physical qubits is $n=\dim_{\mathbb F_2} C^i$. The \emph{code dimension} $k$ which represents the number of encoded logical qubits, is given by $k = \dim_{\mathbb F_2} H^i(C^\bullet)$. The \emph{code distance} is given by $d = \min\{d_X, d_Z\}$, where
	\begin{align}
		d_X = \min_{x \in \ker \delta^i \setminus \im \delta^{i-1}} \vert x \vert, \quad 
		d_Z = \min_{z \in \ker \partial_i \setminus \im \partial_{i+1}} \vert z \vert,
	\end{align}
	where $|x|$ is the Hamming weight of $x$. The \emph{soundness} $\rho$ is given by $\rho=\min\{\rho_X,\rho_Z\}$, where
    \begin{align}
        \rho_X=\min_{x\in C^i\setminus\ker\delta^i} \frac{|\delta^ix|}{\dim_{\mathbb F_2} C^{i+1}}\cdot\frac{\dim_{\mathbb F_2} C^i}{\operatorname{dist}(x,\ker\delta^i)},\quad \rho_Z=\min_{z\in C_i\setminus\ker\partial_i} \frac{|\partial_iz|}{\dim_{\mathbb F_2} C_{i-1}}\cdot\frac{\dim_{\mathbb F_2} C_i}{\operatorname{dist}(z,\ker\partial_i)}.
    \end{align}
\end{definition}

As an important example of a (co)chain complex, suppose $X$ is a cell complex equipped with a sheaf $\mathcal F$, the corresponding sheaved cellular chain complex is defined to be the direct sum of vector spaces
\begin{equation}
    C^i(X,\mathcal F)\coloneqq \bigoplus_{\sigma\in X(i)}\mathcal F_\sigma,
\end{equation}
with the coboundary map $\delta:  C^i(X, \mathcal{F}) \to C^{i+1}(X, \mathcal{F})$ given by
\begin{equation}\label{eq: sheaved coboundary map}
(\delta \alpha)(\sigma') = \sum_{\sigma \lessdot \sigma'} \mathcal{F}_{\sigma,\sigma'}(\alpha(\sigma)),
\end{equation}
where $\alpha\in C^i(X,\mathcal F)$ and $\sigma'\in X(i+1)$. The boundary map $\partial_{i+1}:C_{i+1}(X,\F)\rightarrow C_{i}(X,\F)$ as the transpose of $\delta^i$ is given by
\begin{equation}\label{eq: sheaved boundary map}
(\partial x)(\sigma'') = \sum_{\sigma\gtrdot \sigma ''} \mathcal{F}_{\sigma'',\sigma}^T(x(\sigma)),
\end{equation}
where $x\in C_{i+1}(X,\mathcal F)$ and $\sigma''\in X(i)$.

\subsection{Logical multi-controlled-$Z$ gates}

\begin{definition}[Cohomological invariant form]\label{def:invariant_poly}

Let $C^{(1)},\cdots, C^{(r)}$ be $r$ vector spaces from $r$ cochain complexes, and let $Z^{(i)}$ and $B^{(i)}$ denote the subspaces of cocycles and coboundaries in $C^{(i)}$, respectively, for each $i\in[r]$. A tensor
$T:C^{(1)}\times\cdots\times C^{(r)}\to\mathbb F_q$ is called a \emph{cohomological invariant form} if for any $x_i\in Z^{(i)}$ and $y_i\in B^{(i)}$, $i\in[r]$,
\begin{align}
    T(x_1+y_1,\dots,x_r+y_r)=T(x_1,\dots,x_r).
\end{align}
Then $T$ defines a diagonal unitary on physical qubits:
\begin{align}\label{eq:invariant_poly}
	U_T \coloneqq \sum_{z_1\in C^{(1)},\dots,z_r\in C^{(r)}}
	(-1)^{\tr_{\mathbb{F}_q/\mathbb{F}_2}(T(z_1,\dots,z_r))}
	\ket{z_1,\dots,z_r}\bra{z_1,\dots,z_r},
\end{align}
which induces a logical circuit composed of $\mathrm{C}^{r-1}Z$ gates on the CSS codes built from $C^{(1)},\dots, C^{(r)}$.
\end{definition}

If $T$ is a sparse tensor, then $U_T$ implements a constant-depth $\mathrm{C}^{r-1}Z$ circuit. As mentioned in the beginning of Section \ref{sec:sheaves}, its implementation on qubits consists of multi-qubit gates that act only across code blocks, and each physical qubit is acted on by constant times. 

\subsection{Product-expanding classical codes and extendability}

\begin{definition}[Dual tensor product code]
	Given classical codes $\C_1 \subseteq \mathbb{F}_q^{n_1},\ldots,\C_t \subseteq \mathbb{F}_q^{n_t}$, their \emph{tensor product code} is simply defined by $\C_1 \otimes \cdots \otimes \C_t$. Let 
	\begin{align}
		\C^{(i)} \coloneq \mathbb{F}_q^{n_1} \otimes \cdots \otimes \C_i \otimes \cdots \otimes \mathbb{F}_q^{n_t}.
	\end{align}
	We also consider
	\begin{align}
		\mathcal{C}_1 \boxplus \cdots \boxplus \mathcal{C}_t \coloneq \mathcal{C}^{(1)} + \cdots + \mathcal{C}^{(t)}
	\end{align}
	One can check by basic linear algebra that $\C_1 \boxplus \cdots \boxplus \C_t = ( \C_1^\perp \otimes \cdots \otimes \C_t^\perp  )^\perp$ and hence it is called \emph{dual tensor product code}.
\end{definition}

\begin{definition}[Hamming weight]
	Given a tensor $x \in \mathbb{F}_q^{n_1} \otimes \cdots \otimes \mathbb{F}_q^{n_t}$, the \emph{Hamming weight} $\vert x \vert$ of $x$ is simply defined as the number of nonzero components of $x$. We also define the \emph{normalized Hamming weight}:
	\begin{align}
		\Vert x \Vert = \frac{1}{\prod_{j \in [t]} n_j}\vert x \vert.
	\end{align}
	In addition, relative to different directions, we have
	\begin{enumerate}
		\item $\vert x \vert_i$ is defined as the number of all nonzero vectors $x(r_1,...,r_{i-1},\text{ - },r_{i+1},...,r_{l})$, or lines in the direction labeled by $i$. Obviously, $\vert x \vert_i \leq \vert x \vert \leq n_i \vert x \vert_i$ for any $x$.
		
		\item Let $\mathcal{N}_i  = \frac{1}{n_i} \prod_{j \in [l]} n_j$. Then $\Vert x \Vert_i \coloneq \frac{1}{\mathcal{N}_i} \vert x \vert_i$ and $\frac{1}{n_i} \Vert x \Vert_i \leq \Vert x \Vert \leq \Vert x \Vert_i$ for any $x$. 
	\end{enumerate}
\end{definition}

\begin{definition}[Product expansion]
	Let $(\C_i)$ be linear codes of length $n_i$. Given a real nonnegative number $\rho$, we say that $(\C_i)$ is \emph{$\rho$-product-expanding} if for all $c \in \C_1 \boxplus \cdots \boxplus \C_l$, there exists $c_1 \in \C^{(i)}$ such that $c = \sum_i c_i$ and 
	\begin{align}
		\Vert c \Vert \geq \rho \sum_{i \in [t]} \Vert c_i \Vert_i, \ \text{ 	or equivalently, } \
		\vert c \vert \geq \rho \sum_{i \in [t]} n_i \vert c_i \vert_i.
	\end{align}
\end{definition}

Suppose $(\C_i)$ is $\rho$-product-expanding. It is proved in Ref.~\cite{PK2023RobustlyTestable} that any subcollection of the codes is also product-expanding. Extremely, when we restrict to each single code $\C_i$, this indicates that its distance is $\geq \rho n_i$.

We also borrow definitions and results about extendability from Ref.~\cite{KP2025Extendable}. For simplicity, let $n_i \equiv n$. Given the grid $[n]^t$ and the $i$-th direction, a line $\ell$ is simply a subset consisting of all points in that direction:
\begin{align}
	\{a_1\} \times \cdots \times \{a_{i-1}\} \times [n] \times \{a_{i+1}\} \times \cdots \times \{a_t\}, \quad a_j \in [n]. 
\end{align}
Let $M \subseteq [n]^t$, the collection $L(M)$ contains all lines in $M$. Given codes $\C_1,\ldots,\C_t$, we also define for any line $\ell$ in the $i$-th direction,
\begin{align}
	\C_\ell \coloneq \{c \in \mathbb{F}_q^{[n]^t}: \supp(c) \subset \ell, c \vert_\ell \in \C_i \}.
\end{align}

\begin{definition}[Inner-generation] \label{def:inner_generated}
	A set $M \subseteq [n]^t$ is \emph{inner-generated} for $\C_1 \boxplus \cdots \boxplus \C_t$ if for any codeword $c$ with $\supp(c) \subseteq M$, $c \in \sum_{\ell \in L(M)} \C_\ell$.
\end{definition}

\begin{definition}[Extendability] \label{def:extendable}
	A set $M \subseteq [n]^t$ is \emph{extendable} for $\C_1 \otimes \cdots \otimes \C_t$ if every vector $c_M \in \mathbb{F}_q^M$ satisfying $c_M \vert_\ell \in \C_i$ for any $i$-th directional line in $L(M)$ can be extended to a global codeword $c \in \C_1 \otimes \cdots \otimes \C_t$.
\end{definition}

\begin{proposition}[\cite{KP2025Extendable}] \label{lemma:inner-g/ext}
	A set $M \in [n]^t$ is inner-generated for $\C_1 \boxplus \cdots \boxplus \C_t$ if and only if it is extendable in $\C_1^\perp \otimes \cdots \otimes \C_t^\perp$.
\end{proposition}

\begin{definition}[$\epsilon$-closed set] \label{def:epsilon_close}
	The \emph{$\epsilon$-closure} of a set $M \subseteq [n]^t$ is the minimum set $M_\epsilon$ containing $M$ such that for any line $\ell$, either $\ell \subseteq M_\epsilon$ or $\vert \ell \cap M_\epsilon \vert < \epsilon n$. If $M = M_\epsilon$, it is \emph{$\epsilon$-closed}.
\end{definition}

\begin{lemma}[\cite{KP2025Extendable}] \label{lemma:inner-g/rho}
	Let $\C_1,\ldots,\C_t$ be codes in $\mathbb{F}_q^n$. If for any $\epsilon$-closed subset $M \subseteq [n]^t$ is inner-generated for $\C_1 \boxplus \cdots \boxplus \C_t$, then
	\begin{align}
		\rho(\C_1,\ldots,\C_t) \geq \frac{\epsilon^t}{t(2^t+1)^t}.
	\end{align} 
\end{lemma}

\begin{definition}[Khatri--Rao product] \label{def:KR_product}
	Let $A \in \mathbb{F}_q^{m_A \times n}$, $B \in \mathbb{F}_q^{m_B \times n}$. We denote by $A \ast B$ the $m_A m_B \times n$ matrix obtained from their column-wise tensor product, as a special case of \emph{Khatri--Rao product}. 
\end{definition}

Particularly, let $\C_1,\C_2$ be classical codes of length $n$. Then $\C_1 \ast \C_2$ denotes the code generated by the Khatri--Rao product of any generator matrices of $\C_1$ and $\C_2$. The definition is independent of the choice of generator matrices. Equivalently,  $\C_1 \ast \C_2$ is the span of all componentwise products $c_1*c_2$ with $c_i\in C_i$. 
In this sense, this code product is also called \emph{Schur product} in the literature~\cite{randriambololona2014}.


\section{Cohomological invariant forms on sheaf codes}\label{sec:cup and cap on sheaf}

Cup and cap products are fundamental operations on (co)chain complexes in the construction of cohomological invariant forms: the Leibniz rules ensure that they are well-defined on (co)homology classes. Moreover, the natural pairing between chains and cochains provides the basic linear cohomological invariant form, which can be extended to multi-linear forms by using cup and cap products. Our transversal logical $\mCZ$ gates on almost-good qLDPC codes and qLTCs are induced from invariant forms built in this way.

\subsection{Cup and cap products on sheaved simplicial complexes}

In this section, we restate the definition of cup and cap products on sheaved simplicial complexes constructed in Ref.~\cite{Li2025Poincare} for completeness.

\begin{definition}[Cup product]
	Given a simplicial complex $X$, with two sheaves $\mathcal F,\mathcal G$,  we define the \emph{cup product}
	\begin{align}
		\smile: C^p(X,\mathcal{F}) \times C^q(X,\mathcal{G}) \longrightarrow C^{p+q}(X,\mathcal{F} \otimes \mathcal{G})
	\end{align}
	as follows: for $\alpha\in C^p(X,\F)$, $\beta\in C^{q}(X,\mathcal{G})$, $\sigma=[ v_0,...,v_{p+q}]$ a $(p+q)$-simplex, let ${}_p\sigma=[v_0,v_1,...,v_p]$ be the \emph{former p-face} and $\sigma_q=[v_p,\cdots, v_{p+q}]$ be the \emph{latter q-face},
	\begin{align}
		(\alpha \smile \beta)(\sigma) \coloneqq  \mathcal{F}_{{}_p\sigma, \sigma}(\alpha( {}_p\sigma ) ) \otimes \mathcal{G}_{\sigma_q, \sigma}(\beta( \sigma_q ) ) .
	\end{align}
\end{definition}

\begin{proposition}[\cite{Li2025Poincare}, Proposition~4.3]\label{prop:cup_Leibniz}
	The cup product satisfies the \emph{Leibniz rule}: for $\alpha\in C^p(X,\mathcal F)$ and $\beta\in C^q(X,\mathcal G)$,
	\begin{align}
		\delta(\alpha \smile \beta) = (\delta \alpha) \smile \beta + \alpha \smile (\delta \beta). 
	\end{align}
\end{proposition}
\begin{proof}
For a $(p+q+1)$-cell $\sigma = [v_0,...,v_{p+q+1}]$,
\begin{align}
& ( \delta(\alpha \smile \beta) )(\sigma) 
		= \sum_{\sigma' \lessdot \sigma} (\mathcal{F} \otimes \mathcal{G})_{\sigma',\sigma}[ (\alpha \smile \beta)(\sigma') ] \notag \\
		= & \sum_{i = 0}^p (\mathcal{F} \otimes \mathcal{G})_{\sigma \setminus v_i,\sigma}[ 
		\mathcal{F}_{{}_{p+1}\sigma \setminus v_i, \sigma \setminus v_i}(\alpha( {}_{p+1}\sigma \setminus v_i ) ) \otimes \mathcal{G}_{\sigma_q, \sigma \setminus v_i}(\beta( \sigma_q ) )  ] \notag \\
		& + \sum_{i = p+1}^{p+q+1} (\mathcal{F} \otimes \mathcal{G})_{\sigma \setminus v_i,\sigma}[ 
		\mathcal{F}_{{}_p\sigma, \sigma \setminus v_i}(\alpha( {}_p\sigma ) ) \otimes \mathcal{G}_{\sigma_{q+1} \setminus v_i, \sigma \setminus v_i }(\beta( \sigma_{q+1} \setminus v_i) )  ] \\ 
		= & \sum_{i = 0}^p 
		\mathcal{F}_{{}_{p+1}\sigma \setminus v_i, \sigma }(\alpha( {}_{p+1}\sigma \setminus v_i ) ) \otimes \mathcal{G}_{\sigma_q, \sigma }(\beta( \sigma_q ) ) 
		+ \sum_{i = p+1}^{p+q+1} 
		\mathcal{F}_{{}_p\sigma, \sigma }(\alpha( {}_p\sigma ) ) \otimes \mathcal{G}_{\sigma_{q+1} \setminus v_i, \sigma }(\beta( \sigma_{q+1} \setminus v_i) )  \notag \\ 
		= & \sum_{i = 0}^{p+1}
		\mathcal{F}_{{}_{p+1}\sigma \setminus v_i, \sigma }(\alpha( {}_{p+1}\sigma \setminus v_i ) ) \otimes \mathcal{G}_{\sigma_q, \sigma }(\beta( \sigma_q ) ) 
		+ \sum_{i = p}^{p+q+1} 
		\mathcal{F}_{{}_p\sigma, \sigma }(\alpha( {}_p\sigma ) ) \otimes \mathcal{G}_{\sigma_{q+1} \setminus v_i, \sigma }(\beta( \sigma_{q+1} \setminus v_i) ), \notag
\end{align} 
where in the last line, we add 
\begin{equation}
\mathcal{F}_{{}_{p}\sigma, \sigma }(\alpha( {}_{p}\sigma ) ) \otimes \mathcal{G}_{\sigma_q, \sigma }(\beta( \sigma_q ) )
\end{equation}
twice and hence the equality holds. On the other hand,
\begin{align}
	\begin{aligned}
		( (\delta \alpha) \smile \beta) (\sigma) 
		= & \mathcal{F}_{{}_{p+1}\sigma, \sigma}( (\delta\alpha) ( {}_{p+1}\sigma ) ) \otimes \mathcal{G}_{\sigma_q, \sigma}(\beta( \sigma_q ) ) \\
		= & \sum_{i = 0}^{p+1} \mathcal{F}_{{}_{p+1}\sigma, \sigma}\, \mathcal{F}_{{}_{p+1}\sigma \setminus v_i, {}_{p+1}\sigma} (\alpha ( {}_{p+1}\sigma \setminus v_i ) ) \otimes \mathcal{G}_{\sigma_q, \sigma}(\beta( \sigma_q ) ) \\
		=  & \sum_{i = 0}^{p+1} \mathcal{F}_{{}_{p+1}\sigma \setminus v_i, \sigma} (\alpha ( {}_{p+1}\sigma \setminus v_i ) ) \otimes \mathcal{G}_{\sigma_q, \sigma}(\beta( \sigma_q ) ) 
	\end{aligned}
\end{align}
Similarly,
\begin{align}
	(\alpha \smile (\delta \beta))(\sigma) = 
	\sum_{i = p}^{p+q+1} 
	\mathcal{F}_{{}_p\sigma, \sigma }(\alpha( {}_p\sigma ) ) \otimes \mathcal{G}_{\sigma_{q+1} \setminus v_i, \sigma }(\beta( \sigma_{q+1} \setminus v_i) )
\end{align}
and this finishes the proof.
\end{proof}

Although the cap product defined below is not directly employed in the present paper, it is naturally dual to the cup product, and we expect it to play a useful role in future developments. Indeed, in Ref.~\cite{Li2025Poincare}, the Poincar\'e duality map for quantum codes was realized through the following cap product. This structure is important for estimating the number of logical transversal $\CZ$ gates on good qLDPC codes constructed there.

\begin{definition}[Cap product]
Given a simplicial complex $X$, with two sheaves $\mathcal F,\mathcal G$, we define the \emph{cap product}
\begin{equation}
\frown:C^p(X,\mathcal{F})\times C_{p+q}(X,\mathcal{F}\otimes\mathcal{G})\longrightarrow C_q(X,\mathcal G)
\end{equation}
as follows: for $\alpha \in C^{p}(X,\mathcal{F}), x\in C_{p+q}(X,\mathcal{F}\otimes\mathcal{G})$,
\begin{align}
	\alpha\frown x\coloneqq \sum_{\sigma\in X(p+q)}\mathcal{G}^T_{\sigma_q,\sigma}\langle\mathcal{F}_{{}_p\sigma,\sigma}\alpha({}_p\sigma),x(\sigma)\rangle_{\mathcal{F}}\cdot \sigma_q.
\label{eq:cap_other_definition}
\end{align}
\end{definition}

\begin{proposition}
    For $\alpha\in C^p(X,\mathcal{F})$ and $x\in C_{p+q+1}(X,\mathcal{F}\otimes\mathcal{G})$, we have $\alpha\frown x\in C_{q+1}(X,\mathcal{G})$ and  $\partial(\alpha\frown x)=(\delta\alpha)\frown x+\alpha\frown (\partial x)$
\end{proposition}

\begin{proof}
Consider a single term $x=x(\sigma)\cdot\sigma$, $\sigma=[v_0,\cdots,v_{p+q+1}]$, then
    \begin{equation}
        \begin{aligned}
            &\alpha\frown (\partial x)=\alpha\frown \sum_{i=0}^{p+q+1}(\mathcal F\otimes\mathcal G)^T_{\sigma\setminus v_i,\sigma}x(\sigma)\cdot(\sigma\setminus v_i)\\
            &=\sum_{i=0}^{p+q+1}\mathcal G^T_{(\sigma\setminus v_i)_q,\sigma\setminus v_i}\langle\mathcal F_{{}_{p}(\sigma\setminus v_i),\sigma\setminus v_i}\alpha({}_p(\sigma\setminus v_i)),(\mathcal F\otimes\mathcal G)^T_{\sigma\setminus v_i,\sigma}x(\sigma)\rangle_{\mathcal F}\cdot(\sigma\setminus v_i)_q\\
            &=\sum_{i=0}^{p+q+1}\mathcal G^T_{(\sigma\setminus v_i)_q,\sigma}\langle\mathcal F_{{}_{p}(\sigma\setminus v_i),\sigma}\alpha({}_p(\sigma\setminus v_i)),x(\sigma)\rangle_{\mathcal F}\cdot(\sigma\setminus v_i)_q\\
            &=\sum_{i=0}^p\mathcal G^T_{\sigma_q,\sigma}\langle\mathcal F_{{}_{p+1}\sigma\setminus v_i,\sigma}\alpha({}_{p+1}\sigma\setminus v_i),x(\sigma)\rangle_\mathcal F\cdot\sigma_q+\sum_{i=p+1}^{p+q+1}\mathcal G^T_{\sigma_{q+1}\setminus v_i,\sigma}\langle\mathcal F_{{}_p\sigma,\sigma}\alpha({}_p\sigma),x(\sigma)\rangle_{\mathcal F}\cdot(\sigma_{q+1}\setminus v_i)\\
            &=\sum_{i=0}^{p+1}\mathcal G^T_{\sigma_q,\sigma}\langle\mathcal F_{{}_{p+1}\sigma\setminus v_i,\sigma}\alpha({}_{p+1}\sigma\setminus v_i),x(\sigma)\rangle_\mathcal F\cdot\sigma_q+\sum_{i=p}^{p+q+1}\mathcal G^T_{\sigma_{q+1}\setminus v_i,\sigma}\langle\mathcal F_{{}_p\sigma,\sigma}\alpha({}_p\sigma),x(\sigma)\rangle_{\mathcal F}\cdot(\sigma_{q+1}\setminus v_i),
        \end{aligned}
    \end{equation}
    where the last equality is derived by adding
    \begin{equation}
        \mathcal G^T_{\sigma_q,\sigma}\langle\mathcal F_{{}_p\sigma,\sigma}\alpha({}_p\sigma),x(\sigma)\rangle_{\mathcal F}\cdot\sigma_q
    \end{equation}
    twice. Note that
    \begin{equation}
        \begin{aligned}
            (\delta\alpha)\frown x&=\mathcal G^T_{\sigma_q,\sigma}\langle\mathcal F_{{}_{p+1}\sigma,\sigma}((\delta\alpha)({}_{p+1}\sigma)),x(\sigma)\rangle_{\mathcal F}\cdot\sigma_q\\
            &=\mathcal G^T_{\sigma_q,\sigma}\langle\mathcal F_{{}_{p+1}\sigma,\sigma}\sum_{i=0}^{p+1}\mathcal F_{{}_{p+1}\sigma\setminus v_i, {}_{p+1}\sigma}\alpha({}_{p+1}\sigma\setminus v_i),x(\sigma)\rangle_{\mathcal F}\cdot\sigma_q\\
            &=\sum_{i=0}^{p+1}\mathcal G^T_{\sigma_q,\sigma}\langle\mathcal F_{{}_{p+1}\sigma\setminus v_i,\sigma}\alpha({}_{p+1}\sigma\setminus v_i),x(\sigma)\rangle_{\mathcal F}\cdot \sigma_{q}.
        \end{aligned}
    \end{equation}
On the other hand,
\begin{equation}
    \begin{aligned}
        \partial(\alpha\frown x)&=\partial (\mathcal G^T_{\sigma_{q+1},\sigma}\langle \mathcal F_{{}_p\sigma,\sigma}\alpha({}_p\sigma),x(\sigma)\rangle_{\mathcal F}\cdot \sigma_{q+1})\\
        &=\sum_{i=p}^{p+q+1}\mathcal G^T_{\sigma_{q+1}\setminus v_i,\sigma}\langle\mathcal F_{{}_p\sigma,\sigma}\alpha({}_p\sigma),x(\sigma)\rangle_{\mathcal F}\cdot (\sigma_{q+1}\setminus v_i).
    \end{aligned}
\end{equation}
Therefore, we have
\begin{equation}
\partial(\alpha\frown x)=\alpha\frown (\partial x)+(\delta \alpha)\frown x.
\end{equation}
\end{proof}

\subsection{Barycentric subdivision of cell complexes}
The almost-good qLDPC codes and qLTCs are built from cubical complexes rather than simplicial complexes. To extend the notion of cup and cap products to general cell complexes, we borrow the method of barycentric subdivision.
\begin{definition}[Subdivision functor]
    The \emph{subdivision functor} $\sd$ is a functor from cell posets to simplicial complexes. For each cell poset $X$, $\sd X$ is defined to be the simplicial complex called the \emph{barycentric subdivision} of $X$, whose vertices correspond to the cells of $X$, and whose $k$-simplices $[\sigma_0,\sigma_1,\cdots,\sigma_k]$ are the strictly increasing chains of cells in $X$.
\begin{align}
\sigma_0 <\sigma_1 <\cdots<\sigma_k\in X.
\end{align}
    Let $X,Y$ be cell posets and let $f: X \to  Y$ be an order-preserving map. Then
    \begin{align}
        \sd f:\sd X\to \sd Y
    \end{align}
    is defined to be a map given by, for each $k$-cell $[\sigma_0,\sigma_1,\cdots,\sigma_k]\in \sd X$, 
    \begin{align}
        (\sd f)([\sigma_0,\dots,\sigma_k])\coloneqq[f(\sigma_0),\dots,f(\sigma_k)].
    \end{align}
It is straightforward to verify that $\sd f$ is order-preserving.
\end{definition}

\begin{definition}[Carrier map]
    Suppose $X$ is a cell poset, the \emph{carrier map} $s:\sd X\to X$ is defined by, for each $k$-cell $[\sigma_0,\sigma_1,\cdots,\sigma_k]\in \sd X$,
    \begin{align}
        s([\sigma_0,\sigma_1,\cdots,\sigma_k])\coloneqq \sigma_k.
    \end{align}
    We call $\sigma_k$ to be the \emph{carrier} of $[\sigma_0,\sigma_1,\cdots,\sigma_k]$, since it is the minimal cell in $X$ that contains the cell $[\sigma_0,\sigma_1,\cdots,\sigma_k]\in \sd X$.
\end{definition}

It is a direct calculation to verify the following proposition:
\begin{proposition}
    The carrier map $s$ is a natural transformation between the subdivision functor \emph{sd} and the identity functor, i.e., for any cell posets $X,Y$ and any order-preserving map $f:X\to Y$, the following diagram commutes:
\begin{equation}
	\begin{tikzcd}
	\operatorname{sd} X \arrow[r, "s"] \arrow[d, "\text{\emph{sd}} f"'] & X \arrow[d, "f"]
    \\
	\operatorname{sd} Y \arrow[r, "s"'] & Y
	\end{tikzcd}
	\end{equation}
\end{proposition}

\begin{definition}[Subdivision map]
    Let $X$ be a cell poset, the \emph{subdivision map}
    \begin{align}
        S_\#:C_\bullet(X,\mathbb F_q)\to C_\bullet(\sd X,\mathbb F_q),
    \end{align}
    is a linear map defined by, for each $k$-cell $\sigma\in X$,
    \begin{align}
        S_\#(\sigma)\coloneqq \sum_{\sigma_0<\cdots<\sigma_{k-1}<\sigma}[\sigma_0,\cdots,\sigma_{k-1},\sigma].
    \end{align}
\end{definition}

\begin{proposition}\label{prop:subdivision is chain map}
    The subdivision map $S_\#$ defined above is a chain map.
\end{proposition}
\begin{proof}
    By definition,
    \begin{align}
    \begin{aligned}
        \partial S_\#(\sigma)&=\partial\sum_{\sigma_0<\cdots<\sigma_{k-1}<\sigma}[\sigma_0,\cdots,\sigma_{k-1},\sigma]\\
        &=\sum_{\sigma_0<\cdots<\sigma_{k-1}<\sigma}\bigg(\sum_{i=0}^{k-1}[\sigma_0,\cdots,\hat{\sigma}_i,\cdots,\sigma_{k-1},\sigma]+[\sigma_0,\cdots,\sigma_{k-1}]\bigg).
    \end{aligned}
    \end{align}
    Since we only consider regular cell complexes, each $1$-dimensional edge in $X$ is joint to two different vertices. Therefore, if we fix $i=0$ in the first term, we have
\begin{align}
    \sum_{\sigma_0<\cdots<\sigma_{k-1}<\sigma}[\sigma_1,\cdots,\sigma_{k-1},\sigma]=\sum_{\sigma_1<\cdots<\sigma_{k-1}<\sigma}\sum_{\sigma_0\lessdot\sigma_1}[\sigma_1,\cdots,\sigma_{k-1},\sigma]=0.
\end{align}
    For convenience, we write $\sigma_k\coloneqq\sigma$, then for $1\le i\le k-1$,
    \begin{align}
        \sum_{\sigma_0<\cdots<\sigma_{k-1}<\sigma}[\sigma_0,\cdots,\hat{\sigma}_i,\cdots,\sigma_{k-1},\sigma]=\sum_{\sigma_0<\cdots<\hat{\sigma}_i<\cdots<\sigma_k}\sum_{\sigma_{i-1}\lessdot\sigma_i\lessdot\sigma_{i+1}}[\sigma_0,\cdots,\hat{\sigma}_i,\cdots,\sigma_{k}].
    \end{align}
    Still by our convention, for fixed $\sigma_{i-1}<\sigma_{i+1}$, $\dim \sigma_{i+1}=\dim\sigma_{i-1}+2$, there is an even number of $\sigma_i\in X(i)$ such that $\sigma_{i-1}\lessdot\sigma_i\lessdot\sigma$. Therefore, the above term equals zero, and
    \begin{align}
        \partial S_\#(\sigma)=\sum_{\sigma_0<\cdots<\sigma_{k-1}<\sigma}[\sigma_0,\cdots,\sigma_{k-1}].
    \end{align}
    On the other hand,
    \begin{align}
    \begin{aligned}
        S_\#\partial(\sigma)&=S_\#\sum_{\sigma_{k-1}\lessdot \sigma}\sigma_{k-1}\\
        &=\sum_{\sigma_{k-1}\lessdot \sigma}\sum_{\sigma_0<\cdots<\sigma_{k-1}}[\sigma_0,\cdots,\sigma_{k-1}]\\
        &=\sum_{\sigma_0<\cdots<\sigma_{k-1}<\sigma}[\sigma_0,\cdots,\sigma_{k-1}].
    \end{aligned}
    \end{align}
    As a result, $\partial S_\#=S_\#\partial$, and $S_\#$ is a chain map.
\end{proof}

\begin{definition}[Approximate inverse]
An \emph{approximate inverse} of the subdivision map $S_\#$ is a chain map 
\begin{align}
    A_\#:C_\bullet(\sd X,\mathbb F_q)\to C_\bullet(X,\mathbb F_q),
\end{align}
satisfying
\begin{enumerate}
    \item Carrier condition: for every cell $\rho\in\sd X$,
    \begin{align}
        A_\#(\rho)\in C_{\dim\rho}(X_{\leq s(\rho)},\mathbb F_q),
    \end{align}
    \item Left inverse condition:
    \begin{align}
        A_\#S_\#=\mathrm{id}_{C_\bullet(X;\mathbb F_q)}.
    \end{align}
\end{enumerate}
\end{definition}

When it is necessary to emphasize the domain, we prefer to write $S_{X,\#}$ and $A_{\sd X,\#}$ for $S_\#$ and $A_\#$ defined above, respectively. An important property of the barycentric subdivision is that such an approximate inverse always exists.

\begin{proposition}\label{prop:existence of approximate inverse}
For any finite regular cell complex $X$, an approximate inverse of the subdivision map always exists.
\end{proposition}

\begin{proof}
We construct $A_\#$ inductively on the dimension of cells in $\sd X$.

\medskip
\noindent
\textbf{Step 0: definition of $0$-cells.} For each cell $\sigma$ of $X$, we choose a vertex $v_\sigma\in X(0)$ with $v_\sigma\le \sigma$, and require that
$v_\sigma=\sigma$ whenever $\sigma$ itself is a $0$-cell.
A $0$-cell of $\operatorname{sd}X$ is precisely a cell $\sigma$ of $X$.
Define
\begin{align}
    A_\#([\sigma]):=v_\sigma\in C_0(X,\mathbb F_q).
\end{align}
This clearly satisfies the carrier condition, since $v_\sigma\le \sigma=s([\sigma])$. Moreover, if $v$ is a $0$-cell of $X$, then $S_\#(v)=[v]$ and hence
\begin{align}
    A_\#S_\#(v)=A_\#([v])=v.
\end{align}
Thus, $A_\#S_\#=\mathrm{id}$ on $C_0(X;\mathbb F_q)$.

\medskip
\noindent
\textbf{Step 1: definition of $1$-cells.}
Let $\rho=[\sigma_0,\sigma_1]$ be a $1$-cell of $\sd X$.
Then $\partial \rho=[\sigma_0]+[\sigma_1]$, so by Step~0,
\begin{align}
    A_\#(\partial \rho)=v_{\sigma_0}+v_{\sigma_1}\in C_0(X_{\le \sigma_1},\mathbb F_q).
\end{align}
Since $v_{\sigma_0}\le \sigma_0\le \sigma_1$ and $v_{\sigma_1}\le \sigma_1$,
both vertices lie in $X_{\le \sigma_1}$. Note that $X_{\le \sigma_1}$ is homeomorphic to a closed ball and is path-connected. Therefore, there exists a $1$-chain $b_\rho\in C_1(X_{\le \sigma_1},\mathbb F_q)$ such that
\begin{align}
    \partial b_\rho=v_{\sigma_0}+v_{\sigma_1}=A_\#(\partial \rho).
\end{align}
We define $A_\#(\rho)\coloneqq b_\rho$. Then extending the map linearly defines
\begin{align}
    A_\#:C_1(\sd X,\mathbb F_q)\to C_1(X,\mathbb F_q).
\end{align}
By construction, $\partial A_\#(\rho)=A_\#(\partial \rho)$, so $A_\#$ is a chain map in degree $1$, and the carrier condition holds.

Now let $\sigma$ be a $1$-cell of $X$. Since $S_\#(\sigma)$ is supported in
$s^{-1}(\sigma)$, the carrier condition gives
\begin{align}
    A_\#S_\#(\sigma)\in C_1(X_{\le \sigma},\mathbb F_q).
\end{align}
We set $d_\sigma:=A_\#S_\#(\sigma)-\sigma\in C_1(X_{\le \sigma},\mathbb F_q)$. Because both $A_\#$ and $S_\#$ are chain maps in degree $1$, and $A_\#S_\#=\mathrm{id}$ on $C_0(X,\mathbb F_q)$, we have
\begin{align}
    \partial d_\sigma=A_\#S_\#(\partial \sigma)-\partial \sigma=0.
\end{align}
This indicates that $d_\sigma$ is a $1$-cycle in $C_1(X_{\le \sigma},\mathbb F_q)$. Since $X_{\le \sigma}$ is homeomorphic to a closed $1$-ball, $H_1(X_{\le \sigma},\mathbb F_q)=0$ and there is no $2$-dimensional cell, so $d_\sigma=0$, i.e.
\begin{align}
    A_\#S_\#(\sigma)=\sigma,
\end{align}
which shows $A_\#S_\#=\mathrm{id}$ on $C_1(X,\mathbb F_q)$.

\medskip
\noindent
\textbf{Inductive hypothesis.}
Assume that for some $n\ge 2$ we have already defined $A_\#$ on
$C_k(\operatorname{sd}X;\mathbb F_q)$ for all $k<n$, in such a way that:

\begin{enumerate}
    \item $\partial A_\# = A_\#\partial$ on $C_k(\sd X;\mathbb F_q)$ for all $k<n$;
    \item for every cell $\rho$ of $\sd X$ with $\dim\rho< n$,
    \[
    A_\#(\rho)\in C_{\dim\rho}(X_{\le s(\rho)},\mathbb F_q);
    \]
    \item $A_\#S_\#=\mathrm{id}$ on $C_k(X;\mathbb F_q)$ for all $k<n$.
\end{enumerate}

We now define $A_\#$ on $n$-cells of $\operatorname{sd}X$.

\medskip
\noindent
\textbf{Step $n$: extension to $n$-cells.}
Let $\rho$ be an $n$-cell of $\sd X$ for $n\ge 2$. By the inductive hypothesis, $A_\#$ is defined on $\partial\rho$, and
\begin{align}
    A_\#(\partial\rho)\in C_{n-1}(X_{\le s(\rho)},\mathbb F_q).
\end{align}
Since $A_\#$ is already a chain map in lower degrees,
\begin{align}
    \partial\big(A_\#(\partial\rho)\big)=A_\#(\partial^2\rho)=0.
\end{align}
This further indicates that $A_\#(\partial\rho)$ is an $(n-1)$-cycle in $C_{n-1}(X_{\le s(\rho)},\mathbb F_q)$. Notice that $X_{\le s(\rho)}$ is homeomorphic to a closed ball, hence
\begin{align}
    H_{n-1}(X_{\le s(\rho)},\mathbb F_q)=0,\quad \text{for}\ n\geq 2.
\end{align}
Then we can choose a chain $ b_\rho\in C_n(X_{\le s(\rho)},\mathbb F_q)$ such that $\partial b_\rho = A_\#(\partial\rho)$. We define $A_\#(\rho):=b_\rho$, and extending linearly over all $n$-cells of $\sd X$ gives a map
\begin{align}
    A_\#:C_n(\operatorname{sd}X;\mathbb F_q)\to C_n(X;\mathbb F_q),
\end{align}
satisfying $\partial A_\#(\rho)=A_\#(\partial\rho)$ for every $n$-cell $\rho$. Consequently, $A_\#$ is now a chain map in degree $n$. The carrier condition is immediate from the construction:
\begin{align}
    b_\rho=A_\#(\rho)\in C_n(X_{\le s(\rho)},\mathbb F_q).
\end{align}

Now we verify the left inverse condition. Let $\sigma$ be an $n$-cell of $X$.
Since $S_\#(\sigma)$ is supported in $s^{-1}(\sigma)$, the carrier condition implies
\begin{align}
    A_\#S_\#(\sigma)\in C_n(X_{\leq\sigma},\mathbb F_q).
\end{align}
Consider the difference $d_\sigma:=A_\#S_\#(\sigma)-\sigma$. As both $A_\#$ and $S_\#$ are chain maps, and by the inductive hypothesis
$A_\#S_\#=\mathrm{id}$ on $C_{n-1}(X;\mathbb F_q)$, we have
\begin{align}
    \partial d_\sigma=A_\#S_\#(\partial\sigma)-\partial\sigma=0.
\end{align}
Thus $d_\sigma$ is an $n$-cycle in $C_n(X_{\le\sigma};\mathbb F_q)$. But $X_{\le\sigma}$ is homeomorphic to a closed $n$-ball, so $H_n(X_{\le\sigma};\mathbb F_q)=0$ and there are no $(n+1)$-dimensional cells. Hence $d_\sigma=0$ and $A_\#S_\#(\sigma)=\sigma$. By linearity,
\begin{align}
    A_\#S_\#=\id_{C_n(X;\mathbb F_q)}.
\end{align}

This completes the induction: $A_\#$ is defined in all degrees. It is a chain map, satisfies the carrier condition, and obeys the left inverse condition.
\end{proof}

For clarity, we would like to write $A_\#$ in a more explicit way by, for each $\rho\in \sd X$, 
\begin{align}
    A_\#(\rho)=\sum_{\sigma\in X(\dim\rho)}A_{\rho,\sigma}\cdot \sigma,
\end{align}
where each $A_{\rho,\sigma}\in\mathbb F_q$ is a number. Then $\partial A_\#=A_\#\partial$ directly implies that
\begin{align}
    \sum_{\sigma\in X(\dim\rho)}\ \sum_{\sigma'\lessdot\sigma} A_{\rho,\sigma}\cdot \sigma'=\sum_{\rho'\lessdot\rho}\ \sum_{\sigma'\in X(\dim\rho')} A_{\rho',\sigma'}\cdot\sigma'.
\end{align}
By rearranging the order of the sum, for each $\sigma'\in X(\dim\rho-1)$ we get
\begin{align}
\sum_{\sigma\gtrdot\sigma'}A_{\rho,\sigma}=\sum_{\rho'\lessdot\rho}A_{\rho',\sigma'}.
\end{align}

\medskip
We now introduce sheaf-valued versions of the subdivision map and the approximate inverse for a cell poset $X$ equipped with a sheaf $\mathcal F$. For simplicity, we abuse the notation to write
\begin{align}
    S_\#:C_\bullet(X,\mathcal F)\to C_\bullet(\sd X,s^*\mathcal F).
\end{align}
It is defined by, for each $k$-cell $\sigma\in X$ and $x(\sigma)\in\mathcal F_\sigma$,
\begin{align}
    S_\#(x(\sigma)\cdot\sigma)\coloneqq \sum_{\sigma_0<\cdots<\sigma_{k-1}<\sigma}x(\sigma)\cdot[\sigma_0,\cdots,\sigma_{k-1},\sigma].
\end{align}
This is well-defined since $s([\sigma_0,\cdots,\sigma_{k-1},\sigma])=\sigma$. Similarly to the proof of Proposition~\ref{prop:subdivision is chain map}, we can easily check that the sheaf-valued subdivision map defined above is still a chain map. 

For the sheaf-valued approximate inverse, we define
\begin{align}
    A_\#:C_\bullet(\sd X,s^*\mathcal F)\to C_\bullet(X,\mathcal F),
\end{align}
by, for each $\rho\in\sd X$ and $\tilde x(\rho)\in \mathcal F_{s(\rho)}$,
\begin{align}
    A_\#(\tilde x(\rho)\cdot \rho)\coloneqq\sum_{\sigma\in X(\dim\rho)}\mathcal F^T_{\sigma,s(\rho)}\tilde x(\rho) A_{\rho,\sigma}\cdot \sigma.
\end{align}
This is well-defined because by the carrier condition, $A_\#(\rho)$ is a sum of cells contained in $X_{\le s(\rho)}$. 

\begin{proposition}
    The sheaf-valued approximate inverse is still a chain map.
\end{proposition}
\begin{proof}
    For each $\rho\in\sd X$ and $\tilde x(\rho)\in \mathcal F_{s(\rho)}$, by definition, we have
    \begin{equation}
        \begin{aligned}
        \partial A_{\#}(\tilde x(\rho)\cdot\rho)&=\partial\sum_{\sigma\in X(\dim\rho)}\mathcal{F}^T_{\sigma,s(\rho)}\tilde x(\rho)A_{\rho,\sigma}\cdot\sigma\\
        &=\sum_{\sigma\in X(\dim\rho)}\ \sum_{\sigma'\lessdot\sigma}\mathcal F^T_{\sigma',s(\rho)}\tilde x(\rho)A_{\rho,\sigma}\cdot\sigma'\\
        &=\sum_{\sigma'\in X(\dim\rho-1)}(\sum_{\sigma\gtrdot\sigma'}A_{\rho,\sigma})\mathcal{F}^T_{\sigma',s(\rho)}\tilde x(\rho)\cdot\sigma'\\
        &=\sum_{\sigma'\in X(\dim\rho-1)}(\sum_{\rho'\lessdot\rho}A_{\rho',\sigma'})\mathcal{F}^T_{\sigma',s(\rho)}\tilde x(\rho)\cdot\sigma'\\
        &=\sum_{\rho'\lessdot\rho}\ \sum_{\sigma'\in X(\dim\rho')}\mathcal{F}^T_{\sigma',s(\rho)}\tilde x(\rho)A_{\rho',\sigma'}\cdot\sigma'.
    \end{aligned}
    \end{equation}
    On the other hand,
    \begin{equation}
        \begin{aligned}
            A_\#\partial(\tilde x(\rho)\cdot\rho)&=A_\#\sum_{\rho'\lessdot\rho}(s^*\mathcal F)^T_{\rho',\rho}\tilde x(\rho)\cdot\rho'\\
            &=\sum_{\rho'\lessdot\rho}A_\#\mathcal F^T_{s(\rho'),s(\rho)}\tilde x(\rho)\cdot \rho'\\
            &=\sum_{\rho'\lessdot\rho}\ \sum_{\sigma'\in X(\dim\rho')}\mathcal F^T_{\sigma',s(\rho')}\mathcal F^T_{s(\rho'),s(\rho)}\tilde x(\rho)A_{\rho',\sigma'}\cdot\sigma'\\
            &=\sum_{\rho'\lessdot\rho}\ \sum_{\sigma'\in X(\dim\rho')}\mathcal{F}^T_{\sigma',s(\rho)}\tilde x(\rho)A_{\rho',\sigma'}\cdot\sigma'.
        \end{aligned}
    \end{equation}
    Therefore, $\partial A_\#=A_\#\partial$ and the sheaf-valued approximate inverse $A_\#$ is still a chain map.
\end{proof}

A direct calculation shows that for sheaf-valued subdivision map and approximate inverse, we still have
\begin{align}
    A_\#S_\#=\id_{C_\bullet(X,\mathcal F)}.
\end{align}

\subsection{General definition of cup and cap products}\label{sec:cupcap}

Now we are able to define cup and cap products when $X$ is not a simplicial complex. 

\begin{definition}
    Suppose $X$ is a finite regular cell complex equipped with sheaves $\mathcal F$ and $\mathcal G$, then the cup product is defined by the following commutative diagram:
\begin{equation}
    \begin{tikzcd}
    C^p(X,\mathcal{F}) \times C^q(X,\mathcal{G})
  \arrow[r, "\smile"]
  \arrow[d, "A^\# \times A^\#"']
& C^{p+q}(X,\mathcal{F}\otimes\mathcal{G})
   \\
C^p(\sd X,s^*{\mathcal{F}}) \times C^q(\sd X,s^*{\mathcal{G}})
  \arrow[r, "\smile"]
& C^{p+q}(\sd X,s^*({\mathcal{F}}\otimes{\mathcal{G}})) \arrow[u, "S^\#" ]
     \end{tikzcd}
\end{equation}

The cap product is defined by the following commutative diagram:
\begin{equation}
    \begin{tikzcd}
    C^p(X,\mathcal{F}) \times C_{p+q}(X,\mathcal{F}\otimes \mathcal G)
  \arrow[r, "\frown"]
  \arrow[d, "A^\# \times S_\#"']
& C_q(X,\mathcal{G})
   \\
C^p(\sd X, s^*{\mathcal{F}}) \times C_{p+q}(\sd X, s^*({\mathcal{F}}\otimes\mathcal G))
  \arrow[r, "\frown"]
& C_q(\sd X,s^*{\mathcal{G}}) \arrow[u, "A_\#" ]
     \end{tikzcd}
\end{equation}
\end{definition}

It is obvious that the Leibniz rule for cup and cap products still holds. Note that the chain-level definition of cup and cap products inevitably depends on the choice of approximate inverse. Nevertheless, the (co)homology-level definitions are independent of the choice because we have the following proposition:

\begin{proposition}[\protect{\cite[Theorem 7.3.9]{curry2014sheavescosheavesapplications}}]\label{prop: subdivided sheaf has same cohomology}
    Suppose $\F$ is a sheaf on cell poset $X$, then there is an isomorphism
    \begin{equation}
    H^\bullet(X,\mathcal F)\cong H^\bullet(\sd X,s^*\mathcal F).
    \end{equation}
\end{proposition}

Note that the left inverse condition guarantees that $S^*A^*=\id_{H^\bullet(X,\mathcal F)}$. Since $H^\bullet(X,\mathcal F)$ and $H^\bullet(\sd X,s^*\mathcal F)$ are isomorphic finite dimensional vector spaces, $A^*$ and $A_*$ must be the inverse of $S^*$ and $S_*$, respectively. This verifies the independence of the (co)homology-level definition. 

It is straightforward to verify the following property:
\begin{proposition}
    Let $X$ be a sparse cell complex. Then $\sd X$ is sparse, and the operators $S_\#$, $S^\#$, $A_\#$, $A^\#$ are also sparse. The cup and cap products on $X$ are also sparse.
\end{proposition}

\begin{example}
    Let $X$ be a cell poset, $\sigma\in X$ be a $k$-cell and $A_\#$ be an approximate inverse. If we abuse the notation and identify $\sigma \in C^k(X,\mathbb{F}_q)$ as the function mapping $\sigma$ to $1$ and else to $0$, then
    \begin{align}
        \sigma\frown \sigma=A_\#([\sigma])\in C_0(X,\mathbb F_q),
    \end{align}
    which is just a single point. This is because, by definition
    \begin{align}
    \begin{aligned}
        \sigma\frown\sigma&=A_\#(A^\#\sigma\frown S_\# \sigma)\\
        &=A_\#(\langle A^\#\sigma, S_\#\sigma\rangle\cdot [\sigma])\\
        &=A_\#(\langle \sigma,A_\#S_\#\sigma\rangle\cdot[\sigma])\\
        &=A_\#([\sigma]).
    \end{aligned}
    \end{align}
\end{example}

\begin{example}
    Let $X=[0,1]\times[0,1]\subseteq\mathbb R^2$ be a square. Then $X$ is a $2$-dimensional cell complex with the cell structure illustrated below. If we choose $A_\#[f]=v_2$, $A_\#[e_{01}]=v_1$, $A_\#[e_{30}]=v_3$, and $A_\#[v_0,f]=e_{30}+e_{23}$, then the only nonzero cup products in $C^1(X,\mathbb F_q)$ are $e_{01}\smile e_{12}=e_{30}\smile e_{23}=f$.
    \begin{figure}[H]
\centering
\begin{tikzpicture}[scale=1.25, >=Stealth]
\begin{scope}
\node at (1,2.9) {$X$};
\draw[thick] (0,0) -- (2,0) -- (2,2) -- (0,2) -- cycle;

\fill (0,0) circle (2pt);
\fill (2,0) circle (2pt);
\fill (2,2) circle (2pt);
\fill (0,2) circle (2pt);

\node[below left]  at (0,0) {$v_0$};
\node[below right] at (2,0) {$v_1$};
\node[above right] at (2,2) {$v_2$};
\node[above left]  at (0,2) {$v_3$};

\node[below] at (1,0) {$e_{01}$};
\node[right] at (2,1) {$e_{12}$};
\node[above] at (1,2) {$e_{23}$};
\node[left]  at (0,1) {$e_{30}$};

\node at (1,1) {$f$};
\end{scope}

\begin{scope}[xshift=6cm]
\node at (1,2.9) {$\sd X$};

\draw[thick] (0,0) -- (2,0) -- (2,2) -- (0,2) -- cycle;

\fill (0,0) circle (2pt);
\fill (2,0) circle (2pt);
\fill (2,2) circle (2pt);
\fill (0,2) circle (2pt);

\fill (1,0) circle (2pt);
\fill (2,1) circle (2pt);
\fill (1,2) circle (2pt);
\fill (0,1) circle (2pt);

\fill (1,1) circle (2pt);

\draw (0,0) -- (1,0);
\draw (2,0) -- (1,0);

\draw (2,0) -- (2,1);
\draw (2,2) -- (2,1);

\draw (2,2) -- (1,2);
\draw (0,2) -- (1,2);

\draw (0,2) -- (0,1);
\draw (0,0) -- (0,1);

\draw (1,0) -- (1,1);
\draw (2,1) -- (1,1);
\draw (1,2) -- (1,1);
\draw (0,1) -- (1,1);

\draw (1,1) -- (0,0);
\draw (1,1) -- (2,0);
\draw (1,1) -- (2,2);
\draw (1,1) -- (0,2);

\node[below left]  at (0,0) {$[v_0]$};
\node[below right] at (2,0) {$[v_1]$};
\node[above right] at (2,2) {$[v_2]$};
\node[above left]  at (0,2) {$[v_3]$};

\node[below] at (1,0) {$[e_{01}]$};
\node[right] at (2,1) {$[e_{12}]$};
\node[above] at (1,2) {$[e_{23}]$};
\node[left]  at (0,1) {$[e_{30}]$};

\node[below right] at (1,1) {$[f]$};

\end{scope}
\end{tikzpicture}

\end{figure}
\end{example}

We can use cup product to construct cohomological invariant forms. For example, given any cycle $\xi \in C_{i+j+k}(X,\F \otimes \F \otimes \F)$, define
\begin{align}
        I_\xi: C^i(X,\mathcal{F})\times C^j(X,\mathcal{F})\times C^k(X,\mathcal{F})\to\mathbb F_q
    \end{align}
   by
    \begin{align}\label{eq:invariant form from cup}
        I_\xi(\alpha,\beta,\theta) = \langle \alpha\smile\beta \smile \theta, \xi \rangle,
    \end{align}
where $\alpha\in C^i(X,\mathcal{F})$, $\beta\in C^j(X,\mathcal{F})$ and $\theta \in C^k(X,\mathcal{F})$. Since $\xi$ is a cycle, the Leibniz rule of cup product implies that $I_\xi$ is invariant under changing any cocycle representative by a coboundary. Hence $I_\xi$ defines a cohomological invariant form. Moreover, when $X$ is sparse, the locality of the cup product implies that $I_\xi$ is a sparse tensor. Therefore, the corresponding cohomological invariant form induces a constant-depth logical $\CCZ$ circuit. We will show in Section \ref{sec:proof} that $I_\xi$ and its generalizations induce nontrivial multi-controlled-$Z$ gates on almost-good qLDPC codes and qLTCs.

\section{Existence of product-expanding punctured Reed--Solomon codes}\label{sec:GRS}

For simplicity, we only prove Theorem \ref{thm:GRS} for $t = 2,3$ and this suffices to exhibit the significant differences between 2-dimensional and higher dimensional cases. By definition, each punctured RS code $\C_i$ is defined by evaluating over the polynomial space of degree $\leq m_i$. To keep the notation simple, we avoid excessive subscripts by letting all codes defined with dimension $m = \min_i \{m_i\}$ (as it will become clear gradually that dealing with small $m$ is more intricate). Then the polynomial space is spanned by monomials: $1,X,\ldots,X^{m-1}$, and for punctured RS codes, the evaluation map $\ev: \mathbb{F}_q^m \to \mathbb{F}_q^n$ has the following matrix representation:
\begin{align}
	\ev\Big( f(X) = \sum_i c_i X^i \Big) = \begin{pmatrix} c_1 & \cdots & c_n	\end{pmatrix} \begin{pmatrix}
		1      & 1      & \dots  & 1      \\
		x_1 & x_2 & \dots  & x_n  \\
		x_1^2  & x_2^2  & \dots  & x_n^2  \\
		\vdots & \vdots & \ddots & \vdots \\
		x_1^{m-1} & x_2^{m-1} & \dots & x_n^{m-1}
	\end{pmatrix},
\end{align}
where $\{x_1,\ldots,x_n\}$ is the evaluation set. Moreover, evaluating products of polynomials is tantamount to taking tensor product of Vandermonde matrices, which produces the tensor product code. We would index the columns of each Vandermonde matrix by $a,b,c \in [n]$ in the following. 

\subsection{Two-dimensional preliminaries}

We first treat the two-dimensional case. Let $M \subseteq [n]^2$ a subset. Within the 2-dimensional grid, let $\Delta(M) = \Delta$ be the maximal size of $M \cap \ell$ among any full line $\ell \subseteq [n]^2$. Let $\vert M \vert$ denote the total number of points. We set $D_2 = \min_{\mathfrak{d} \in \mathbb{Z}}\left\{ \binom{\mathfrak{d}+1}{2} \geq \vert M \vert \right\}$.
If
\begin{align}\label{eq:2D_size}
	(\Delta - 1) + D_2 \leq m-1,
\end{align}
then the following lemma holds. 

\begin{lemma}\label{lemma:2D_independence}
	For any finite field $\mathbb{F}_q$ with $q > 2^{n^2}$, there always exist evaluation sets $\{y_1,\ldots,y_n\}$ and $\{z_1,\ldots,z_n\}$ such that the evaluation map of $\C_2 \otimes \C_3$ is surjective when restricted to $\mathbb{F}_q^M \subseteq \mathbb{F}_q^{[n]^2}$ for any $M \subseteq [n]^2$ satisfying~\eqref{eq:2D_size}. Taking transposes, this is equivalent to saying that	
	\begin{align}
	\{ (1,\ldots,y_b^{m-1}) \otimes (1,\ldots,z_c^{m-1}): (b,c) \in M \}
	\end{align}
	is an independent set for any of these $M$.
\end{lemma}

As a general property of RS codes, any of its $n-k$ columns in the Vandermonde (generator) matrix are linearly independent, provided the evaluation points are distinct. This makes it quite flexible to choose information sets of free variables of the matrix. However, for a tensor product of Vandermonde matrices, we cannot find free variables arbitrarily. Lemma \ref{lemma:2D_independence} provides one condition when this is still possible. 

\begin{proof}[Proof of Lemma \ref{lemma:2D_independence}]
	Given any specific $M \subseteq [n]^2$, let
	\begin{align}
		\mathcal{P} \coloneq \{f \in \mathbb{F}_q[Y,Z]: \deg f \leq D_2 \} \text{ with } \dim \mathcal{P} = \binom{D_2 + 1}{2} > \vert M \vert. 
	\end{align}
	Polynomials in $\mathcal{P}$ are spanned by monomials $1,X,Y,XY,\ldots,X^{D_2},Y^{D_2}$. 
	
	We also define a polynomial ring $R = \mathbb{F}_q[\cdots,Y_b,\cdots,Z_c,\cdots]$ which exhausting indeterminates $Y_b,Z_c$ over any $b,c \in [n]$. Let 
	\begin{align}
		\mathbb{K} \coloneq \left\{ \frac{A}{B}: A,B \in R, B \neq 0 \right\}
	\end{align}
	be the fraction field obtained from $R$. We prove the lemma over $\mathbb{K}$ and find concrete evaluation points in the last step. It will become clear subsequently that working with the generic field is essential to overcome several difficulties in the proof. For any fixed point $p = (b_0,c_0) \in M$, we also define another polynomial ring $\mathbb{F}_q[\cdots,\hat{Y}_{b_0},\cdots,\hat{Z}_{c_0},\cdots]$ with $Y_{b_0}, Z_{c_0}$ being removed. The associated fraction field $\mathbb{K}_p \subseteq \mathbb{K}$.
	
	Now, we treat $\mathcal{P}$ as a polynomial space over $\mathbb{K}_p$. For the set
	\begin{align}
		T_p = \{(b,c) \in M: b \neq b_0, c \neq c_0\},
	\end{align}
	we consider linear equations
	\begin{align}\label{eq:H_p}
		\begin{pmatrix}
			\vdots & \vdots & \vdots & \vdots & \cdots & \vdots & \vdots \\
			1 & Y_b & Z_c & Y_b Z_c & \cdots & Y_b^{D_2} & Z_c^{D_2} \\
			\vdots & \vdots & \vdots & \vdots & \cdots & \vdots & \vdots
		\end{pmatrix} \begin{pmatrix} \vdots \\ f_{ij} \\ \vdots \end{pmatrix} = 0.
	\end{align}
	As a reminder, $Y_b,Z_c$ are scalars in $\mathbb{K}_p$ which are not evaluated over $\mathbb{F}_q$. Since
	\begin{align}
		\dim_{\mathbb{K}_p} \mathcal{P} \geq \vert M \vert > \vert T_p \vert,  
	\end{align}
	we can solve $f_{ij}$ and define $H_p(Y,Z) = \sum_{i+j \leq D_2} f_{ij} Y^i Z^j \in \mathcal{P}$ such that $H_p(Y,Z)$ vanishes at any $(Y_b,Z_c)$ with $b \neq b_0$ and $c \neq c_0$. More importantly, $H_p(Y,Z)$ can be redefined as a polynomial over $\mathbb{K}$ and by the definition of $\mathbb{K}$, $(Y_{b_0}, Z_{c_0})$ are algebraically independent over $\mathbb{K}_p$. This implies $H_p(Y_{b_0}, Z_{c_0}) \neq 0 \in \mathbb{K}$.
	
	On the other hand, we consider polynomials 
	\begin{align}
		B_p(Y) \coloneq \prod_{b \neq b_0,(b,c_0) \in M} (Y - Y_b), \quad C_p(Z) \coloneq \prod_{c \neq c_0,(b_0,c) \in M} (Z - Z_c).
	\end{align}
	The product $\Psi_p(Y,Z) \coloneq B_p(Y) C_p(Z) H_p(Y,Z)$ now vanishes for any $(Y_b,Z_c)$ labeled by $(b,c) \neq (b_0,c_0)$ inside $M$. By definition, the degree of $B_p$ and $C_p$ are bounded by $\Delta - 1$, and thus
	\begin{align}
		& \deg_Y \Psi(Y,Z) \leq \deg B_p + \deg_Y P \leq (\Delta - 1) + D_2 \leq m-1, \\ 
		& \deg_Z \Psi(Y,Z) \leq \deg C_p + \deg_Z P \leq (\Delta - 1) + D_2 \leq m-1 \\
		\implies &
		\Psi_p \in \{f \in \mathbb{K}[Y,Z]: \deg_Y f, \deg_Z f \leq m-1 \} \cong \mathbb{K}^{m \times m}.
	\end{align}
	Since we can do this for any point in $M$, the evaluation map
	\begin{align}
		\ev_M: \mathbb{K}^{m \times m} \to \mathbb{K}^{M} \subseteq \mathbb{K}^{[n]^2}
	\end{align}
	is surjective. Equivalently
	\begin{align}\label{eq:column_tensors}
		(b,c) \mapsto (1,\ldots,Y_b^{m-1}) \otimes (1,\ldots,Z_c^{m-1}) \in \mathbb{K}^{m \times m}.
	\end{align}
	is injective and the corresponding tensors are independent over $\mathbb{K}$. 
	
	The remaining problem is verifying this result for every admissible $M$ and moving back to the number field $\mathbb{F}_q$. The surjectivity of the evaluation map $\ev_M$ indicates there is a nonzero minor $\det_M (Y_b^i Z_c^j)$ in tensor product of Vandermonde matrices over $\mathbb{K}$. In the formal sense, this minor is the evaluation of a nontrivial polynomial from another polynomial ring
	\begin{align}
		\mathbb{K}[\cdots,\tilde{Y}_b,\cdots,\tilde{Z}_c,\cdots]
	\end{align}
	where $\tilde{Y}_b,\tilde{Z}_c$ are indeterminates but different from $Y_b,Z_c$. By \eqref{eq:column_tensors}, the polynomial of minor is simply a sum of powers of $\tilde{Y}_b$ and $\tilde{Z}_c$. The coefficients are always $1 \in \mathbb{K}$. We denote it by $P_M$. As a comparison, coefficients $f_{ij}$ solved in Eq.~\eqref{eq:H_p} are complicated rational functions in $\mathbb{K}$. 
	
	Since there are only finitely many admissible $M \subseteq [n]^2$ and each is assigned with a nontrivial polynomial $P_M$ given by a minor, 
	\begin{align}
		P(\cdots,\tilde{Y}_b,\cdots,\tilde{Z}_c,\cdots) = \prod_M P_M(\cdots,\tilde{Y}_b,\cdots,\tilde{Z}_c,\cdots) \text{ over } \mathbb{K}
	\end{align}
	is nontrivial and the coefficients are still $1 \in \mathbb{K}$. As a result, given $\mathbb{F}_q$ with $q > 2^{n^2}$, we must have
	\begin{align}
		P(\cdots,\tilde{Y}_b,\cdots,\tilde{Z}_c,\cdots) = \prod_M P_M(\cdots,\tilde{Y}_b,\cdots,\tilde{Z}_c,\cdots) \text{ over } \mathbb{F}_q
	\end{align}
	Since the degree of each indeterminate cannot exceed
	\begin{align}
		\frac{m(m-1)}{2} \cdot \text{ total number of possible } M < 2^{n^2}
	\end{align}
	for large $n$, by Schwartz--Zippel lemma, we can find $y_1,\ldots,y_n$ and $z_1,\ldots,z_n \in \mathbb{F}_q$ such that
	\begin{align}
		P(\cdots,z_b,\cdots,z_c,\cdots) \neq 0 \implies P_M(\cdots,z_b,\cdots,z_c,\cdots) \neq 0 
	\end{align}
	for each $M$. This finishes the proof.
\end{proof}

\begin{remark}
	As a quick check, let $(1,\ldots,y_b^{m-1}) = (1,\ldots,z_c^{m-1})$ be the same evaluation set in $\mathbb{F}_q^n$. Then a row of Eq.~\eqref{eq:H_p} in becomes
	\begin{align}
		\begin{pmatrix}
			1 & y_b & y_c & y_b^2 & \cdots & y_b^{2 D_2} 
		\end{pmatrix}
	\end{align}
	Let $D_2 = m-1$ be the largest possible value and let $M = \{(b,c): b = c \in [n]\}$ be the diagonal set with $(b_0,c_0) = (n,n)$. It is true that we can always find nontrivial $f_{ij}$ to solve Eq.~\eqref{eq:H_p} as the column size is far beyond the row size. However, the matrix is occupied by too many repeated columns and it is easy to check that $H_p$ must also vanish at $(n,n)$. This is one prominent reason that we have to use generic distinct evaluation sets.
\end{remark}

Let $M \subseteq [n]^2$ an $\epsilon$-closed set, it would contain full lines by Definition \ref{def:epsilon_close}. Then the above requirement that $\Delta(M)$ is small can never be fulfilled. Nevertheless, the 2-dimensional independence still holds at points that are not in any of these full lines. To be precise, let $M \subseteq [n]^2$ an $\epsilon$-closed set. Let $M' = M \setminus L(M)$ with $L(M)$ being full lines inside $M$. Then $\Delta(M') \leq \epsilon n$, $\vert M' \vert \leq \epsilon n^2$. Given fixed $\nu \in (0,1)$, let $m = \nu n$ for sufficiently large $n$ and let $D_2 = \sqrt{2 \vert M' \vert}$. More rigorously, we should take the smallest integers greater than each of these numbers, but this slight difference is omitted in the following computations. The inequality 
\begin{align}\label{eq:epsilon_2D}
	(\epsilon n - 1) + \sqrt{2\epsilon} n + \epsilon n \leq m-1 = \nu n -1
\end{align}
can hold for small $\epsilon$ which depends only on $\nu$ and $t = 2$. Then the following lemma holds and Theorem \ref{thm:GRS} can be verified in 2D.

\begin{lemma}\label{lemma:2D_independence_2}
	For any finite field $\mathbb{F}_q$ with $q > 2^{n^2}$, there are always evaluation sets $\{y_1,\ldots,y_n\}$ and $\{z_1,\ldots,z_n\}$ such that for any $\epsilon$-closed $M \subseteq [n]^2$ with $\epsilon$ satisfying \eqref{eq:epsilon_2D} and for any $(b,c) \in M' \subseteq M$, there is a polynomial whose evaluation vanishes inside $(M' \setminus \{(b,c)\} ) \cup L(M) = M \setminus \{(b,c)\}$.
\end{lemma}
\begin{proof}
	For any $\epsilon$-closed $M$, we apply Lemma \ref{lemma:2D_independence} to $M'$. When \eqref{eq:epsilon_2D} holds, by using Vandermonde matrices with $m- \epsilon n$ rows, we can find $\Psi_p(Y,Z)$ with even lower degrees:
	\begin{align}
		(\Delta(M') - 1) + \sqrt{2 \vert M' \vert} \leq m-1 - \epsilon n.
	\end{align} 
	Then, with respect to any concrete evaluation sets $\{y_1,\ldots,y_n\}$ and $\{z_1,\ldots,z_n\}$, the evaluation of $\Psi_p(Y,Z)$ over $\mathbb{F}_q$ supports solely on the point $p \in M'$. 
	
	Now, let $I_2, I_3 \subseteq [n]$ denote the index of full $z$-lines and $y$-lines, respectively. Then
	\begin{align}
		\Psi_p(Y,Z) \prod_{b \in I_2} (Y - y_b) \prod_{c \in I_3} (Z - z_c)
	\end{align} 
	has degree
	\begin{align}\label{eq:2D_independence_2}
		\deg_Y \Psi_p(Y,Z) + \deg \prod_{b \in I_2} (Y - y_b) \leq 	(\epsilon n - 1) + \sqrt{2\epsilon} n  + \epsilon n \leq m-1 
	\end{align}
	Similar inequality holds for $\deg_Z$ and thus the polynomial represents a tensor product code in $\C_2 \otimes \C_3$ that will also vanish on any of the full lines. 
\end{proof}

\begin{corollary}\label{coro:GRS_extendable_2D}
	Any $\epsilon$-closed $M$ with $\epsilon$ satisfying \eqref{eq:epsilon_2D} is extendable for the punctured RS codes $\C_1,\C_2$ obtained from the evaluation sets in Lemma \ref{lemma:2D_independence_2}. As a result
	\begin{align}
		\rho(\C_1^\perp,\C_2^\perp) \geq \frac{\epsilon^2}{2(2^2+1)^2},
	\end{align}
	where $\epsilon$ depends only on $\nu$ and $t = 2$.
\end{corollary}
\begin{proof}
	Let $M$ be any $\epsilon$-closed set with $\epsilon$ satisfying \eqref{eq:epsilon_2D}. Given any local data of codewords on sparse points and full lines in $M$. We first extend the data from full lines $L(M)$. In 2 dimensions, there are always $I_1, I_2 \subseteq [n]$ for which
	\begin{align}
		L(M) = (I_1 \times [n]) \cup ([n] \times I_2) 
	\end{align}
	The geometry of $L(M)$ becomes drastically complicated in high dimension (see Lemma \ref{lemma:3D_line_ext}), but the 2D case is simple. By definition, $\vert I_i \vert < \epsilon n$. In each direction, for each $a,b$, let
	\begin{align}
		\phi_a(X) \coloneq \mathcal{N}_a \prod_{a' \neq a \in I_1}(X - x_{a'}), \quad
		\psi_b(Y) \coloneq \mathcal{N}_b \prod_{b' \neq b \in I_2}(Y - y_{b'}).
	\end{align}
	where $\mathcal{N}_i$ are normalization constant such that, e.g., $\phi_a(x_a) = 1$ and thus the evaluation within $I_1$ is the delta function. Note that for general classical codes, finding such codewords use information sets \cite{Dinur2006,Meir2021,PK2023RobustlyTestable}. On the other hand, local codewords in $I_1 \times [n]$ and $[n] \times I_2$ are
	\begin{align}
		& A_a(Y) \in \C_2, \quad a \in I_1, \quad \deg_Y A_a \leq m-1 \\
		& B_b(X) \in \C_1, \quad b \in I_2, \quad \deg_X B_b \leq m-1.
	\end{align} 
	They agree with each other on the intersections $A_a(y_b) = B_b(x_a)$, $a \in I_1, b \in I_2$. We consider the interpolation:
	\begin{align}
		F(X,Y) = \sum_{a \in I_1} \phi_a(X) A_a(Y) + \sum_{b \in I_2} \psi_b(Y) B_b(X) 
		- \sum_{a \in I_1, b \in I_2} \phi_a(X) \psi_b(Y) A_a(y_b) 
	\end{align}
	By checking the degrees, we see $F(X,Y) \in \C_1 \otimes \C_2$. Then we show that it matches the local data: e.g., for any $y$-line $\{a'\} \times [n]$, 	
	\begin{align}
		& F(x_{a'},Y) = A_{a'}(Y) + \sum_{b \in I_2} \psi_b(Y) B_b(x_{a'}) 
		- \sum_{b \in I_2} \psi_b(Y) A_{a'}(y_b) = A_{a'}(Y)
	\end{align}
	by the intersection agreement: $B_b(x_{a'}) = A_{a'}(y_b)$.
	
	The global extension $F(X,Y)$ would change the residual data on sparse points, but each of them can be redefined independently by using Lemma \ref{lemma:2D_independence_2} and the involved polynomials vanish on $L(M)$. As a result, codewords on full lines are intact. This establishes the global extension and the result follows from Lemma \ref{lemma:inner-g/rho}.
\end{proof}

Demonstrating $\rho(\C_1,\C_2) \geq \frac{\epsilon^2}{2(2^2+1)^2}$ is more straightforward by using Lemma \ref{lem:diagonal-scaling-preserves-extendability} and Corollary \ref{coro:RS_dual} in 2-dimensional.


\subsection{Three- and higher-dimensional extendability}

We now tackle the three-dimensional case. Let $M \subseteq [n]^3$ be $\epsilon$-closed. By definition, $M$ may contain full lines in each of the following three directions:
\begin{align}
	\{a\} \times \{b\} \times [n], \quad \{a\} \times [n] \times \{c\}, \quad [n] \times \{b\} \times \{c\}.
\end{align}
A key difference from the two-dimensional case is that, if one defines
\begin{align}
	I_1 = \{a \in [n]: \text{ if there is some } \{a\} \times \{b\} \times [n] \text{ or } \{a\} \times [n] \times \{c\} \subseteq M \}
\end{align}
with $I_2$ and $I_3$ defined analogously, it is impossible to infer that $\vert I_i \vert \leq \epsilon n$. A counterexample is obtained by taking the full lines to be
\begin{align}\label{eq:diagonal_lines}
	\{a\} \times \{b\} \times [n], \ a = b = 1,\ldots,n.
\end{align}
The difficulty here is that, if we define
\begin{align}
	\prod_{a \in I_1} (X - x_a) \prod_{b \in I_2} (Y - y_b)
\end{align}
to annihilate the nonzero values on these full lines, the degree is exactly $n$ and hence cannot be realized as a codeword in $\C_2 \otimes \C_3$ for small $m$. This reflects the fundamental difficulty of achieving product expansion in three and higher dimensions. Fortunately, Lemma \ref{lemma:2D_independence_2} and a weaker three-dimensional version without full lines help resolve this issue.

Analogously to the setting of Lemma \ref{lemma:2D_independence}, let $M \subseteq [n]^3$ be a subset. Suppose that the maximal size of a line in any direction is $\Delta = \Delta(M)$, the maximal slice size in any direction is $\Sigma = \Sigma(M)$, and the total number of points is $\vert M \vert$. We consider the condition
\begin{align}\label{eq:3D_size}
	2\left( (\Delta - 1) + D_2 \right) + D_3 \leq m-1,
\end{align}
where
\begin{align}
	D_2 = \min_{\mathfrak{d} \in \mathbb{Z}}\left\{ \binom{\mathfrak{d}+1}{2} \geq \Sigma \right\}, \qquad
	D_3 = \min_{\mathfrak{d} \in \mathbb{Z}}\left\{ \binom{\mathfrak{d}+2}{3} \geq  \vert M \vert \right\}.
\end{align}
Compared with the 2D case, the additional term $D_3 \approx \vert M \vert^{1/3}$ appears. In particular, \eqref{eq:3D_size} prevents $M$ from containing any full line. The following lemma then holds in 3D.

\begin{lemma}\label{lemma:3D_independence}
	For any finite field $\mathbb{F}_q$ with $q > 2^{n^3}$, there exist evaluation sets $\{x_1,\ldots,x_n\}$, $\{y_1,\ldots,y_n\}$, and $\{z_1,\ldots,z_n\}$ such that, for any $M \subseteq [n]^3$ satisfying \eqref{eq:3D_size}, the evaluation map on $\C_1 \otimes \C_2 \otimes \C_3$ is surjective when restricted to $\mathbb{F}_q^M \subseteq \mathbb{F}_q^{[n]^3}$. Taking transposes, this is equivalent to saying that
	\begin{align}
		\{ (1,\ldots,x_a^{m-1}) \otimes (1,\ldots,y_b^{m-1}) \otimes  (1,\ldots,z_c^{m-1}): (a,b,c) \in M \}
	\end{align}
	is an independent set for every such $M$.
\end{lemma}
\begin{proof}
	The proof is a straightforward three-dimensional generalization of Lemma \ref{lemma:2D_independence}. Fix $M$, let $p = (a_0,b_0,c_0) \in M$, and set
	\begin{align}
		T_p = \{(a,b,c) \in M: a \neq a_0, b \neq b_0, c \neq c_0 \}, \quad
		\mathcal{P} = \{f(X,Y,Z): \deg f \leq D_3 \}.
	\end{align} 
	By the definition of $D_3$, $\dim \mathcal{P} = \binom{D_3+2}{3} \geq \vert M \vert > \vert T_p \vert$. Let $\mathbb{K}$ and $\mathbb{K}_p$ be the fraction fields obtained from the polynomial rings
	\begin{align}
		\mathbb{F}_q[\cdots,X_a,\cdots,Y_b,\cdots,Z_c,\cdots], \qquad 		\mathbb{F}_q[\cdots,\hat{X}_{a_0},\cdots,\hat{Y}_{b_0},\cdots,\hat{Z}_{c_0},\cdots],
	\end{align}
	respectively. We seek $H_p$ by choosing coefficients $f_{ijk}$ such that
	\begin{align}
		\sum_{i+j+k \leq D_3} f_{ijk} X_a^i Y_b^j Z_c^k = 0
	\end{align}
	for all $(a,b,c) \in T_p$, while $H_p(X_{a_0},Y_{b_0},Z_{c_0}) \neq 0$ in $\mathbb{K}$.
	
	It remains to handle the points lying in the coordinate slices through $p$:
	\begin{align}
		\{(b,c): (a_0,b,c) \in M \}, \quad \{(a,c): (a,b_0,c) \in M \}, \quad \{(a,b): (a,b,c_0) \in M \}.
	\end{align}
	We use Lemma \ref{lemma:2D_independence} to construct three two-dimensional polynomials $P_X,P_Y,P_Z$ that vanish on the relevant nonzero points in these directions, and set
	\begin{align}
		\Psi_p(X,Y,Z) = H_p(X,Y,Z) P_X(Y,Z) P_Y(X,Z) P_Z(X,Y).
	\end{align}
	Note that, for example, $P_X(Y,Z)$ is contained in $\mathbb{K}_p[X,Y,Z] \subseteq \mathbb{K}[X,Y,Z]$ although it does not depend on $X$. 
	By assumption, each slice has bounded size; hence
	\begin{align}
		& \deg_Y P_X, \ \deg_Z P_X, \ \deg_X P_Y, \ \deg_Z P_Y, \ \deg_X P_Z, \ \deg_Y P_Y \leq (\Delta-1) + D_2, \\
		& \deg_X H_p, \ \deg_Y H_p, \ \deg_Z H_p \leq D_3, \\
		\implies & \deg_X \Psi_p(X,Y,Z) \leq 2\left( (\Delta - 1) + D_2 \right) + D_3 \leq m-1. \label{eq:3D_independence}
	\end{align}	 
	
	Applying the same construction to every point of $M$ shows that the evaluation map
	\begin{align}
		\ev_M: \mathbb{K}^{m \times m \times m} \to \mathbb{K}^{M} \subseteq \mathbb{K}^{[n]^3}
	\end{align}
	is surjective. Taking the product of the nonzero minor polynomials in
	\begin{align}
		\mathbb{K}[\cdots,\tilde{X}_a,\cdots,\tilde{Y}_b,\cdots,\tilde{Z}_c,\cdots],
	\end{align}
	where $\tilde{X}_a,\tilde{Y}_b,\tilde{Z}_c$ are new indeterminates, we use the same argument as before to conclude the result over a sufficiently large field $\mathbb{F}_q$.
\end{proof}

We now consider the more difficult case in which $M$ is an $\epsilon$-closed set that may contain full lines. Let $M' = M \setminus L(M)$. Then $\Delta(M') = \Delta < \epsilon n$, $\Sigma(M') = \Sigma \leq \epsilon n^2$, and $\vert M' \vert \leq \epsilon n^3$. Suppose $D_2 =  \sqrt{2 \Sigma} $ and $D_3 = (6 \vert M' \vert)^{1/3}$ (omit taking integers as before); we strengthen \eqref{eq:3D_size} in Lemma \ref{lemma:3D_independence} to
\begin{align}\label{eq:epsilon_3D}
	2\left( \epsilon n +   \sqrt{2\epsilon}n  \right) + (6\epsilon)^{1/3} n + 2\left[ (\epsilon n - 1) +  \sqrt{2\epsilon} n  + \epsilon n \right] \leq m-1 = \nu n - 1,
\end{align}
which still admits a solution for sufficiently small $\epsilon$ depending on $\nu$ and $t = 3$. Note that the additional term in brackets is inherited from the 2D case in \eqref{eq:epsilon_2D}; it is used to control the degree when designing polynomials that vanish on full lines. However, unlike in the 2D case, Theorem \ref{thm:GRS} in 3D requires the stronger condition \eqref{eq:eplision_3D_condition}, while \eqref{eq:epsilon_3D} is sufficient to establish the following lemma.

\begin{lemma}\label{lemma:3D_independence_2}
	For any finite field $\mathbb{F}_q$ with $q > 2^{3 n^3}$, there exist evaluation sets $\{x_1,\ldots,x_n\}$, $\{y_1,\ldots,y_n\}$, and $\{z_1,\ldots,z_n\}$ such that, for any $\epsilon$-closed $M \subseteq [n]^3$ with $\epsilon$ satisfying \eqref{eq:epsilon_3D} and for any $(a,b,c) \in M' \subseteq M$, there is a polynomial whose evaluation vanishes on all of $M$, including the full lines, except at $(a,b,c)$.
\end{lemma}
\begin{proof}
	If $M = [n]^3$, then $M' = \emptyset$ and the statement holds trivially. For the given $\epsilon$, we need to choose evaluation sets satisfying the conclusions of Lemma \ref{lemma:2D_independence_2} and Lemma \ref{lemma:3D_independence} simultaneously. First, consider nonzero minor polynomials $P_{M_{3D}}$ for any $M_{3D} \subseteq [n]^3$ without full lines such that $\Delta(M_{3D})$, $\Sigma(M_{3D})$, and $\vert M_{3D} \vert$ are controlled by $\epsilon$. As in Lemma \ref{lemma:2D_independence_2}, each involved Vandermonde matrix has row size $2\left( (\epsilon n - 1) + \sqrt{2\epsilon}n \right) + (6\epsilon)^{1/3} n + 1$. As a reminder, the final codes always have code dimension $m$. We control the numbers of rows of these Vandermonde matrices in order to control the resulting polynomial degrees in \eqref{eq:Psi_deg_3D}.
	
	Next, we consider the nonzero minor polynomials $P_{X, M_{2D}}, P_{Y, M_{2D}}, P_{Z, M_{2D}}$ for the 2D case. As before, throughout the construction, the degrees of the indeterminates in the support polynomials cannot exceed $(\epsilon n - 1) +  \sqrt{2\epsilon} n + \epsilon n + 1$.
	It is also important to note that, by definition, no slice of any $M_{3D}$ can contain full lines. Therefore, each such slice is an $M_{2D}$ under consideration. Let
	\begin{align}
		\begin{aligned}
			& P(\cdots,\tilde{X}_a,\cdots,\tilde{Y}_b,\cdots,\tilde{Z}_c,\cdots) \\
			= & \prod_{M_{3D}} P_{M_{3D}}(\cdots,\tilde{X}_a,\cdots,\tilde{Y}_b,\cdots,\tilde{Z}_c,\cdots) \\
			& \cdot \prod_{M_{2D}} P_{X, M_{2D}} (\cdots,\tilde{Y}_b,\cdots,\tilde{Z}_c,\cdots) P_{Y, M_{2D}} (\cdots,\tilde{X}_a,\cdots,\tilde{Z}_c,\cdots) P_{Z, M_{2D}} (\cdots,\tilde{X}_a,\cdots,\tilde{Y}_b,\cdots)
		\end{aligned}
	\end{align}
	over $\mathbb{K}$. As before, this polynomial consists of monomials with coefficients equal to $1 \in \mathbb{K}$. A nonzero specialization exists whenever $q > 2^{n^3 + 3n^2}$. For $n \geq 2$, we use the larger but more compact condition $q > 2^{3n^3}$.
	
	Now fix an $\epsilon$-closed set $M$ and define $M'$ as above. Using the explicit evaluation sets $\{x_1,\ldots,x_n\}$, $\{y_1,\ldots,y_n\}$, and $\{z_1,\ldots,z_n\}$ obtained above, for any $p = (a_0,b_0,c_0) \in M'$, we can find a polynomial $\Psi_p(X,Y,Z)$ from Lemma \ref{lemma:3D_independence}, which is nonzero only at $p$ within $M$. Its degree satisfies
	\begin{align}\label{eq:Psi_deg_3D}
		\deg_X \Psi_p(X,Y,Z) \leq 2\left( (\epsilon n - 1) + \sqrt{2\epsilon}n \right) + (6\epsilon)^{1/3} n.
	\end{align}
	
	We consider the following slices in the three coordinate directions:
	\begin{align}
		& S_1 = \{(b,c): (a_0,b,c) \in M \setminus M' = L(M) \}, \\
		& S_2 = \{(a,c): (a,b_0,c) \in M \setminus M' = L(M) \}, \\
		& S_3 = \{(a,b): (a,b,c_0) \in M \setminus M' = L(M) \}
	\end{align}
	For instance, any $(b,c) \in S_1 \subseteq [n]^2$ with $(b,c) \neq (b_0,c_0)$ either supports a full line in the $x$-direction or belongs to a $y$- or $z$-line contained in $S_1$. Since $p$ does not lie on any full line, $S_1$ resembles Slice $A$ rather than Slice $B$ in Figure~\ref{fig:3D_full_lines}. By Lemma \ref{lemma:eplison_geo} below, $L(M)$ is $\epsilon$-closed, and $S_i$ is one of its cross sections; hence $S_i$ is also $\epsilon$-closed. We claim that $S_1 \cup \{(b_0,c_0)\}$ is also $\epsilon$-closed: for any line $\ell$, if $\ell \cap \{(b_0,c_0)\} = \emptyset$, then
	\begin{align}
		(S_1 \cup \{b_0,c_0\}) \cap \ell = S_1 \cap \ell
	\end{align}
	is either exactly $\ell$ or has size at most $\epsilon n$. If $\{(b_0,c_0)\} \subseteq \ell$, then, by definition, $\ell$ cannot be a full line in $L(M)$, and thus
	\begin{align}
		(S_1 \cup \{b_0,c_0\}) \cap \ell \subseteq M \cap \ell \implies \vert (S_1 \cup \{b_0,c_0\}) \cap \ell  \vert < \epsilon n.
	\end{align}
	Applying Lemma \ref{lemma:2D_independence_2} to $S_1 \cup \{(b_0,c_0)\}$ with the chosen evaluation sets, we can find a polynomial $\Phi_p(Y,Z)$ that is nonzero only at $(b_0,c_0)$ and whose degree satisfies (cf. \eqref{eq:2D_independence_2})
	\begin{align}\label{eq:Phi_deg_3D}
		\deg_Y \Phi_p(Y,Z) \leq (\epsilon n - 1) +  \sqrt{2\epsilon} n + \epsilon n. 
	\end{align}
	The same estimate holds for $\deg_Z \Phi_p(Y,Z)$, as well as for $\Phi_p(X,Z)$ and $\Phi_p(X,Y)$. The final polynomial is
	\begin{align}
		F_p(X,Y,Z) \coloneq \Psi_p(X,Y,Z) \Phi_p(Y,Z) \Phi_p(X,Z)  \Phi_p(X,Y)
	\end{align}
	with
	\begin{align}
		\begin{aligned}
			& \deg_X F_p, \ \deg_Y F_p, \ \deg_Z F_p  \\
			\leq & 
			2\left( (\epsilon n - 1) +  \sqrt{2\epsilon}n \right) + (6\epsilon)^{1/3} n + 
			2\left[ (\epsilon n - 1) +  \sqrt{2\epsilon} n + \epsilon n \right]
			\leq m - 1.
		\end{aligned}
	\end{align} 
	By definition, 
	\begin{enumerate}
		\item $F_p(x_{a_0},y_{b_0},z_{c_0}) \neq 0$.
		
		\item $F_p(x_a,y_b,z_c) \equiv 0$ for all other points in $M'$ since $\Psi_p(X,Y,Z)$ vanishes there.
		
		\item For each full line, e.g., in the $x$-direction, the line must pass through $S_1$ at some $(b',c')$, and $F_p(x_a,y_{b'},z_{c'}) \equiv 0$ on that line because $\Phi_p(y_{b'},z_{c'}) = 0$.
	\end{enumerate}
	This completes the proof.
\end{proof}

\begin{figure}[ht]
	\centering
	\includegraphics[width=0.8\textwidth]{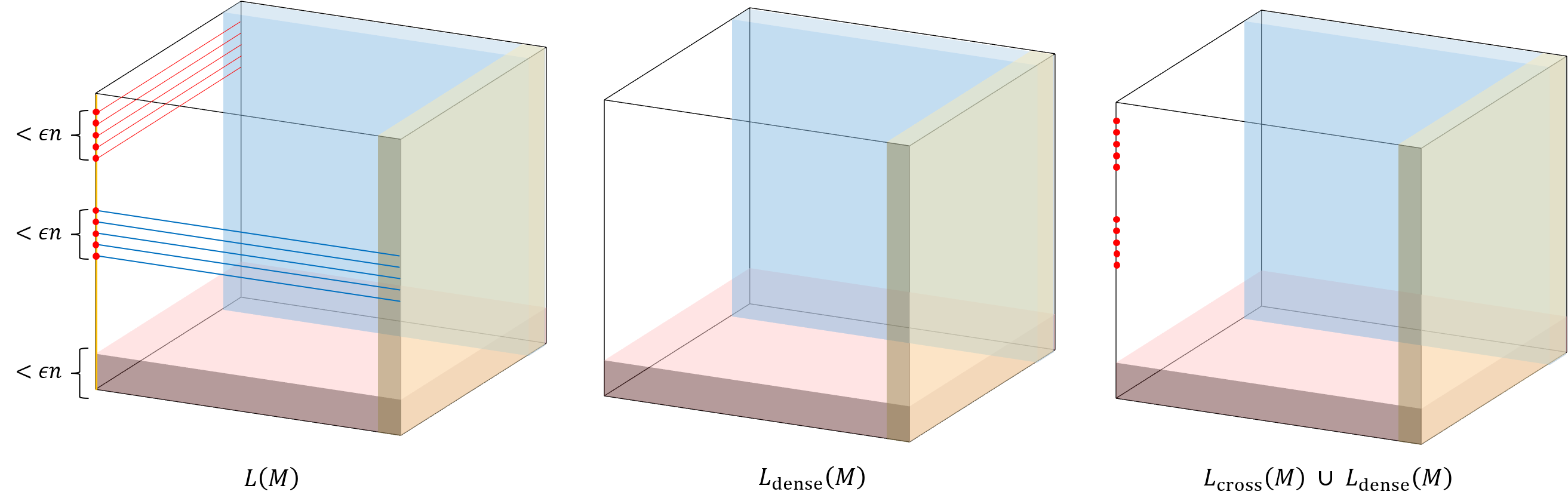}
	\caption{A sketch of the set of full lines $L(M)$, the dense set $L_{\text{dense}}(M)$, and the intersection set $L_{\text{cross}}(M)$ in 3D. The red, blue, and brown blocks represent dense collections of full lines, namely $I_1 \times [n] \times [n]$, $[n] \times I_2 \times [n]$, and $[n] \times [n] \times I_3$. Here $L(M)$ and $L_{\text{dense}}(M)$ are $\epsilon$-closed, while $L_{\text{cross}}(M) \cup L_{\text{dense}}(M)$ is $3\epsilon$-closed.}
	\label{fig:3D_dense_cross}
\end{figure}

The last ingredient for obtaining the complete extension is the ability to extend any local codeword defined on $L(M)$ (see Definition \ref{def:extendable}). Compared with the 2D case treated in Corollary \ref{coro:GRS_extendable_2D}, the problem is more intricate. Before the proof, we formally introduce several geometric concepts and properties of $\epsilon$-closed sets: let $M \subseteq [n]^3$ be $\epsilon$-closed, let $L(M)$ denote the subset of full lines, and set $M' = M \setminus L(M)$. Suppose $L(M) \neq [n]^3$. Let $I_1,I_2,I_3 \subseteq [n]$ be the maximal subset such that
\begin{align}
	I_1 \times [n] \times [n], \quad [n] \times I_2 \times [n], \quad [n] \times [n] \times I_3 \subseteq L(M).
\end{align}
Clearly, $I_i$ may be empty. If $I_i$ is nonempty, we say that there is a \emph{dense} hyperplane of full lines perpendicular to the $i$-th direction. We divide $L(M)$ into dense and sparse part (see also Figure~\ref{fig:3D_dense_cross}):
\begin{align}
	& L_{\text{dense}}(M) \coloneq (I_1 \times [n] \times [n]) \cup ([n] \times I_2 \times [n]) \cup ([n] \times [n] \times I_3), \label{eq:L_dense_3D} \\
	& L_{\text{sparse}}(M) \coloneq L(M) \setminus L_{\text{dense}}(M). \label{eq:L_sparse_3D}
\end{align}
However, within $L_{\text{sparse}}(M)$, there may be maximal subsets, e.g., $J_1, J_2 \subseteq [n]$ such that
\begin{align}
	J_1 \times J_2 \times [n] \subseteq L(M).
\end{align}
This set is regarded as sparse. Intuitively, such a set is a thin tube with $\vert J_i \vert < \epsilon n$, and it is straightforward to extend codewords from $\C_3$ along that direction. The following proof gives more details. 

Within $L(M)$, the \emph{sectional graph} of full lines in the $i$-th direction is defined by
\begin{align}
	& G_1 \coloneq \{(b,c): [n] \times \{b\} \times \{c\} \subseteq M \}, \\
	& G_2 \coloneq \{(a,c): \{a\} \times [n] \times \{c\} \subseteq M \},  \label{eq:G_2_def} \\
	& G_3 \coloneq \{(a,b): \{a\} \times \{b\} \times [n] \subseteq M \}.
\end{align}
Within $G_i \subseteq [n]^2$, we also consider $L(G_i)$ and distinguish its dense and sparse parts as before. We define, for example,
\begin{align}\label{eq:G_1_dense}
	L_{\text{dense}}(G_1) \coloneq (I_2 \times [n]) \cup ([n] \times I_3)
\end{align}
and $G_1 \setminus L_{\text{dense}}(G_1)$ consists only of sparse points (see also Figure \ref{fig:3D_full_lines}).

\begin{lemma}\label{lemma:eplison_geo}
	Let $M \subseteq [n]^3$ be $\epsilon$-closed. Then the following properties hold:
	\begin{enumerate}
		\item Any cross section of $M$ is $\epsilon$-closed in $2$ dimensions.
		
		\item The set of full lines $L(M)$ and $M' \coloneq M \setminus L(M)$ are $\epsilon$-closed, and so are their cross sections.
		
		\item Within $L(M)$, any sectional graph of full lines is $\epsilon$-closed.
		
		\item The set
		\begin{align}\label{eq:L_cross}
			L_{\text{cross}}(M) \coloneq \{ (a,b,c): \text{intersection points of two lines in } L_{\text{sparse}}(M) \}
		\end{align}
		is $2\epsilon$-closed. More than two lines may intersect at the same point, but such cases are encompassed by the definition above. Moreover, $L_{\text{cross}}(M) \cup L_{\text{dense}}(M)$ is $3\epsilon$-closed.
	\end{enumerate}
\end{lemma}
\begin{proof}
	The first property has been used above. Formally, let $S \subseteq M$ be any cross section, and let $\ell$ be any line in the corresponding hyperplane.
	\begin{align}
		\vert \ell \cap S \vert \geq \epsilon n \implies \vert \ell \cap M \vert \geq \epsilon n \implies \ell \subseteq M.
	\end{align}
	Since $\ell$ is in the hyperplane that contains $S$, $\ell \subseteq S$.
	
	Similarly,
	\begin{align}
		\vert \ell \cap L(M) \vert \geq \epsilon n \implies \vert \ell \cap M \vert \geq \epsilon n \implies \ell \subseteq M \implies \ell \subseteq L(M).
	\end{align}
	For $M'$, similarly,
	\begin{align}
		\vert \ell \cap M' \vert \geq \epsilon n \implies \vert \ell \cap M \vert \geq \epsilon n \implies \ell \subseteq M \implies \ell \nsubseteq M' = M \setminus L(M)
	\end{align}
	
	The sectional graph $G_i$ is contained in any cross section $S$ normal to the $i$-th direction. Then
	\begin{align}
		\vert \ell \cap G_i \vert \geq \epsilon n \implies \ell \subseteq S \subseteq M.
	\end{align}
	Thus $\ell$ lies in the dense part of $G_i$ as its intersection with the sparse part cannot $\geq \epsilon n$.
	
	It is important to note that not every subset of $M$ is $\epsilon$-closed. For the last case, suppose $\ell \cap L_{\text{cross}}(M) \neq \emptyset$. Since these intersection points lie on lines in sparse positions, in any fixed direction there can be at most $\epsilon n - 1$ parallel sparse lines, and the total number cannot exceed $2\epsilon n - 2 < 2\epsilon n$. For $L_{\text{cross}}(M) \cup L_{\text{dense}}(M)$, the line $\ell$ can intersect $L_{\text{dense}}(M)$ along its direction. Summing these contributions yields the result.
\end{proof}

These notions can be easily generalized for $M, L(M) \subsetneqq [n]^t$. A dense hyperplane of full lines perpendicular to the $i$-th direction is of the form:
\begin{align}
	[n] \times \cdots \times [n] \times I_i \times [n] \times \cdots \times [n] \subseteq M.
\end{align}
and there will be lower dimensional structures in the definition of dense sets, e.g.,
\begin{align}
	[n] \times \cdots \times J_i \times \cdots \times J_j \times \cdots \times [n] \subseteq L(M).
\end{align}
Sectional graphs and intersection sets can be defined accordingly, but for simplicity, we stick with the 3D case in the following proof.

\begin{figure}[ht]
	\centering
	\includegraphics[width=0.8\textwidth]{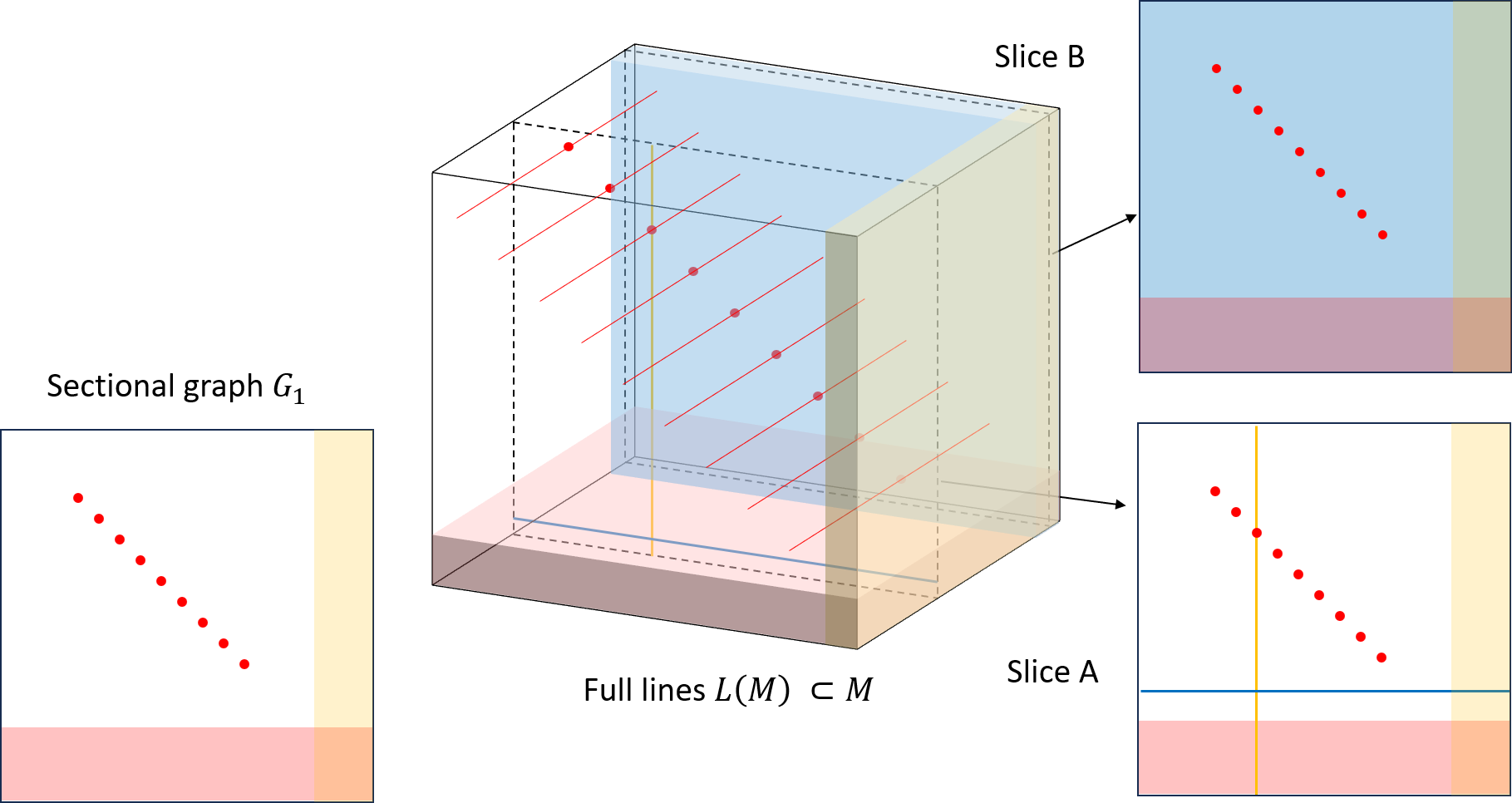}
	\caption{Another sketch of the set of full lines $L(M)$ in a set $M$. The red, blue, and brown blocks represent dense collections of full lines, namely $I_1 \times [n] \times [n]$, $[n] \times I_2 \times [n]$, and $[n] \times [n] \times I_3$. The red $x$-lines are in sparse locations, together with a single blue $y$-line and a brown $z$-line. The sectional graph $G_1$ consists of those $(b,c)$ that support an $x$-line. We define $G_2$ and $G_3$ similarly.}
	\label{fig:3D_full_lines}
\end{figure}

\begin{lemma}\label{lemma:3D_line_ext}
	Let $M \subseteq [n]^3$ be $\epsilon$-closed. For fixed $\nu \in (0,1)$ and sufficiently large $n$, we replace \eqref{eq:epsilon_3D} by the following more stringent condition: 
    \begin{align}\label{eq:eplision_3D_condition}
		23\epsilon + 4 \sqrt{6\epsilon}  + (18\epsilon)^{1/3} +  \sqrt{2\epsilon} \leq \nu.
	\end{align}
	Then, for any evaluation sets obtained in Lemma \ref{lemma:3D_independence_2} under \eqref{eq:eplision_3D_condition}, we can extend any local data (see Definition \ref{def:extendable}) from the full lines in $L(M) \subseteq M$ to a global tensor in $\C_1 \otimes \C_2 \otimes \C_3$.
\end{lemma}
\begin{proof} 
	In two dimensions, this result follows easily from Corollary \ref{coro:GRS_extendable_2D}. However, the problem suddenly becomes nontrivial in three dimensions. One example is provided by sparse lines as in \eqref{eq:diagonal_lines}. Since these lines are not localized in any smaller slices, the information-set technique fails. Nevertheless, we can solve the problem by combining Lemma \ref{lemma:2D_independence_2} and Lemma \ref{lemma:3D_independence_2} in a different way.
	
	\paragraph{Extension of the dense part.} 
	We first extend local codewords on the dense part $L_{\text{dense}}(M)$; the sparse part will be treated independently, in a manner similar to Lemma \ref{lemma:3D_independence_2}. By \eqref{eq:L_dense_3D}, the dense part consists of
	\begin{align}
		I_1 \times [n] \times [n], \ [n] \times I_2 \times [n], \ [n] \times [n] \times I_3 \subseteq M.
	\end{align}
	By definition, $\vert I_i \vert < \epsilon n$. For each direction and each $a,b,c$, define
	\begin{align}
		\phi_a(X) \coloneq \mathcal{N}_a \prod_{a' \neq a \in I_1}(X - x_{a'}), \quad
		\psi_b(Y) \coloneq \mathcal{N}_b \prod_{b' \neq b \in I_2}(Y - y_{b'}), \quad
		\varphi_c(Z) \coloneq \mathcal{N}_c \prod_{c' \neq c \in I_3}(Z - z_{c'}),
	\end{align}
	where $\mathcal{N}_i$ are normalization constants; for example, $\phi_a(x_a) = 1$, so the evaluation within $I_1$ is the delta function. The local codewords in the dense collection of lines are
	\begin{align}
		& A_a(Y,Z) \in \C_2 \otimes \C_3, \quad a \in I_1, \quad \deg_Y A_a, \ \deg_Z A_a \leq m-1 \\
		& B_b(X,Z) \in \C_1 \otimes \C_3, \quad b \in I_2, \quad \deg_X B_b, \ \deg_Z B_b \leq m-1 \\  
		& C_c(X,Y) \in \C_1 \otimes \C_2, \quad c \in I_3, \quad \deg_X C_c, \ \deg_Y C_c \leq m-1. 
	\end{align} 
	These local codewords agree on pairwise intersections:
	\begin{align}
		& A_a(y_b, Z) = B_b(x_a,Z), \quad a \in I_1, b \in I_2, \\
		& A_a(Y, z_c) = C_c(x_a,Y), \quad a \in I_1, c \in I_3, \\
		& B_b(X, z_c) = C_c(X,y_b), \quad b \in I_2, c \in I_3.
	\end{align}
	We consider the interpolation:
	\begin{align}
		& F_{\text{dense}}(X,Y,Z) \notag \\
		= & \sum_{a \in I_1} \phi_a(X) A_a(Y,Z) + \sum_{b \in I_2} \psi_b(Y) B_b(X,Z) + \sum_{c \in I_3} \varphi_c(Z) C_c(X,Y) \\
		& - \sum_{a \in I_1, b \in I_2} \phi_a(X) \psi_b(Y) A_a(y_b,Z) 
		- \sum_{a \in I_1, c \in I_3} \phi_a(X) \varphi_c(Z) A_a(Y,z_c) 
		- \sum_{b \in I_2, c \in I_3} \psi_b(Y) \varphi_c(Z) B_b(X,z_c) \notag \\
		& + \sum_{a \in I_1, b \in I_2, c \in I_3} \phi_a(X) \psi_b(Y) \varphi_c(Z) A_a(y_b,z_c). \notag
	\end{align}
	A degree check shows that $F_{\text{dense}}(X,Y,Z) \in \C_1 \otimes \C_2 \otimes \C_3$. It also matches the local data; for example, for any $z$-line $\{a'\} \times \{b'\} \times [n]$,
	\begin{align}
		& F_{\text{dense}}(x_{a'},x_{b'},Z) \notag \\
		= & A_{a'}(y_{b'},Z) + B_{b'}(x_{a'},Z) + \sum_{c \in I_3} \varphi_c(Z) C_c(x_{a'},x_{b'}) \\
		& - A_{a'}(y_{b'},Z) 
		- \sum_{c \in I_3} \varphi_c(Z) A_{a'}(x_{b'},z_c) 
		- \sum_{c \in I_3} \varphi_c(Z) B_{b'}(x_{a'},z_c) \notag \\
		& + \sum_{c \in I_3} \varphi_c(Z) A_{a'}(y_{b'},z_c) = B_{b'}(x_{a'},Z)  \notag
	\end{align}
	by the agreement on intersections: $B_{b'}(x_{a'},z_c) = C_c(x_{a'},x_{b'})$.
	
	\paragraph{Peeling off high-order powers.}
	The next step is to extend local data in sparse positions. However, local data may correspond to a high-degree polynomial, and unlike in the previous case, we cannot simply multiply it by polynomials in other variables. To address this issue, we peel off high-order powers. Consider the sectional graphs of full lines in $L(M) \subseteq M$ in 3D. By~\eqref{eq:G_1_dense}, 
	\begin{align}
		G_1 & = \{ (b,c) \text{ in sparse positions} \} \cup (I_2 \times [n]) \cup ([n] \times I_3) \\
		G_2 & = \{ (a,c) \text{ in sparse positions} \} \cup (I_1 \times [n]) \cup ([n] \times I_3) \label{eq:G_2_eq} \\
		G_3 & = \{ (a,b) \text{ in sparse positions} \} \cup (I_1 \times [n]) \cup ([n] \times I_2).
	\end{align}
	For any $(b,c) \in G_1 \setminus L_{\text{dense}}(G_1)$ in the sparse region, the updated local data for the $x$-line $[n] \times \{b\} \times \{c\}$ satisfies
	\begin{align}
		& g_{(b,c)}(X) = \text{ original local data } + F_{\text{dense}}(X,y_b,z_c), \\
		& g_{(b,c)}(x_a) = 0 \text{ if } (a,b,c) \in I_1 \times [n] \times [n] \ \text{ due to agreement on the local intersections} \\
		\implies & g_{(b,c)}(X) = \left( \prod_{a \in I_1 } (X - x_a) \right) \cdot \left( \sum_{i \leq m-1 - \vert I_i \vert} g_{i,(b,c)} X^i \right).
	\end{align}
	Let $r = 3\left[ (\epsilon n - 1) + \sqrt{2\epsilon} n + \epsilon n \right]$. Again, $r$ should be an integer, but we omit the details. We decompose the polynomial as
	\begin{align}
		\begin{aligned}
			g_{(b,c)}(X) = & \left( \prod_{a \in I_1 } (X - x_a) \right) \cdot \left( \sum_{i \leq m-r-1} g_{i,(b,c)} X^i \right) \\
			& + \left( \prod_{a \in I_1 } (X - x_a) \right) \cdot \left( \sum_{m-r \leq i \leq m-1 - \vert I_i \vert} g_{i,(b,c)} X^i \right).  
		\end{aligned}
	\end{align}
	Since $\vert I_i \vert < \epsilon n$, the possible highest degree of the first part is $m-r-1 + \epsilon n$, which is clearly $< m-1$ by the definition of $r$. Both parts vanish on $I_1 \times [n] \times [n]$, and note that $b \notin I_2$ and $c \notin I_3$.
	
	By Lemma \ref{lemma:eplison_geo}, $G_1$ is $\epsilon$-closed. For each $i = m-r,\ldots,m-1$, Lemma \ref{lemma:2D_independence_2} provides an interpolation polynomial $L_i(Y,Z)$ such that $L_i(b,c) = g_{i,(b,c)}$ in the sparse region of $G_1$ and $L_i$ vanishes on the dense region. Its degrees satisfy
	\begin{align}
		\deg_Y L_i(Y,Z), \ \deg_Z L_i(Y,Z) \leq (\epsilon n - 1) + \sqrt{2\epsilon} n 
	\end{align}
	Then we obtain
	\begin{align}
		\left( \prod_{ a \in I_1 } (X - x_a) \right) \cdot \left( \sum_{i = m-r}^{m-1} L_i(Y, Z) X^i \right)
	\end{align}
	with the following properties:
	\begin{enumerate}
		\item For any $(b,c) \in G_1$ in the sparse region, it matches the high-order powers of $g_{(b,c)}(X)$.
		
		\item If either $b \in I_2$ or $c \in I_3$, then $L_i(y_b, z_c) = 0$. If $X = x_a$ with $a \in I_1$, the polynomial also vanishes. Therefore, it vanishes on $L_{\text{dense}}(M)$.
		
		\item Its value may, however, affect the local data on $y$- and $z$-lines at sparse locations. The crucial point is that it does not introduce high-order powers into the local data $g_{(a,c)}(Y)$ and $g_{(a,b)}(Z)$. To guarantee this, we require
		\begin{align}\label{eq:eplision_3D_condition1}
			(\epsilon n - 1) + \sqrt{2\epsilon} n  \leq m-r-1 + \epsilon n
			\implies 4\left[ (\epsilon n - 1) +  \sqrt{2\epsilon} n + \epsilon n \right] -\epsilon n \leq m - 1,
		\end{align}
		which can be achieved for sufficiently small $\epsilon$.
	\end{enumerate}
	Applying the same peeling-off method to $y$- and $z$-lines also does not reintroduce high-order powers into $x$-lines. The whole process yields the following polynomial, with an associated global tensor codeword:
	\begin{align}\label{eq:3D_line_ext}
		\begin{aligned}
			F_{\text{peel}}(X,Y,Z) = & \left( \prod_{a \in I_1 } (X - x_a) \right) \cdot \sum_{i = m-r}^{m-1} L_i(Y, Z) X^i 
			+ \left( \prod_{b \in I_2} (Y - y_b) \right) \cdot \sum_{j = m-r}^{m-1} L_j(X, Z) Y^j \\
			& + \left( \prod_{c \in I_3} (Z - z_c) \right) \cdot \sum_{k = m-r}^{m-1} L_k(X, Y) Z^k.
		\end{aligned}
	\end{align}
	As a reminder, $L_j(X, Z)$ is extended exactly from $g_{j,(a,c)}$ with $j \geq m-r$, since the influence of peeling off in the $X$-direction only introduces lower-order terms. The same applies to $L_k(X, Y)$. In other words, although the terms in Eq.~\eqref{eq:3D_line_ext} interact along lines in different directions, no additional high-order powers are introduced, and $F_{\text{peel}}(X,Y,Z)$ still vanishes $L_{\text{dense}}(M)$.
	
	\paragraph{Intersections of lines in sparse positions.}
	One more step is needed before extending the local line data from the sparse part. We copy~\eqref{eq:L_cross} here:
	\begin{align}
		L_{\text{cross}}(M) \coloneq \{ (a,b,c): \text{intersection points of two lines in } L_{\text{sparse}}(M) \}.
	\end{align}
	A line may be isolated from every other line in the sparse part (although it must cross the normal dense part). In that case, none of its points belongs to $L_{\text{cross}}(M)$. By Lemma \ref{lemma:eplison_geo}, $L_{\text{cross}}(M)$ is $2\epsilon$-closed and $L_{\text{cross}}(M) \cup L_{\text{dense}}(M)$ is $3\epsilon$-closed.
	
	For each $(a,b,c) \in L_{\text{cross}}(M)$, we apply Lemma \ref{lemma:2D_independence_2} again to find a global polynomial whose evaluation is nonzero only at $(a,b,c)$ and vanishes elsewhere on $L_{\text{cross}}(M) \cup L_{\text{dense}}(M)$. We then rescale it so that its value equals
	\begin{align}\label{eq:def_cross}
		\text{original local data at } (a,b,c) + F_{\text{dense}}(x_a,y_b,z_c) +  F_{\text{peel}}(x_a,y_b,z_c) 
	\end{align}
	Doing this for every point in $L_{\text{cross}}(M)$ and taking the sum, we obtain $F_{\text{cross}}(X,Y,Z)$ with degrees
	\begin{align}\label{eq:deg_cross}
		\begin{aligned}
			& \deg_X F_{\text{cross}}, \ \deg_Y F_{\text{cross}}, \ \deg_Z F_{\text{cross}} \\
			\leq & 2\left( ((3\epsilon) n - 1) + \sqrt{2(3\epsilon)}n  \right) + (6 (3\epsilon))^{1/3} n + 
			2\left[ ((3\epsilon) n - 1) + \sqrt{2 (3\epsilon)} n + (3\epsilon) n \right].  
		\end{aligned}
	\end{align}  
	We now have
	\begin{align}
		F_{\text{dense}}(X,Y,Z) + F_{\text{peel}}(X,Y,Z) + F_{\text{cross}}(X,Y,Z)
	\end{align}
	such that
	\begin{enumerate}
		\item It agrees exactly with the original data on $L_{\text{dense}}(M)$ because $F_{\text{peel}}(X,Y,Z)$ and $F_{\text{cross}}(X,Y,Z)$ vanish there.
		
		\item For each line in a sparse position, for example $[n] \times \{b\} \times \{c\}$,
		\begin{align}\label{eq:local_data_sparse}
			\tilde{g}_{(b,c)}(X) = \text{ original local data } + F_{\text{dense}}(X,y_b,z_c) +  F_{\text{peel}}(X,y_b,z_c) + F_{\text{cross}}(X,y_b,z_c)
		\end{align}
		is of degree
		\begin{align}\label{eq:local_data_deg}
			\leq m-r-1 + \epsilon n = m - 3\left[ (\epsilon n - 1) + \sqrt{2\epsilon} n + \epsilon n \right] -1 + \epsilon n.
		\end{align}
		This is because, for sufficiently small $\epsilon$, \eqref{eq:deg_cross} cannot contribute high-order terms:
		\begin{align}
			& (18\epsilon) n - 4 + 4 \sqrt{6\epsilon}n  + (18\epsilon)^{1/3} n 
			\leq m-r-1 + \epsilon n \\ 
			\implies & 
			23\epsilon n - 7 + 4 \sqrt{6\epsilon}n + (18\epsilon)^{1/3} n +  \sqrt{2\epsilon} n \leq m - 1 \label{eq:eplision_3D_condition2}.
		\end{align}
		The evaluation of $\tilde{g}_{(b,c)}(X)$ vanishes on $L_{\text{dense}}(M)$, and by \eqref{eq:def_cross}, it also vanishes at every intersection point $(a,b,c) \in L_{\text{cross}}(M)$.
	\end{enumerate}
	
	\paragraph{Extension of the sparse part.}
	Now let $p = (b_0,c_0) \in G_1 \setminus L_{\text{dense}}(G_1)$. We extend $\tilde{g}_p(X)$ from \eqref{eq:local_data_sparse}. Applying Lemma \ref{lemma:2D_independence_2} to $G_1$, we can find a polynomial $G_p(Y,Z)$ that is nonzero at $p$ within $G_1$. We normalize it so that $G_p(y_{b_0}, z_{c_0}) = 1$ and
	\begin{align}
		\deg_Y G_p(Y,Z), \ \deg_Z G_p(Y,Z) \leq  (\epsilon n - 1) + \sqrt{2\epsilon} n + \epsilon n. 
	\end{align} 
	Then we consider the union $G_2 \cup ([n] \times \{c_0\})$. By an argument similar to that in Lemma \ref{lemma:eplison_geo}, it is $(\epsilon + \frac{1}{n})$-closed.
	As mentioned, $[n] \times \{b_0\} \times \{c_0\}$ intersects the dense part $I_1 \times [n] \times [n]$ unless $I_1 = \emptyset$. It may also intersect other $y$- and $z$-lines at sparse locations. Nevertheless, by Lemma \ref{lemma:2D_independence_2}, for any $(a,c) \in G_2$ outside $[n] \times \{c_0\}$ and the dense region, we can find a polynomial $f_{(a,c)}(X,Z)$ that vanishes everywhere, including on $[n] \times \{c_0\}$, except at $(a,c)$. We normalize it so that $f_{(a,c)}(x_a, z_c) = 1$. Its degrees are 
	\begin{align}
		\deg_X f_{(a,c)}(X,Z), \ \deg_Z f_{(a,c)}(X,Z) \leq \left( \left(\epsilon + \frac{1}{n} \right)   n - 1 \right) +  \sqrt{2\left(\epsilon + \frac{1}{n} \right)  } n  + \left(\epsilon + \frac{1}{n} \right) n. 
	\end{align} 
	The key trick is to consider the following polynomial, where $\mathrm{1}$ denotes the constant polynomial:
	\begin{align}
		\tilde{f}_Y(X,Z) = \mathrm{1} - \sum_{(a,c) \in G_2 \setminus (\text{dense} \cup  [n] \times \{c_0\} )} f_{(a,c)}(X,Z).
	\end{align}
	It gives the all-ones vector along $[n] \times \{c_0\}$ and vanishes at the other points of $G_2 \cup ([n] \times \{c_0\})$. The summation does not increase the polynomial degree, and we define $\tilde{f}_Z(X,Y)$ based on $G_3 \cup ([n] \times \{b_0\})$ in the same way. The product $\tilde{g}_p(X) G_p(Y,Z) \tilde{f}_Y(X,Z) \tilde{f}_Z(X,Y)$ satisfies the following properties:
	\begin{enumerate}
		\item The values are identical to $\tilde{g}_p(X)$ along $[n] \times \{b_0\} \times \{c_0\}$. 
		
		\item The values are zero on other $x$-lines $[n] \times \{b\} \times \{c\}$ due to the vanishing of $G_p(Y,Z)$.
		
		\item For lines in other directions, such as $\{a\} \times [n] \times \{c\}$, if the line does not intersect $[n] \times \{b_0\} \times \{c_0\}$, then $\tilde{f}_Y(x_a, z_c) = 0$ annihilates the value.
		
		\item For any possible intersection $(a,b_0,c_0)$ in the dense or sparse part, $\tilde{g}_p(x_a) = 0$ implies that the product vanishes on that line, regardless of whether $\tilde{f}_Y(X,Z)$ or $\tilde{f}_Z(X,Y)$ is zero.
		
		\item It vanishes on $L_{\text{dense}}(M)$ because at least one of $G_p(Y,Z)$, $\tilde{f}_Y(X,Z)$, or $\tilde{f}_Z(X,Y)$ vanishes there.
	\end{enumerate}
	By \eqref{eq:local_data_deg}, the total $X$ degree equals
	\begin{align}
		\begin{aligned}
			& \deg \tilde{g}_p(X) + 2\left[ \left( \left(\epsilon + \frac{1}{n} \right)   n - 1 \right) +  \sqrt{2\left(\epsilon + \frac{1}{n} \right)  } n + \left(\epsilon + \frac{1}{n} \right)  n \right]  \\
			\leq & m - 1 - 3\left[ (3\epsilon n - 1) +  \sqrt{2\epsilon} n  \right] + \epsilon n
			+ 2\left[ \left( 2\left(\epsilon + \frac{1}{n} \right)  n - 1 \right) +  \sqrt{2\left(\epsilon + \frac{1}{n} \right)  } n \right] 
		\end{aligned}
	\end{align}
	which is smaller than $m$ for sufficiently large $n$. For the degree in $Y$ (and similarly in $Z$),
	\begin{align}
		\begin{aligned}
			& \deg_Y \left( \tilde{g}_p(X) G_p(Y,Z) \tilde{f}_Y(X,Z) \tilde{f}_Z(X,Y) \right) \\
			\leq  & (2\epsilon n - 1) +  \sqrt{2\epsilon} n  + \left[ \left( 2\left(\epsilon + \frac{1}{n} \right)   n - 1 \right) + \sqrt{2\left(\epsilon + \frac{1}{n} \right)  } n \right] \leq m -1 - r < m-1
		\end{aligned}
	\end{align}
	by the more stringent inequality \eqref{eq:eplision_3D_condition2}.
	
	We relabel $\tilde{f}_{Y}(X,Z)$ and $\tilde{f}_Z(X,Y)$ as $\tilde{f}_{1,Y}(X,Z)$ and $\tilde{f}_{1,Z}(X,Y)$, respectively, and apply the same strategy to the $y$- and $z$-lines. Let
	\begin{align}
		\begin{aligned}
			F_{\text{sparse}}(X,Y,Z) \coloneq
			& \sum_{p \in G_1 \setminus L_{\text{dense}}(G_1)} \tilde{g}_p(X) G_p(Y,Z) \tilde{f}_{1,Y}(X,Z) \tilde{f}_{1,Z}(X,Y) \\
			& + \sum_{p \in G_2 \setminus L_{\text{dense}}(G_2)} \tilde{g}_p(Y) G_p(X,Z) \tilde{f}_{2,Z}(X,Y) \tilde{f}_{2,X}(Y,Z) \\
			& + \sum_{p \in G_3 \setminus L_{\text{dense}}(G_3)} \tilde{g}_p(Z) G_p(X,Y) \tilde{f}_{3,Z}(X,Y) \tilde{f}_{3,X}(Y,Z).
		\end{aligned}
	\end{align}
	Then,
	\begin{align}
		F_{\text{dense}}(X,Y,Z) + F_{\text{peel}}(X,Y,Z) + F_{\text{cross}}(X,Y,Z) + F_{\text{sparse}}(X,Y,Z)
	\end{align}
	is the desired global extension.
\end{proof}

\begin{corollary}\label{coro:GRS_extendable}
	Any $\epsilon$-closed $M$ with $\epsilon$ satisfying \eqref{eq:eplision_3D_condition} is extendable for the punctured RS codes $\C_1,\C_2,\C_3$ obtained from the evaluation sets in Lemma \ref{lemma:3D_independence_2}. As a result,
	\begin{align}
		\rho(\C_1^\perp,\C_2^\perp,\C_3^\perp) \geq \frac{\epsilon^3}{3(2^3+1)^3},
	\end{align}
	where $\epsilon$ depends only on $\nu$ and $t = 3$.
\end{corollary}
\begin{proof}
	Let $M$ be any $\epsilon$-closed set with $\epsilon$ satisfying \eqref{eq:eplision_3D_condition}. Given any local data on sparse points and full lines in $M$, we first extend the full lines by Lemma \ref{lemma:3D_line_ext}. This may change the residual data on sparse points, but each such residual value can be reassigned independently by Lemma \ref{lemma:3D_independence_2}. One crucial observation is that, when choosing the polynomials in Lemma \ref{lemma:3D_independence_2}, codewords on $L_{\text{dense}}(M)$ remain intact. This establishes the global extension, and the product expansion follows from Lemma \ref{lemma:inner-g/rho}.
\end{proof}

It is straightforward to generalize these results to higher dimensions by induction. In retrospect, when $q > 2^{t n^t}$, the Schwartz--Zippel lemma shows that the probability of obtaining favorable evaluation sets is greater than $1 - 2^{t n^t - \log q}$. We next show that the dual codes automatically satisfy all desired properties when $\nu_i \leq \frac{1}{2}$, thereby proving Theorem \ref{thm:GRS}.


\subsection{Extendability of dual codes}

We now prove that $\rho(\C_1^\times,\ldots,\C_t^\times) > c(\nu_i,t)$ when involving dual codes.

\begin{lemma}[Diagonal scaling preserves extendability]
	\label{lem:diagonal-scaling-preserves-extendability}
	Let $\C_i \subseteq \mathbb F_q^n$ be classical linear codes for $i \in[t]$, let $D_i = \text{diag}(\lambda_{i,1},\ldots,\lambda_{i,n})$ be invertible diagonal matrices, and set $\C_i' \coloneq D_i \C_i$. Let $\C \coloneq \bigotimes_{i\in[t]} \C_i \subseteq \mathbb F_q^{[n]^t}$ and $\C' \coloneq \bigotimes_{i\in[t]} \C_i' \subseteq \mathbb F_q^{[n]^t}$. Then, for every subset $M\subseteq[n]^t$, $M$ is extendable with respect to $\C$ if and only if $M$ is extendable with respect to $\C'$.
\end{lemma}
\begin{proof}
	We define the diagonal linear map by
	\begin{equation}
		\label{eq:global-diagonal-map}
		T:=D_1\otimes\cdots\otimes D_t:
		\mathbb F_q^{[n]^t}
		\longrightarrow
		\mathbb F_q^{[n]^t}.
	\end{equation}
	Equivalently, for $x\in\mathbb F_q^{[n]^t}$, one has
	\begin{equation}
		\label{eq:global-diagonal-map-coordinate}
		(Tx)_{i_1,\ldots,i_t}
		=
		\left(\prod_{j=1}^t \lambda_{j,i_j}\right)
		x_{i_1,\ldots,i_t}.
	\end{equation}
	Since each $D_i$ is invertible, $T$ is invertible. For $M\subseteq[n]^t$, the restriction $T|_M$ is also invertible. Therefore, if $M$ is extendable with respect to $\C$, then for each local codeword $c_M'\in \C'$, there exists a global codeword $c \in \C$ such that $c|_M=(T|_M)^{-1}c_M'$. Then $(Tc)|_M=T|_M c|_M=c_M'$. Hence $M$ is also extendable with respect to $\C'$. The converse follows by applying the same argument to the inverse diagonal matrices $D_i^{-1}$.
\end{proof}

\begin{corollary}\label{coro:RS_dual}
	Suppose $\nu_i \leq \frac{1}{2}$. If every $\epsilon$-closed set $M$ is extendable for the tensor product $\C_1 \otimes \cdots \otimes \C_t$ of punctured RS codes, then it is also extendable for $\C_1^\perp \otimes \cdots \otimes \C_t^\perp$. Moreover, the same extendability holds for any combination of the codes, for example $\C_1,\C_2^\perp,\ldots,\C_t^\perp$ or $\C_1,\C_2,\ldots,\C_t^\perp$. As a result, $(\C_1^\times,\ldots,\C_t^\times)$, where each $\C_i^\times$ denotes either the code or its dual, is product-expanding with at least the same $\rho$ as $(\C_1^\perp,\ldots,\C_t^\perp)$.
\end{corollary}
\begin{proof}
	Each dual code of $\C_i$ is a GRS code with generator matrix
	\begin{align}\label{eq:C1_dual}
		\begin{pmatrix}
			1      & 1      & \dots  & 1      \\
			x_1 & x_2 & \dots  & x_n  \\
			x_1^2  & x_2^2  & \dots  & x_n^2  \\
			\vdots & \vdots & \ddots & \vdots \\
			x_1^{n-m_i-1} & x_2^{n-m_i-1} & \dots & x_n^{n-m_i-1}
		\end{pmatrix} D, \qquad  D_i = \text{diag}(\lambda_{i,1},\ldots,\lambda_{i,n}),
	\end{align}
	where the $\lambda_{i,a}$ are nonzero Lagrange multipliers.
	
	Since $\nu_i \leq \frac{1}{2}$, we have $m_i \leq n-m_i$. Consider the punctured RS codes generated by Vandermonde matrices with $n-m_i$ rows on the same evaluation sets. Since each of these codes contains the corresponding $\C_i$, all previous proofs of surjectivity of evaluation maps and independence of tensors apply immediately. For instance, in Lemma \ref{lemma:3D_independence_2} and Lemma \ref{lemma:3D_line_ext}, given any $\epsilon$-closed $M$ and any $p = (a_0,b_0,c_0) \in M'$, we build the same polynomial as in the case of $\C_i$. Applying Lemma \ref{lem:diagonal-scaling-preserves-extendability} to the Lagrange multipliers completes the proof.
\end{proof}

\begin{corollary}\label{coro:GRS_expanding}
	Suppose $\C_1,\ldots,\C_t$ are punctured RS codes with the extendability property in Corollary \ref{coro:GRS_extendable}, and suppose that for each $i = 1,\ldots,t$, $m_i = \nu_i n \leq \frac{1}{2r} n$. Then, for the $r$-fold Schur products (see below Definition \ref{def:KR_product}),
	\begin{align}
		\rho\left( (\C_1^{\ast r})^\times, \ldots, (\C_t^{\ast r})^\times \right) > c(\nu_i,t)
	\end{align}  
	where $(\C_i^{\ast r})^\times$ denotes either the product code $\C_i^{\ast r}$ or its dual. Moreover, 
	\begin{align}
		\rho\left( \mathcal{D}_1, \ldots, \mathcal{D}_t \right) > c(\nu_i,t)
	\end{align}  
	where, for each $i =1,\ldots,t$, $\mathcal{D}_i$ can be freely taken among $\C_i$, $\C_i^\perp$, $\C_i^{\ast r}$, and $(\C_i^{\ast r})^\perp$.
\end{corollary}
\begin{proof}
	Evidently, the rate of any $r$-fold product $\C_i^{\ast r}$ is larger than that of $\C_i$, so extendability can be proved as before. Given $m_i = \nu_i n \leq \frac{1}{2r} n$, the rate of $\C_i^{\ast r}$ is at most $\frac{1}{2}$. The extendability of tensor products with dual codes then follows from Lemma \ref{coro:RS_dual}. If $\mathcal{D}_i$ is one of $\C_i$, $\C_i^\perp$, $\C_i^{\ast r}$, and $(\C_i^{\ast r})^\perp$, then, since $\C_i$ has the smallest rate among these codes, extendability can be proved for any combination of them.
\end{proof}


\section{Transversal non-Clifford gates on almost-good qLDPC codes and qLTCs}

We are now going to prove that the cup and cap products defined in Section~\ref{sec:cup and cap on sheaf} combining with the GRS codes in Section~\ref{sec:GRS} induce nontrivial multi-controlled-$Z$ gates on almost-good qLDPC codes and qLTCs. The key observation underlying our proof is that the almost-good qLTCs constructed in Ref.~\cite{Dinur2024sheaf} are closely related to a class of sheaved HGP codes~\cite{LSWLL2026Theory}. To be precise, the base spaces and sheaves of these codes arise as covering spaces and pullback sheaves of the corresponding HGP codes, on which it is straightforward to compute cup and cap products. Furthermore, the two-way product-expanding punctured RS codes can help preserve code parameters ~\cite{Dinur2024sheaf,KP2025Extendable} and ensure nontrivial cup and cap products~\cite{LSWLL2026Theory}. We now elucidate the construction, thereby establishing transversal non-Clifford gates on almost-good qLDPC codes and qLTCs. This answers Conjecture~1.2 in Ref.~\cite{Li2025Poincare} in the affirmative.

\subsection{Covering spaces and covering maps}

As preparation for inducing cup and cap products from HGP codes to sheaf codes, we formally define the following cellular notion of covering spaces and introduce various key properties. This also generalizes the induction scheme for qLDPC codes on cubical complexes to arbitrary cell complexes~\cite{LSWLL2026Theory}.

\begin{definition}[Covering map]
Let $\widetilde X$ and $X$ be topological spaces, and let $P:\widetilde X\to X$ be a continuous surjection. We say that $P$ is a \emph{covering map} and $\widetilde X$ is the \emph{covering space} of $X$, if for every $x\in X$, there exists a closed set $F\subseteq X$ containing $x$ such that
\begin{align}
    P^{-1}(F)=\bigsqcup_{\tilde x\in P^{-1}(x)} F_{\tilde x},
\end{align}
where each $F_{\tilde x}\subseteq \widetilde X$ is closed, $\tilde x\in F_{\tilde x}$ and the restriction
\begin{align}
P|_{F_{\tilde x}}:F_{\tilde x}\to F    
\end{align}
is a homeomorphism. We say that $P$ is an $\ell$-sheeted covering if for every $x\in X$, $|P^{-1}(x)|=\ell$. 
\end{definition}

We are particularly interested in the case when $X$ and $\widetilde X$ are finite cell posets with Alexandrov topology. Then for each $\sigma\in X$, the closure of $\sigma$ is exactly $X_{\le\sigma}$. Therefore, the map $P:\widetilde X\to X$ is a covering map if and only if it is an order-preserving surjection such that for each $\sigma\in X$,
\begin{align}
    P^{-1}(X_{\le\sigma})=\bigsqcup_{\tilde\sigma\in P^{-1}(\sigma)}\widetilde X_{\le\tilde\sigma}
\end{align}
and for every $\tilde\sigma\in P^{-1}(\sigma)$, the restricted map
\begin{align}
    P|_{\widetilde X_{\le\tilde\sigma}}:\widetilde X_{\le\tilde\sigma}\to X_{\le \sigma}
\end{align}
is an isomorphism of posets.

Since the covering map $P:\widetilde X\to X$ is an order-preserving map, $\sd P:\sd \widetilde X\to \sd X$ is well-defined. Then we have the following proposition: 

\begin{proposition}
    Let $P:\widetilde X\to X$ be a covering map of cell posets. Then
    \begin{align}
        \sd P:\sd \widetilde X\to \sd X
    \end{align}
    is also a covering map.
\end{proposition}
\begin{proof}
    Let $\rho=[\sigma_0,\cdots,\sigma_k]\in \sd X$. By definition,
    \begin{align}
        (\sd X)_{\le\rho}=\{[\sigma_{i_0},\cdots,\sigma_{i_m}]:0\le i_0<\cdots<i_m\le k, m\le k\},
    \end{align}
    Then we have
    \begin{align}
        &(\sd P)^{-1}( (\sd X)_{\le\rho})=\{[\tilde\sigma_{i_0},\cdots,\tilde\sigma_{i_m}]:0\le i_0<\cdots<i_m\le k, m\le k, P(\tilde\sigma_{i_j})=\sigma_{i_j},\, 0\le j\le m\} \notag \\
        &=\bigsqcup_{\tilde\sigma_k\in P^{-1}(\sigma_k)}\{[\tilde\sigma_{i_0},\cdots,\tilde\sigma_{i_m}]:0\le i_0<\cdots<i_m\le k, m\le k, P(\tilde\sigma_{i_j})=\sigma_{i_j},\, 0\le j\le m,\tilde\sigma_{i_m}\le\tilde\sigma_k\} \notag \\
        &=\bigsqcup_{\tilde\sigma_k\in P^{-1}(\sigma_k)}\{\tilde\rho'\in \sd \tilde X:\tilde\rho'\le\tilde\rho, s(\tilde\rho)=\sigma_k, \tilde\rho\in (\sd P)^{-1}(\rho)\}\\
        &=\bigsqcup_{\tilde\rho\in (\sd P)^{-1}(\rho)}(\sd \widetilde X)_{\le \tilde\rho}\ , \notag
    \end{align}
   where the last equality holds because for each $\tilde\rho=[\tilde\sigma_0,\cdots,\tilde\sigma_k]\in (\sd P)^{-1}(\rho)$, $\tilde\rho$ is uniquely determined by $\tilde\sigma_k$ since $\tilde\sigma_0<\cdots<\tilde\sigma_k\in \widetilde X_{\le\tilde\sigma_k}$, and each cell in $X_{\le\sigma_k}$ has a unique lift in $\widetilde X_{\le\tilde\sigma_k}$. This further implies that each cell in $(\sd X)_{\le\rho}$ has a unique lift in $(\sd \widetilde X)_{\le \tilde\rho}$. Hence
   \begin{align}
       (\sd P)|_{(\sd \widetilde X)_{\le \tilde\rho}}:(\sd \widetilde X)_{\le \tilde\rho}\to(\sd X)_{\le\rho}
   \end{align}
   is an isomorphism, which finishes the proof.
\end{proof}

Suppose $X$ is a cell complex and $P:\widetilde X\to X$ is a covering map, then $\widetilde X$ automatically admits a cell structure via $P$ and it induces a chain map
\begin{align}
    P_\#: C_\bullet(\widetilde X,\mathbb F_q)\rightarrow C_\bullet(X,\mathbb F_q),
\end{align}
by setting 
\begin{align}
    P_\#(\tilde \sigma)=\sigma,
\end{align}
for any $\sigma\in X$ and $\tilde \sigma\in P^{-1}(\sigma)$.
Furthermore, we have another chain map called \emph{transfer map}
\begin{align}
    T_\#:C_\bullet (X,\mathbb F_q)\to C_\bullet (\widetilde X,\mathbb F_q),
\end{align}
defined by, for each $\sigma\in X$, 
\begin{align}
    T_\#(\sigma):=\sum_{\tilde\sigma\in P^{-1}(\sigma)}\tilde\sigma.
\end{align}

When it is necessary to emphasize the domain, we write $T_{X,\#}$ for $T_\#$ defined above. It is straightforward to check that when $P$ is an $\ell$-sheeted covering map, then
\begin{align}
    P_\#T_\#=\ell\cdot \text{id}_{C_\bullet(X,\mathbb F_q)}.
\end{align}
Here and throughout, by slight abuse of notation, we also treat $\ell$ as an element in $\mathbb F_q$ obtained by summing $1$ exactly $\ell$ times. Obviously, when $\ell$ is an odd number, then $P_\#$ and $T^\#$ are surjective, $T_\#$ and $P^\#$ are injective. They further induce surjective and injective maps on (co)homology groups, respectively.

Moreover, a direct calculation shows that
\begin{align}
    S_{\widetilde X,\#}T_{X,\#}=T_{\sd X,\#}S_{X,\#},\quad (\sd P)_\#S_{\widetilde X,\#}=S_{X,\#}P_\#.
\end{align}
We call this the compatibility of the subdivision map with covering maps and transfer maps. The following proposition shows that one can construct approximate inverses that enjoy similar compatibility properties.

\begin{proposition}\label{prop:compatible approximate inverse}
Suppose $P:\widetilde X\to X$ is covering map. Let $S_{X,\#}:C_\bullet(X;\mathbb F_q)\to C_\bullet(\operatorname{sd}X;\mathbb F_q)$ be the subdivision map, and let $A_{\sd X,\#}$ be any approximate inverse. Then there exists an approximate inverse
\begin{align}
    A_{\sd \widetilde X,\#}:C_\bullet(\operatorname{sd}\widetilde X;\mathbb F_q)\to C_\bullet(\widetilde X;\mathbb F_q)
\end{align}
of the subdivision map $S_{\widetilde X,\#}:C_\bullet(\widetilde X,\mathbb F_q)\to C_\bullet(\sd \widetilde X,\mathbb F_q),$
such that $A_{\sd \widetilde X,\#}$ is compatible with the covering maps and transfer maps in the sense of
\begin{align}
    P_\#A_{\sd \widetilde X,\#}=A_{\sd X,\#}(\operatorname{sd}P)_\#,\quad A_{\sd \widetilde X,\#}T_{\sd X,\#}=T_{X,\#}A_{\sd X,\#}.
\end{align}
\end{proposition}
\begin{proof}
We divide the proof into the following steps.

\medskip
\noindent
\textbf{Step 1. Construction of $A_{\sd \widetilde X,\#}$.} Let $\widetilde\rho$ be a cell of $\sd \widetilde X$. By our convention,
\begin{align}
    A_{\sd X,\#}((\sd P)(\widetilde\rho))=\sum_{\sigma\in X(\dim\tilde\rho)}A_{(\sd P)(\widetilde\rho),\sigma}\cdot\sigma,
\end{align}
where each $\sigma\le s((\sd P)(\widetilde\rho))=P(s(\widetilde\rho))$ is guaranteed by the carrier condition. Since $P$ is a covering map, the restriction
\begin{align}
    P|_{\widetilde X_{\le s(\tilde\rho)}}:\widetilde X_{\le s(\tilde\rho)}\to X_{\le P(s(\tilde\rho))}
\end{align}
is an isomorphism. As a result, each cell $\sigma\in X_{\le P(s(\tilde\rho))}$ has a unique lift 
$\tilde\sigma\in \widetilde X_{\le s(\tilde\rho)}$ such that $P(\tilde\sigma)=\sigma$. Now we define
\begin{align}\label{eq:def of subdivided approximate inverse}
    A_{\sd \widetilde X,\#}(\tilde\rho)\coloneqq \sum_{\sigma\in X(\dim\tilde\rho)} A_{(\sd P)(\widetilde\rho),\sigma}\cdot\tilde\sigma.
\end{align}

Extending linearly over all cells of $\operatorname{sd}\widetilde X$
defines a linear map
\begin{align}
    A_{\sd \widetilde X,\#}:C_\bullet(\operatorname{sd}\widetilde X,\mathbb F_q)\to C_\bullet(\widetilde X,\mathbb F_q).
\end{align}

\medskip
\noindent
\textbf{Step 2. Verification of the compatibility.} By our construction,
\begin{align}
    P_\#A_{\sd \widetilde X,\#}(\tilde\rho) =\sum_{\sigma\in X_{\le P(s(\tilde\rho))}(\dim\tilde\rho)} A_{(\sd P)(\widetilde\rho),\sigma}\cdot P(\tilde\sigma) 
=
A_{\sd X,\#}((\operatorname{sd}P)(\tilde\rho))=A_{\sd X,\#}(\sd P)_\#(\tilde\rho),
\end{align}
for any cell $\tilde\rho\in\sd \widetilde X$, so $P_\#A_{\sd \widetilde X,\#}=A_{\sd X,\#}(\sd P)_\#$.

Now we consider a cell $\rho\in \sd X$, then by the definition of transfer map,
\begin{align}
    T_{\sd X,\#}(\rho)=\sum_{\tilde\rho\in (\sd P)^{-1}(\rho)}\tilde\rho,
\end{align}
and we have
\begin{align}
    A_{\sd \widetilde X,\#}T_{\sd X,\#}(\rho)
    =\sum_{\tilde\rho\in (\sd P)^{-1}(\rho)}\ \sum_{\sigma\in X(\dim\rho)}A_{\rho,\sigma}\cdot \tilde\sigma_{\tilde\rho}=\sum_{\sigma\in X(\dim\rho)}A_{\rho,\sigma} \sum_{\tilde\rho\in (\sd P)^{-1}(\rho)}\tilde\sigma_{\tilde\rho}
\end{align}
where each $\tilde \sigma_{\tilde\rho}$ is the unique lift of $\sigma$ in $\widetilde X_{\le s(\tilde\rho)}$. On the other hand,
\begin{align}
    T_{X,\#}A_{\sd X,\#}(\rho)
    =\sum_{\sigma\in X(\dim\rho)}A_{\rho,\sigma}\sum_{\tilde\sigma\in P^{-1}(\sigma)}\tilde\sigma.
\end{align}

As $\tilde\rho$ ranges over $(\sd P)^{-1}(\rho)$ for each fixed $\sigma\le s(\rho)$, the lifted cell $\tilde\sigma_{\tilde\rho}$ ranges exactly once over all elements of $P^{-1}(\sigma)$. As a result,
\begin{align}
    A_{\sd \widetilde X,\#}T_{\sd X,\#}(\rho)=T_{X,\#}A_{\sd X,\#}(\rho).
\end{align}
By linearity, $A_{\sd \widetilde X,\#}T_{\sd X,\#}=T_{X,\#}A_{\sd X,\#}$.

\medskip
\noindent
\textbf{Step 3. Verification of the carrier condition.} Since each cell $\tilde\sigma$ in~\eqref{eq:def of subdivided approximate inverse} lies in $\widetilde X_{\le s(\tilde\rho)}$ by definition, we have
\begin{align}
    A_{\sd \widetilde X,\#}(\tilde\rho)\in C_{\dim\tilde\rho}(\widetilde X_{\le s(\tilde\rho)},\mathbb F_q).
\end{align}

\medskip
\noindent
\textbf{Step 4. Verification that $A_{\sd \widetilde X,\#}$ is a chain map.} Since
\begin{align}
    P|_{\widetilde X_{\le s(\tilde\rho)}}:\widetilde X_{\le s(\tilde\rho)}\to X_{\le P(s(\tilde\rho))}
\end{align}
is a homeomorphism of cell posets, there is a natural correspondence among the cellular boundary operators through lifting. Then lifting the equality
\begin{align}
    \partial A_{\sd X,\#}((\operatorname{sd}P)(\tilde\rho))=A_{\sd X,\#}\partial((\operatorname{sd}P)(\tilde\rho))
\end{align}
from $X_{\le P(s(\tilde\rho))}$ to $\widetilde X_{\le s(\tilde\rho)}$ yields 
\begin{align}
    \partial A_{\sd \widetilde X,\#}(\tilde\rho)=A_{\sd \widetilde X,\#}\partial(\tilde\rho).    
\end{align}
This indicates that $A_{\sd \widetilde X,\#}$ is a chain map.

\medskip
\noindent
\textbf{Step 5. Verification of the left inverse condition.} Using the compatibility relation and the fact that the subdivision commutes with covering maps,
\begin{align}
    (\sd P)_\#S_{\widetilde X,\#}=S_{X,\#}P_\#,
\end{align}
we compute
\begin{align}
    P_\#A_{\sd \widetilde X,\#}S_{\widetilde X,\#}=A_{\sd X,\#}S_{X,\#}P_\#=P_\#.
\end{align}

Let $\widetilde\sigma$ be a cell of $\widetilde X$.
Since $S_{\widetilde X,\#}(\widetilde\sigma)$ is supported in $(\sd \widetilde X)_{\le s^{-1}(\tilde\sigma)}$, the carrier condition implies
\begin{align}
    A_{\sd \widetilde X,\#}S_{\widetilde X,\#}(\widetilde\sigma)
\in
C_{\dim\tilde\sigma}(\widetilde X_{\le \tilde\sigma}).
\end{align}
Applying $P_\#$ gives
\begin{align}
    P_\#A_{\sd \widetilde X,\#}S_{\widetilde X,\#}(\tilde\sigma)=P_\#(\tilde\sigma).
\end{align}
Since $P_\#$ restricts to an isomorphism from $C_\bullet(\widetilde X_{\le\tilde\sigma},\mathbb F_q)$ to $C_\bullet(X_{\le P(\tilde\sigma)},\mathbb F_q) $, it follows that
\begin{align}
    A_{\sd \widetilde X,\#}S_{\widetilde X,\#}(\tilde\sigma)=\tilde\sigma,
\end{align}
which finishes the proof.
\end{proof}

It is well known that the projection $P_\#$ and transfer map $T_\#$ are compatible with the cup and cap product when the coefficients are scalars in a field. So it is natural to expect that, when $X$ is equipped with a sheaf $\mathcal F$, they are still compatible with cup and cap product. In other words, if we extend the definition of $P_\#$ and $T_\#$ as follows:
\begin{align}
    P_\#:C_\bullet(\widetilde X,P^*\mathcal F)\to C_\bullet(X,\mathcal F),
\end{align}
where for each $\tilde \sigma\in \widetilde X$, $\tilde x(\tilde\sigma)\in (P^*\mathcal F)_{\tilde\sigma}=\mathcal F_{P(\tilde\sigma)}$,
\begin{align}
    P_\#(\tilde x(\tilde \sigma)\cdot \tilde \sigma):=\tilde x(\tilde\sigma)\cdot P(\tilde\sigma).
\end{align}
And
\begin{align}
    T_\#:C_\bullet(X,\mathcal F)\to C_\bullet(\widetilde X,P^*\mathcal F),
\end{align}
where for each $\sigma\in X$, $x(\sigma)\in \mathcal F_\sigma$,
\begin{align}
    T_\#(x(\sigma)\cdot \sigma)=\sum_{\tilde\sigma\in P^{-1}(\sigma)}x(\sigma)\cdot \tilde\sigma.
\end{align}
Again we get $P_\#T_\#=\ell\cdot \id_{C_\bullet(X,\mathcal F)}.$ It is also routine to verify that one can construct sheaf-valued subdivision maps and approximate inverses that remain compatible with the sheaf-valued covering maps and transfer maps, in a way similar to Proposition~\ref{prop:compatible approximate inverse}. All these favorable properties are demonstrated in the following two propositions.

\begin{proposition}\label{prop:lift of cup product}
    Let $P:\widetilde X\to X$ be a covering map and $\mathcal{F}$ and $\mathcal G$ are two sheaves on $X$. Suppose $\alpha\in C^p(X,\mathcal F)$, $\beta\in C^q(X,\mathcal G)$, then
    \begin{align}
        P^\#(\alpha\smile\beta)=(P^\#\alpha)\smile(P^\#\beta),
    \end{align}
\end{proposition}
\begin{proof}
    First we suppose $X$ and $\widetilde X$ are simplicial complexes. Then for each $\tilde\sigma\in\widetilde X(p+q)$ and $P(\tilde\sigma)=\sigma$, by definition,
    \begin{align}
        \big(P^\#(\alpha\smile \beta)\big)(\tilde\sigma)=(\alpha\smile\beta)(\sigma)=\alpha({}_p\sigma)\otimes\beta(\sigma_q).       
    \end{align}
    On the other hand,
    \begin{align}
        \big((P^\#\alpha)\smile(P^\#\beta)\big)(\tilde \sigma)=\big(P^\#\alpha({}_p\tilde\sigma)\big)\otimes\big(P^\#\beta(\tilde\sigma_q))=\alpha({}_p\sigma)\otimes\beta(\sigma_q).
    \end{align}
    Therefore $P^\#(\alpha\smile\beta)=(P^\#\alpha)\smile(P^\#\beta)$. Now for general cell complexes, we notice that $s\circ \sd P=P\circ s$, hence for any sheaf $\mathcal H$ on $X$, $(\sd P)^*s^*\mathcal H= s^*P^* \mathcal H$ always holds. Then the following commutative diagram finishes the proof.
\begin{equation*}
    \begin{tikzpicture}[>=Stealth, font=\footnotesize]
\node (A) at (0,4.8) {$C^p(\widetilde X,P^*\mathcal F)\times C^{q}(\widetilde X,P^*\mathcal G)$};
\node (B) at (11,4.8) {$C^{p+q}(\widetilde X,P^*(\mathcal F\otimes\mathcal G))$};
\node (C) at (0,0) {$C^p(X,\mathcal F)\times C^{q}(X,\mathcal G)$};
\node (D) at (11,0) {$C^{p+q}(X,\mathcal F\otimes\mathcal G)$};

\node (Ap) at (3.2,3.2) {$C^p(\sd \widetilde X,s^*P^*\mathcal F)\times C^{q}(\sd \widetilde X,s^*P^*\mathcal G)$};
\node (Bp) at (8.0,3.2) {$C^{p+q}(\sd \widetilde X,s^*P^*(\mathcal F\otimes\mathcal G))$};
\node (Cp) at (3.2,1.6) {$C^p(\sd X,s^*\mathcal F)\times C^{q}(\sd X,s^*\mathcal G)$};
\node (Dp) at (8.0,1.6) {$C^{p+q}(\sd X,s^*(\mathcal F\otimes\mathcal G))$};

\draw[->] (A) -- node[midway, above=2pt] {$\smile$} (B);
\draw[->] (C) -- node[midway, left=4pt] {$P^\#\times P^\#$} (A);
\draw[->] (D) -- node[midway, right=4pt] {$P^\#$} (B);
\draw[->] (C) -- node[midway, below=2pt] {$\smile$} (D);

\draw[->] (Ap) -- node[midway, above=2pt] {$\smile$} (Bp);
\draw[->] (Cp) -- node[midway, left=4pt] {$(\sd P)^\#\times (\sd P)^\#$} (Ap);
\draw[->] (Dp) -- node[midway, right=4pt] {$(\sd P)^\#$} (Bp);
\draw[->] (Cp) -- node[midway, below=2pt] {$\smile $} (Dp);

\draw[->] (A) -- node[midway, right=10pt] {$A^\#\times A^\#$} (Ap);
\draw[->] (Bp) -- node[midway, left=10pt] {$S^\#$} (B);
\draw[->] (C) -- node[midway, right=10pt] {$A^\#\times A^\#$} (Cp);
\draw[->] (Dp) -- node[midway, left=10pt] {$S^\#$} (D);

\end{tikzpicture}
\end{equation*}
\end{proof}

\begin{proposition}\label{prop:lift of cap product}
    Let $P:\widetilde X\to X$ be a covering map and $\mathcal{F}$ and $\mathcal G$ are two sheaves on $X$. Suppose $\alpha\in C^p(X,\mathcal{F})$ and $x\in C_{p+q}(X,\mathcal F\otimes\mathcal{G})$, then 
    \begin{align}
        T_\#(\alpha\frown x)=(P^\#\alpha)\frown (T_\# x)
    \end{align}
\end{proposition}
\begin{proof}
    First we suppose $X$ and $\widetilde X$ are simplicial complexes. Then for each $x=x(\sigma)\cdot\sigma$ for $\sigma\in X(p+q)$ and $x(\sigma)\in\mathcal F_\sigma\otimes \mathcal G_\sigma$, by definition
    \begin{align}
    \begin{aligned}
        T_\#(\alpha\frown x(\sigma)\cdot\sigma)&=T_\#\big(\mathcal{G}^T_{\sigma_q,\sigma}\langle \mathcal{F}^T_{{}_p\sigma,\sigma}\alpha({}_p\sigma),x(\sigma)\rangle_{\mathcal{F}}\cdot\sigma_q\big)\\
        &=\sum_{\tilde\sigma\in P^{-1}(\sigma)}\mathcal{G}^T_{\sigma_q,\sigma}\langle \mathcal{F}^T_{{}_p\sigma,\sigma}\alpha({}_p\sigma),x(\sigma)\rangle_{\mathcal{F}}\cdot\tilde\sigma_q.
    \end{aligned}
    \end{align}
    On the other hand,
    \begin{align}
    \begin{aligned}
        P^\#\alpha\frown T_\#(x(\sigma)\cdot\sigma)&=P^\#\alpha\frown \sum_{\tilde\sigma\in P^{-1}(\sigma)}x(\sigma)\cdot \tilde\sigma\\
        &=\sum_{\tilde\sigma\in P^{-1}(\sigma)}(P^*\mathcal{G})^T_{\tilde\sigma_q,\tilde\sigma}\langle(P^*\mathcal{F})^T_{{}_p\tilde\sigma,\tilde\sigma} (P^\#\alpha)({}_p\tilde\sigma),x(\sigma)\rangle_{P^*\mathcal{F}}\cdot\tilde\sigma_q\\
        &=\sum_{\tilde\sigma\in P^{-1}(\sigma)}\mathcal{G}^T_{\sigma_q,\sigma}\langle \mathcal{F}^T_{{}_p\sigma,\sigma}\alpha({}_p\sigma),x(\sigma)\rangle_{\mathcal{F}}\cdot\tilde\sigma_q.
    \end{aligned}
    \end{align}
    Therefore $T_\#(\alpha\frown x)=(P^\#\alpha)\frown (T_\# x)$. Now for general cell complexes, the following commutative diagram finishes the proof.
\begin{equation*}
    \begin{tikzpicture}[>=Stealth, font=\footnotesize]
\node (A) at (0,4.8) {$C^p(\widetilde X,P^*\mathcal F)\times C_{p+q}(\widetilde X,P^*(\mathcal F\otimes\mathcal G))$};
\node (B) at (11,4.8) {$C_q(\widetilde X,P^*\mathcal G)$};
\node (C) at (0,0) {$C^p(X,\mathcal F)\times C_{p+q}(X,\mathcal F\otimes\mathcal G)$};
\node (D) at (11,0) {$C_q(X,\mathcal G)$};

\node (Ap) at (3.2,3.2) {$C^p(\sd \widetilde X,s^*P^*\mathcal F)\times C_{p+q}(\sd \widetilde X,s^*P^*(\mathcal F\otimes\mathcal G))$};
\node (Bp) at (8.0,3.2) {$C_q(\sd \widetilde X,s^*P^*\mathcal G)$};
\node (Cp) at (3.2,1.6) {$C^p(\sd X,s^*\mathcal F)\times C_{p+q}(\sd X,s^*(\mathcal F\otimes\mathcal G))$};
\node (Dp) at (8.0,1.6) {$C_q(\sd X,s^*\mathcal G)$};

\draw[->] (A) -- node[midway, above=2pt] {$\frown$} (B);
\draw[->] (C) -- node[midway, left=4pt] {$P^\#\times T_\#$} (A);
\draw[->] (D) -- node[midway, right=4pt] {$T_\#$} (B);
\draw[->] (C) -- node[midway, below=2pt] {$\frown$} (D);

\draw[->] (Ap) -- node[midway, above=2pt] {$\frown$} (Bp);
\draw[->] (Cp) -- node[midway, left=4pt] {$(\sd P)^\#\times T_\#$} (Ap);
\draw[->] (Dp) -- node[midway, right=4pt] {$T_\#$} (Bp);
\draw[->] (Cp) -- node[midway, below=2pt] {$\frown$} (Dp);

\draw[->] (A) -- node[midway, right=10pt] {$A^\#\times S_\#$} (Ap);
\draw[->] (Bp) -- node[midway, left=10pt] {$A_\#$} (B);
\draw[->] (C) -- node[midway, right=10pt] {$A^\#\times S_\#$} (Cp);
\draw[->] (Dp) -- node[midway, left=10pt] {$A_\#$} (D);

\end{tikzpicture}
\end{equation*}
\end{proof}


\subsection{Geometry of the cell complexes of almost-good qLDPC codes and qLTCs}
Let $G_0=(V_0,E_0)$ be an $n$-regular Cayley expander graph on $n'$ vertices. Based on existing constructions of expanders, e.g., Refs.~\cite{Marcus2015I,Hall_2018,Agarwal2016,Jeronimo2021,Huang2025Ramanujan}, $n$ can be any large integer. 
Let $\mathbb{H}$ be a finite abelian group with $|\HH|=\ell=\exp(O(n'))$ elements, where $\ell$ is assumed to be an odd integer. Let $\gamma: E(G_0) \rightarrow \HH$ be a map assigning each edge $(u,v)$ a group element $s \in \HH$. Then we have $t$ permutation sets $A_1,\cdots,A_t$ on $\HH\times V_0^t\times\{0,1\}^t$, where each $A_j=\{a_j^1,\cdots a_j^n\}$ and
\begin{align}\label{eq:generator_action}
		a_j^{\mu} \cdot (h,v_1,...,v_j,...,v_t; b_1,...,b_j,...,b_t) 
		= (\gamma_{(v_j, v_j')} \cdot h,v_1,...,v_j',...,v_t; b_1,...,1-b_j,...,b_t),
\end{align} 
where $(v_j',1-b_j)$ is the $\mu$-th neighbor of $(v_j,b_j)$ in the double cover of $G_0$. 

Now we define a $t$-dimensional complex $X^t_{\HH}$ as follows:

\medskip
\noindent
\textbf{Cells.} For $0 \leq k \leq t$ the set of $k$-dimensional cells of $X^t_{\HH}$ is defined by
\begin{align}
X_{\HH}^t(k)\coloneqq \bigg \{[g;\,(a_j)_{j\in S},\,(b_j)_{j\notin S}] : g\in \HH\times V_0^t,\, S\subseteq [t],\, |S|=k,\,
a\in\prod_{j\in S}A_j,\,
b\in\prod_{j\notin S}\{0,1\}\bigg\}.
\end{align}
For a cell $\sigma=[g;\,(a_j)_{j\in S},\,(b_j)_{j\notin S}]$, we call $S=\operatorname{type}(\sigma)$ the \emph{type} of $\sigma$, and write
\begin{align*}
\sigma_0=g
\qquad\text{and}\qquad
\sigma_j=
\begin{cases}
a_j & j\in S,\\
b_j & j\notin S,
\end{cases}
\end{align*}
for $j=1,\ldots,t$.

\medskip
\noindent
\textbf{Partial order.} Given two cells $\sigma,\sigma'\in X^t_{\HH}$ and $j\in \{1,\ldots,t\}$ we write $\sigma'\lessdot_j \sigma$ if all the following hold:
\begin{itemize}
\item $\sigma'_j\in \{0,1\}$ while $\sigma_j\in A_j$,
\item $\sigma_i=\sigma'_i$ for all $i\notin \{0,j\}$,
\item $\sigma'_0=\begin{cases}
\sigma_0 & \text{if } \sigma'_j=0,\\
\sigma_0\cdot \sigma_j & \text{if } \sigma'_j=1.
\end{cases}$
\end{itemize}
For two general cells $\sigma'$ and $\sigma$, we write $\sigma'\le \sigma$ if $\sigma'=\sigma$ or if there is a chain from $\sigma'$ to $\sigma$ generated by the relations $\lessdot_j$.

If we replace the group $\HH$ by the trivial group $\{e\}$, then the above construction
produces a $t$-dimensional cell complex $X^t_{\{e\}}$. The trivial group homomorphism
$\HH\to \{e\}$ induces a projection
\begin{align}
    P:X^t_{\HH}\to X^t_{\{e\}},
\end{align}
which sends each cell
\[
[(h,v_1,\dots,v_t);\,(a_j)_{j\in S},\,(b_j)_{j\notin S}]\in X^t_{\HH}
\]
to
\[
[(e,v_1,\dots,v_t);\,(a_j)_{j\in S},\,(b_j)_{j\notin S}]\in X^t_{\{e\}}.
\]
This projection $P$ is an $\ell$-sheeted covering map, and $X^t_{\{e\}}$ is naturally identified with the
$t$-fold Cartesian product of the double cover of the graph $G_0$, namely the $t$-fold hypergraph product of the double cover of $G_0$.

\subsection{Proof details}\label{sec:proof}

By the convention above, $X_{\{e\}}^1$ is exactly the double cover of $G_0$. Given any local parity check matrix $h:\mathbb F_q^{A}\to \mathbb F_q^{\hat A}$, where $|\hat A|$ is a constant, the associated Tanner code can be realized as follows. Let $\mathcal F^{h}$ be a sheaf on $X^1_{\{e\}}$ defined as follows. For each vertex $[g;b]\in X^1_{\{e\}}(0)$,
\begin{align}
    \mathcal F^{h}_{[g;b]}\coloneqq \mathbb F_q^{\hat A},
\end{align}
for each edge $[g';a]\in X^1_{\{e\}}(1)$,
\begin{align}
    \mathcal F^{h}_{[g';a]}:=\mathbb F_q,
\end{align}
and whenever $[g;b]\lessdot[g';a]$, the restriction map $\mathcal F^{h}_{[g;b],[g';a]}:\mathcal F^{h}_{[g;b]} \to \mathcal F^{h}_{[g';a]}$ is defined by setting, for any vector $x\in \mathcal F^{h}_{[g;b]}$,
\begin{align}
    \mathcal F^{h}_{[g;b],[g';a]}(x) \coloneqq  \langle h\mathbbm 1_{\{a\}},x\rangle,
\end{align}
where for each $a\in A$, $\mathbbm 1_{\{a\}}$ is the corresponding standard basis vector in $\mathbb F_q^{A}$. The transpose of this restriction map satisfies
\begin{align}
    (\mathcal{F}^{h}_{[g;b],[g';a]})^T(1) = h\mathbbm 1_{\{a\}},
\end{align}
where $1 \in \mathcal F^{h}_{[g';a]} = \mathbb{F}_q$.

Now consider the complex $X^t_{\{e\}}$ with $t$ local codes $h_1,\dots, h_t$, where for each $i\in[t]$, $h_i:\mathbb F_q^{A_i}\to \mathbb F_q^{\hat A_i}$, $|\hat A_i|$ is a constant and $|A_i|=n$, with $n$ being the degree of the Cayley graph $G_0$. Notice that $X^t_{\{e\}}$ is the $t$-fold HGP of $X^1_{\{e\}}$, and for each copy of $X^1_{\{e\}}$ there is a sheaf $\mathcal{F}^{h_i}$. We equip $X^t_{\{e\}}$ with the external tensor product sheaf $\boxtimes_{i=1}^t\mathcal F^{h_i}$. Then by the K\"unneth formula,
\begin{align}\label{eq:chain complex of HGP with sheaf}
    C^\bullet(X^t_{\{e\}},\boxtimes_{i=1}^t\mathcal F^{h_i})=\bigotimes_{i=1}^tC^\bullet(X^1_{\{e\}},\mathcal F^{h_i}),
\end{align}
and
\begin{align}
    H^{p}(X^t_{\{e\}},\boxtimes_{i=1}^t\mathcal F^{h_i})\cong\bigoplus_{p_1+\cdots+p_t=p}\bigotimes_{i=1}^tH^{p_i}(X^1_{\{e\}},\mathcal F^{h_i}).
\end{align}
Therefore,
\begin{align}
    H_0(X^t_{\{e\}},\boxtimes_{i=1}^t\mathcal F^{h_i})\cong \bigotimes_{i=1}^tH_{0}(X^1_{\{e\}},\mathcal F^{h_i}),
\end{align}

The covering map $P:X^t_{\HH}\to X^t_{\{e\}}$ induces a pullback sheaf $P^*(\boxtimes_{i=1}^t\mathcal F^{h_i})$ on $X^t_{\HH}$. By construction, $X^t_{\HH}$ together with the pullback sheaf
$P^*(\boxtimes_{i=1}^t\mathcal F^{h_i})$
agrees with the complex and sheaf used in Ref.~\cite{Dinur2024sheaf}
to construct almost-good qLDPC codes and qLTCs. Since $\ell$ is an odd number, $P^\#$ and the transfer map $T_\#$ are injective. By Propositions~\ref{prop:lift of cup product} and~\ref{prop:lift of cap product}, nontrivial cup products on $X^t_{\{e\}}$ induce nontrivial cup products on $X^t_{\HH}$. This further implies the nontriviality of the multi-controlled-$Z$ gates on these almost-good codes.

\medskip
We elaborate on the 3-dimensional case for almost-good qLDPC codes, as a warm-up. The arguments can be easily extended to higher dimensions for qLTCs. We prepare three collections of local parity-check matrices $\{h_1^{(j)},h_2^{(j)},h_3^{(j)}\}, j\in[3]$, where each local code is specified by a matrix $h_i^{(j)}: \mathbb{F}_q^{A_i} \to \mathbb{F}_q^{\hat{A}_i^{(j)}}$. Suppose that there exist $a_1\in A_1$, $a_2\in A_2$ and $a_3\in A_3$ satisfying 
\begin{align}\label{eq: condition for nontirvial cup}
    h_i^{(j)}\mathbbm 1_{\{a_i\}}\ne 0, i,j\in [3], i\ne j
\end{align}
simultaneously, then we can find $z_i^{(j)}\in \mathbb F_q^{\hat{A}_i^{(j)}}$, such that $\langle h_i^{(j)}\mathbbm 1_{\{a_i\}},z_i^{(j)}\rangle= 1$, where $i,j\in [3]$, and define $\zeta_i^{(j)}\in C^0(X^1_{\{e\}},\mathcal F^{h_i^{(j)}})$ by, for each vertex $v\in X^1_{\{e\}}(0)$,
\begin{align}\label{eq:zeta}
    \zeta_i^{(j)}(v)=z_i^{(j)}.
\end{align}
Each $\zeta_i^{(j)}$ is then a cocycle in $C^0(X^1_{\{e\}},\mathcal F^{h_i^{(j)}})$. Therefore, if we choose three arbitrary vertices $v_{1},v_{2},v_{3}\in V_0$, then by Eq.~\eqref{eq:chain complex of HGP with sheaf},
\begin{align}
\begin{aligned}
    \alpha^{(1)}\coloneqq[v_{1};a_1]\otimes \zeta_2^{(1)}\otimes \zeta_3^{(1)}\in C^1(X^3_{\{e\}},\boxtimes_{i=1}^3\mathcal F^{h_i^{(1)}}),\\
    \alpha^{(2)}\coloneqq \zeta_1^{(2)}\otimes[v_{2};a_2]\otimes\zeta_3^{(2)}\in C^1(X^3_{\{e\}},\boxtimes_{i=1}^3\mathcal F^{h_i^{(2)}}),\\
    \alpha^{(3)}\coloneqq \zeta_1^{(3)}\otimes\zeta_2^{(3)}\otimes[v_{3};a_3]\in C^1(X^3_{\{e\}},\boxtimes_{i=1}^3\mathcal F^{h_i^{(3)}}),
\end{aligned}
\end{align}
are cocycles. By the definition of cup product, we have
\begin{align}
    &\alpha^{(1)}\smile \alpha^{(2)}\smile\alpha^{(3)}=\prod_{i,j\in [3], i\ne j} \langle h_i^{(j)}\mathbbm 1_{\{a_i\}},z_i^{(j)}\rangle \cdot [(v_{1},v_{2},v_{3}); a_{1},a_{2},a_{3}],
\end{align}
which is supported on the single cube $[(v_{1},v_{2},v_{3}); a_{1},a_{2},a_{3}]$. To find a cycle in $C_3(X^3_{\{e\}},\boxtimes_{i=1}^3\mathcal F^{h_i})$, we have the following observation: let $E=\{\beta_1,\ldots,\beta_n\}\subseteq \mathbb{F}_q$
be a set of distinct evaluation points, and let
$\C = \RS_{q}(E,m)\subseteq\mathbb{F}_q^n$ be the punctured RS code. We choose the standard Vandermonde matrix $h$ whose rows generate $\C$, namely
\begin{align}
    h=
    \begin{pmatrix}
        1 & 1 & \cdots & 1\\
        \beta_1 & \beta_2 & \cdots & \beta_n\\
        \beta_1^2 & \beta_2^2 & \cdots & \beta_n^2\\
        \vdots & \vdots & & \vdots\\
        \beta_1^{m-1} & \beta_2^{m-1} & \cdots & \beta_n^{m-1}
    \end{pmatrix}.
\end{align}
Thus $\im h^T=\C$, and the $a$-th column is $h\mathbbm 1_{\{a\}}=(1,\beta_a,\beta_a^2,\dots,\beta_a^{m-1})^T$.

Suppose $3m-3\le n-2$. We consider the 3-fold column-wise tensor product, or Khatri--Rao product, $h*h*h$ in Definition \ref{def:KR_product}. We index the rows of $h*h*h$ by triples $(p,q,r)$ with $0\le p,q,r\le m-1$. Then the entry in row $(p,q,r)$ and column $a$ is
\begin{align}
    (h*h*h)_{(p,q,r),a}=(h\mathbbm 1_{\{a\}}\otimes h\mathbbm 1_{\{a\}}\otimes h\mathbbm 1_{\{a\}})_{(p,q,r)}=\beta_a^{p+q+r}.
\end{align}
For each $a\in[n]$, the multiplier is 
\begin{align}
    \lambda_a = \frac{1}{\prod_{b\neq a}(\beta_a-\beta_b)}.
\end{align}
Then for each $0\le d\le n-2$, a direct calculation verifies $\sum_{a=1}^n \lambda_a \beta_a^d=0$. Hence, if $3m-3\le n-2$, the vector $\lambda=(\lambda_a)_{a\in[n]}\in \ker (h*h*h)$.

Suppose $2m-1<n$, we set $\im (h')^T \coloneq (\C * \C)^\perp$. Then a similar calculation shows that the all-ones vector $\mathbbm 1_{[n]}\in \ker( h*h*h')\subseteq \mathbb F_q^n$.

Given this observation, we are now able to provide a specific choice of local codes such that \eqref{eq: condition for nontirvial cup} holds and give a nontrivial cycle in $C_3(X_{\{e\}}^3,\otimes_{j=1}^3(\boxtimes_{i=1}^3\mathcal F^{h_i^{(j)}}))$. Let $\C_i=\RS_q(E_i,|\hat A_i^{(i)}|)\subseteq\mathbb F_q^n$, $i\in[3]$ be low rate two-way product-expanding punctured RS codes satisfying Corollary~\ref{coro:RS_dual} and \ref{coro:GRS_expanding}, then $(\C_1*\C_1)^\perp, (\C_2*\C_2)^\perp, \C_3$ is also two-way product-expanding. 
 
We choose the local parity check matrices $h_i^{(j)}$ according to the following table:
\begin{align}
\begin{array}{c|ccc}
\im (h_i^{(j)})^T & j=1 & j=2 & j=3\\
\hline
i=1 & \C_1 & \C_1 & (\C_1 * \C_1)^\perp\\
i=2 & \C_2 & \C_2 & (\C_2 * \C_2)^\perp\\
i=3 & \C_3 & \C_3 & \C_3
\end{array}
\end{align}

To satisfy the requirement in \eqref{eq: condition for nontirvial cup}, we choose $a_1,a_2,a_3$ such that $h_i^{(1)}\mathbbm 1_{\{a_i\}}\ne 0$, $i\in[3]$. Then this automatically gives a valid choice since for $i\in[2]$, $\im (h_i^{(j)})^T= (\C_i*\C_i)^\perp$, then $h_i^{(j)}\mathbbm 1_{\{a_i\}}=0$ if and only if $\mathbbm 1_{\{a_i\}}\in \C_i * \C_i=\RS_q(E_i,2|\hat{A}_i^{(i)}|-1)$, which is impossible since a polynomial of degree at most $2|A_i^{(i)}|-2$ cannot have $n-1$ roots. 

Notice that $X^3_{\{e\}}$ is the HGP of three copies of $X^1_{\{e\}}$, where the $i$-th copy can be identified with the double cover of the Cayley graph $G_0$ with generating set $A_i$. Since $|A_i|=n=|E_i|$, we can associate each evaluation point with an element in $A_i$. Therefore we may write $E_{i}=\{\beta_a\}_{a\in A_i}$. In this sense, we note that by the above observation, we can write down three cycles $x_i\in C_1(X^1_{\{e\}},\otimes_{j=1}^3\mathcal F^{h_i^{(j)}})$, $i\in[3]$ respectively as follows:
\begin{align}
    \begin{aligned}
        &x_1=\mathbbm 1_{X^1_{\{e\}}}\in \mathbb F_q^{X^1_{\{e\}}}\cong C_1(X^1_{\{e\}},\otimes_{j=1}^3\mathcal F^{h_1^{(j)}}),\\
        &x_2=\mathbbm 1_{X^1_{\{e\}}}\in \mathbb F_q^{X^1_{\{e\}}}\cong C_1(X^1_{\{e\}},\otimes_{j=1}^3\mathcal F^{h_2^{(j)}}),\\
        &x_3=\sum_{v\in V_0,a\in A_3}\frac{1}{\prod_{a'\in A_3, a'\ne a}(\beta_a-\beta_{a'})}\cdot [v;a].
    \end{aligned}
\end{align}
By the relationship between tensor product and external tensor product of sheaves, we have
\begin{align}
    C_\bullet(X^3_{\{e\}},\otimes_{j=1}^3(\boxtimes_{i=1}^3\mathcal F^{h_i^{(j)}}))\cong C_\bullet(X^3_{\{e\}},\boxtimes_{i=1}^3(\otimes_{j=1}^3\mathcal F^{h_i^{(j)}}))\cong\bigotimes_{i=1}^3 C_\bullet(X_{\{e\}}^1,\otimes_{j=1}^3\mathcal F^{h_i^{(j)}}).
\end{align}
Therefore, $x_1\otimes x_2\otimes x_3$ is a cycle in $C_3(X^3_{\{e\}},\otimes_{j=1}^3(\boxtimes_{i=1}^3\mathcal F^{h_i^{(j)}}))$. A direct calculation shows that 
\begin{align}
    \left\langle \alpha^{(1)}\smile \alpha^{(2)}\smile \alpha^{(3)}, \ x_1\otimes x_2\otimes x_3 \right\rangle = \frac{1}{\prod_{a_3'\in A_3, a_3'\ne a_3}(\beta_{a_3}-\beta_{a_3'})}\ne 0.
\end{align}

This computation lifts to the covering space $X^3_{\HH}$, i.e., there exist nontrivial (co)homology elements
\begin{align}
    \begin{aligned}
        & P^*[\alpha^{(j)}]\in H^1(X^3_{\HH},P^*(\boxtimes_{i=1}^3\mathcal F^{h_i^{(j)}})), \quad j\in [3], \\
        & T_{*}[x_1\otimes x_2\otimes x_3]\in  H_3(X^3_{\HH},\otimes_{j=1}^3(\boxtimes_{i=1}^3\mathcal F^{h_i^{(j)}})),
    \end{aligned}
\end{align}
such that
\begin{align}
    \left\langle P^*[\alpha^{(1)}]\smile P^*[\alpha^{(2)}]\smile P^*[\alpha^{(3)}], \ T_*[x_1\otimes x_2\otimes x_3] \right\rangle
    =\frac{\ell}{\prod_{a_3'\in A_3, a_3'\ne a_3}(\beta_{a_3'}-\beta_{a_3})}\ne 0
\end{align}

Let $\xi\coloneq T_\#(x_1\otimes x_2\otimes x_3)$. Then $\xi$ induce the following cohomological invariant form 
\begin{align}
    I_{\xi}:\prod_{j\in[3]} C_1(X^3_{\HH},P^*(\boxtimes_{i=1}^3\mathcal F^{h_i^{(j)}}))\to\mathbb F_q,
\end{align}
defined by
\begin{align}
    I_{\xi}(\alpha_1,\alpha_2,\alpha_3)\coloneq \langle \alpha_1\smile\alpha_2\smile\alpha_3, \ \xi \rangle,
\end{align}
where $\alpha_j\in C_1(X^3_{\HH},P^*(\boxtimes_{i=1}^3\mathcal F^{h_i^{(j)}})),j\in[3]$. This cohomological invariant form $I_\xi$ induces a nontrivial constant-depth logical $\CCZ$ gate among the three corresponding CSS codes.

By the rate requirement of local codes in Ref.~\cite{Dinur2024sheaf}, since we choose $\C_1, \C_2, \C_3$ to have small rate, $(\C_1*\C_1)^\perp$ and $(\C_2*\C_2)^\perp$ have large rate. Then the CSS code built from $C_1(X^3_{\HH},P^*(\boxtimes_{i=1}^3\mathcal F^{h_i^{(3)}}))$ is a qLDPC code with parameters
\begin{align}
    [\![N,\Theta(N),\Theta(N/(\operatorname{log}N)^2)]\!].
\end{align}
On the other hand, the CSS codes built from $C_1(X^3_{\HH},P^*(\boxtimes_{i=1}^3\mathcal F^{h_i^{(j)}}))$, $j\in[2]$ admit almost-good distance, but the number of logical qubits is only lower bounded by a constant, which gives a qLDPC code with parameters
\begin{align}
    [\![N,\Omega(1),\Theta(N/(\operatorname{log}N)^2)]\!].
\end{align}
Fortunately, in implementing the $\CCZ$ gate, we prepare three code blocks and only the last one may have merely constantly many logical qubits, so the rate of the combined blocks is still constant. 

Note that any constant-depth logical $\CCZ$ gate can be made transversal, at the cost of only a constant-factor loss in code parameters, as stated in the following lemma.

\begin{lemma}[\cite{Nguyen2025CCZ, Golowich_Lin2024}]\label{lem:constant-depth-to-transversal}
    Suppose the quantum codes $\mathcal Q_1,\ldots,\mathcal Q_r$ support a constant-depth $\mCZ$ circuit. Then we can construct new codes $\mathcal Q'_1,\ldots,\mathcal Q'_r$ that support a transversal $\mCZ$ circuit inducing the same number of logical $\mCZ$ gates. Furthermore, the parameters of the new codes (sparsity, qubits, rate, relative distance, soundness) worsen by at most a constant factor.
\end{lemma}

Therefore, we obtain the following corollary.

\begin{corollary}
    There exist $[\![N,\Theta(N),\Theta(N/(\operatorname{log}N)^2)]\!]$ qLDPC codes that support nontrivial transversal logical $\CCZ$ gates.
\end{corollary}

The $\CCZ$ cases above should be viewed as illustrative examples, which we now generalize to multi-controlled-$Z$ gates. For general $t\ge 1$, we may choose $t$ punctured RS codes $\C_1,\ldots,\C_t$ satisfying Corollary~\ref{coro:GRS_expanding}. Then, we choose parity-check matrices $\{h_i^{(j)}\}_{i,j\in[t]}$ such that $\im (h_i^{(j)})^T=C_i$ when $j\le t-1$, $\im (h_i^{(t)})^T = (\C_i^{*(t-1)})^\perp$ when $i\le t-1$ and $\im (h_t^{(t)})^T=C_t$. Let the evaluation set of $\C_t$ be $\{\beta_{a_t}\}_{a_t\in A_t}$, then
\begin{align}
    \xi_t = T_\#\big(\bigotimes_{j=1}^{t-1}(\sum_{v_j\in V_0,a_j\in A_j}[v_j;a_j])\otimes(\sum_{v_t\in V_0,a_t\in A_t} \prod_{a_t'\in A_t,a_t'\ne a_t}(\beta_{a_t}-\beta_{a_t'})^{-1}\cdot[v_t;a_t]) \big),
\end{align}
induce a cohomological invariant form
\begin{align}
    I_{\xi_t}:\prod_{j\in[t]} C^1(X^t_{\HH},P^*(\boxtimes_{i=1}^t\mathcal F^{h_i^{(j)}}))\to\mathbb F_q,
\end{align}
defined by
\begin{align}
    I_{\xi_t}(\alpha_1,\cdots,\alpha_t)\coloneq \langle \alpha_1\smile\dots\smile\alpha_t, \ \xi_t\rangle,
\end{align}
where $\alpha_j\in C^1(X^t_{\HH},P^*(\boxtimes_{i=1}^t\mathcal F^{h_i^{(j)}})),j\in[t]$. This cohomological invariant form $I_{\xi_t}$ induces a nontrivial transversal logical $\mathrm{C}^{t-1}Z$ gate on qLDPC codes with parameters $[\![N,\Theta(N),\Theta(N/(\operatorname{log}N)^{t-1})]\!]$.

If $t$ is an even number, then we can choose the parity-check matrices such that $\im (h_i^{(j)})^T = \C_i$ when $j\le t-1$, $\im (h_i^{(t)})^T = (\C_i^{*(t-1)})^\perp$ when $i \le t-2$ and $\im (h_i^{(t)})^T = \C_i$ when $i=t-1$ or $t$. Let the evaluation sets of $\C_{t-1}$ and $\C_t$ to be $\{\beta_{a_{t-1}}\}_{a_{t-1}\in A_{t-1}}$ and $\{\beta_{a_{t}}\}_{a_{t}\in A_{t}}$ respectively, then
\begin{align}
    \xi'_t=T_\#\big(\bigotimes_{j=1}^{t-2}(\sum_{v_j\in V_0,a_j\in A_j}[v_j;a_j])\otimes\bigotimes_{j=t-1}^{t}(\sum_{v_j\in V_0,a_j\in A_j} \prod_{a_j'\in A_j,a_j'\ne a_j}(\beta_{a_j}-\beta_{a_j'})^{-1}\cdot[v_j;a_j]) \big),
\end{align}
induce a cohomological invariant form
\begin{align}
    I_{\xi'_t}:\prod_{j\in[t/2]} C^2(X^t_{\HH},P^*(\boxtimes_{i=1}^t\mathcal F^{h_i^{(j)}}))\to\mathbb F_q,
\end{align}
defined by
\begin{align}
    I_{\xi'_t}(\alpha_1,\cdots,\alpha_{t/2})\coloneq \langle \alpha_1\smile\dots\smile\alpha_{t/2}, \ \xi'_t\rangle,
\end{align}
where $\alpha_j\in C^2(X^t_{\HH},P^*(\boxtimes_{i=1}^t\mathcal F^{h_i^{(j)}})),j\in[t/2]$. This cohomological invariant form $I_{\xi'_t}$ induces a nontrivial transversal logical $\mathrm{C}^{t/2-1}Z$ gate on qLDPC codes with parameters $[\![N,\Theta(N),\Theta(N/(\operatorname{log}N)^{t-1})]\!]$ and soundness $\Theta(1/(\operatorname{log}N)^{t-1})$. This completes the proof of Theorem~\ref{thm:main}.

\section{Discussion and outlook}

Although the transversal non-Clifford gates on almost-good qLDPC codes and qLTCs are presented as main results, we highlight that they emerge as an  implication of highly general algebraic reasoning. Beyond the  specific cases, our derivation points to a broader principle: given a quantum code defined from a cell complex, with or without sheaf structures, fault-tolerant multi-controlled-$Z$ gates can be systematically constructed by cohomological invariant forms in a  flexible manner. 
Looking ahead, should good qLTCs eventually be constructed from cell complexes and sheaves as one may expect, our framework should apply with little additional effort. This leads to the expectation that such codes would likewise admit transversal logical multi‑controlled‑$Z$ gates. 

An immediate next goal is to consider more intricate features of the logical operations, such as addressability and parallelizability, for which we expect that significantly refined analysis and ideas will be needed. Here, we have only rigorously proved the occurrence of nontrivial logical $\mCZ$ action, that is, the subrank $k_{\mCZ}$ of the invariant form, which characterizes the number of parallelizable logical gates, does not vanish. Establishing better lower bounds on $k_{\mCZ}$ remains open. One apparent limitation of our current construction is that one of the code blocks involved is only known to possess a constant number of logical qubits. Progress toward sharper estimates will likely first require a more refined strategy for evaluating the rates of sheaf codes in Ref.~\cite{Dinur2024sheaf} which facilitates a broader choice of local codes in order to overcome the constant-logical-qubit bottleneck.
Recent work \cite{KoppartyTamo2026} realizes Tanner codes with constant rate by local codes of rate $\leq \frac{1}{2}$. This breaks the low-local-rate barrier for achieving a nontrivial global rate, but the construction is non-LDPC because the graph node degree and local code length grow sublinearly as a fractional power. Nevertheless, it remains important to study the low-local-rate regime and extend these results to high-dimensional sheaf codes.
Subsequently, a deeper understanding of specific logical representatives in the code space would be important, as they make the expression of cohomological invariants, such as cup products, fully explicit~\cite{LSWLL2026Theory}.

Furthermore, our constructions open the possibility of further improving the asymptotic spacetime overhead of quantum fault tolerance, which we leave for separate work.

\section*{Acknowledgments} 
{We thank an anonymous referee for making us aware of a gap in v1 of this paper, which led us to the corrected argument presented here.}
This work is supported in part by NSFC under Grant No.~12475023, Dushi Program, and a startup funding from YMSC.


\printbibliography[heading=bibintoc,title=References]

\end{document}